\documentclass[prx,floatfix,superscriptaddress,nofootinbib,notitlepage]{revtex4-2}
\pdfoutput=1
\usepackage{graphicx,xcolor}
\usepackage{amsthm,amsmath,bm,dsfont,amssymb,mathtools}
\usepackage{physics}
\usepackage[normalem]{ulem}
\usepackage[colorlinks]{hyperref}
\hypersetup{
	pdfstartview={FitH},
	pdfnewwindow=true,
	colorlinks=true,
	linkcolor=blue,
	citecolor=blue,
	filecolor=blue,
	urlcolor=blue}
\usepackage[capitalise]{cleveref} 
\usepackage{booktabs,longtable,ltablex} 
\keepXColumns
\usepackage{qcircuit}

\numberwithin{equation}{section}

\numberwithin{equation}{section}
\theoremstyle{plain}
\newtheorem{theorem}{Theorem}[section]
\newtheorem{lemma}[theorem]{Lemma}
\newtheorem{proposition}[theorem]{Proposition}
\newtheorem{corollary}[theorem]{Corollary}
\newtheorem{prob}[theorem]{Problem}
\newtheorem{exmp}[theorem]{Example}
\theoremstyle{definition}
\newtheorem{remark}[theorem]{Remark}
\newtheorem{definition}[theorem]{Definition}

\newcommand{\MQ}{\affiliation{%
School of Mathematical and Physical Sciences, Macquarie University, NSW  2109, Australia} }

\newcommand{\USYDPHYS}{\affiliation{% 
School of Physics,
The University of Sydney, NSW 2006, AU}}

\newcommand{\USYDMATH}{\affiliation{% 
School of Mathematics and Statistics,
The University of Sydney, NSW 2006, AU}}

\newcommand{\concordia}{\affiliation{%
Department of Physics, Concordia University, Montreal, QC H4B 1R6, Canada
}}

\newcommand{\oxford}{\affiliation{%
Department of Physics, The University of Oxford, OX2 6GG, United Kingdom
}}

\begin{abstract}
Quantum algorithms offer an exponential advantage with respect to the number of dependent variables
for solving certain nonlinear ordinary differential equations (ODEs). 
These algorithms typically begin by transforming 
the original nonlinear ODE into a higher-dimensional linear ODE using a linearization technique, 
most commonly Carleman linearization. 
Existing works restrict their analysis to ODEs where the nonlinearities 
are polynomial functions of the dependent variables, 
significantly limiting their applicability.
In this work we construct an efficient quantum algorithm for solving
ODEs with `Fourier' nonlinear terms 
expressible as $d{\bf u}/dt = G_0 + G_1 e^{i{\bf u}}$,
where ${\bf u}$ denotes a vector of $n$ complex variables evolving with $t$,
$G_0$ is an $n$-dimensional complex vector, $G_1$ is an $n \times n$ 
complex matrix and $e^{i{\bf u}}$ denotes the vector with entries $\{e^{iu_j}\}$.
To tackle the Fourier nonlinear term, which is not expressible as a finite sum of polynomials
of ${\bf u}$, our algorithm employs a generalization of the 
Carleman linearization technique known as Koopman linearization.
We also make other methodological advances towards relaxing 
the stringent dissipativity condition required for efficient solution extraction
and towards integrated readout of classical quantities from the solution state.
Our results open avenues to the development of efficient
quantum algorithms for a significantly wider class of 
high-dimensional nonlinear ODEs, thereby broadening the scope of
their applications.

\end{abstract}

\begin{document}
\title{Efficient quantum algorithm for solving differential equations\\ 
with Fourier nonlinearity via Koopman linearization}

\author{Judd Katz}\email{judd.katz@physics.ox.ac.uk}\USYDPHYS\USYDMATH\oxford
\author{Gopikrishnan Muraleedharan}\MQ

\author{Abhijeet Alase}\USYDPHYS\concordia

\maketitle

\tableofcontents

\section{Introduction}
\subsection{Motivation and context}
Differential equations are fundamental to modeling dynamical systems across science,
technology and engineering. Quantum algorithms offer the promise of exponential speedups in terms of the number of dependent variables over classical methods for solving such equations. Building on the Harrow-Hassidim-Lloyd (HHL) algorithm for solving linear systems~\cite{harrow2009quantum} and its adaptation to ordinary differential equations (ODEs)~\cite{berry2014high}, the past decade has seen rapid progress in quantum algorithms for a broad class of differential equations, including linear and nonlinear, 
ordinary and partial differential equations~\cite{berry2014high, costa2025further, liu2021efficient, krovi2023improved, lloyd2020quantum, novikau2025quantum, berry2024quantum, engel2021linear}. This growing body of work aims to extend the range of differential equations that are efficiently solvable on quantum computers \cite{costa2025further, liu2021efficient, jennings2025quantum}, 
optimize algorithmic performance and resource requirements \cite{krovi2023improved, costa2025further}, 
explore the complexity-theoretic boundaries of quantum computation in this domain \cite{lewis2024limitations},
and assess the feasibility of practical applications \cite{tennie2025quantum}.

Many of the applications of differential equations rely on solving
an initial value problem for a high-dimensional 
system of nonlinear  
ODEs~\cite{gaitan2020finding, sanavio2024three,vaszary2024solving, dodin2021applications,jin2024quantum, jin2022quantum}.
Quantum computers are naturally better equipped to solve linear ODEs
owing to the linearity of Schrodinger's equation. Consequently,
to solve systems of nonlinear ODEs on a quantum computer, the standard
approaches take advantage of `linearization' techniques, which
map the given system of nonlinear ODEs to a larger
system of linear ODEs.
In principle, several techniques could be used to linearize a system of nonlinear ODEs, including 
Carleman linearization~\cite{brustle2024quantum, wu2024quantum, krovi2023improved, costa2025further, liu2021efficient, surana2024efficient},
Koopman-von Neumann linearization~\cite{jin2023time, joseph2020koopman, barthe2023continuous, tanaka2024polynomialtimequantumalgorithm, novikau2025quantum}, 
pivot switching method~\cite{endo2024divergence},
conversion to Liouville equation~\cite{lin2022koopman,jin2023time},
conversion to Fokker-Planck approach~\cite{tennie2024solving}
homotopy pertubation approach~\cite{xue2021quantum} 
and reproducing kernel Hilbert spaces approach \cite{giannakis2022embedding}. 
All of these approaches map the original system of nonlinear 
ODEs to either a partial differential equation
or an infinite-dimensional system of linear ODEs, necessitating 
a discretization or a truncation step in the design of the quantum algorithm.
Carleman linearization-based algorithms 
typically enforce a dissipativity condition
for bounding the error due to such a truncation. 
This dissipativity condition rules out
chaotic or turbulent behaviour of the solution. 
These results tie in well with the complexity theoretic bounds that exclude the possibility
of a linear in time and logarithmic in system-size quantum algorithm 
for solving non-dissipative ODEs \cite{liu2021efficient, lewis2024limitations}. Despite their success, quantum algorithms 
based on Carleman linearization apply only if the nonlinear terms 
in the ODEs are polynomials of the dependent variables, 
therefore limiting their applicability. 
This work expands the class of
nonlinear ODEs that can be solved efficiently on a quantum computer.

In this work, we construct an efficient quantum algorithm for solving
a system of ODEs with Fourier nonlinearity, 
expressed as $d{\bf u}/dt = G_0 + G_1 e^{i{\bf u}}$,
where ${\bf u}$ denotes a vector of $n$ complex variables evolving with $t$,
$G_0$ is an $n$-dimensional complex vector, $G_1$ is an $n \times n$ 
complex matrix and $e^{i{\bf u}}$ denotes the vector with entries $\{e^{iu_j}\}$
(\cref{table:symbols} for table of symbols for important quantities).
To tackle the Fourier nonlinear term, 
our algorithm employs a generalization of the 
Carleman linearization technique known as Koopman linearization \cite{mauroy2020koopman, brunton2016koopman}.
The exposition of the Koopman linearization technique towards constructing 
an efficient quantum algorithm with rigorous bounds in the worst-case 
complexity forms the main advance in this paper. An overview of 
our approach based on Koopman linearization
and the complexities of the resulting quantum algorithms is
provided in the next section.

We emphasize that the Koopman linearization employed in this work is conceptually different from the Koopman-von Neumann approach used in several previous quantum algorithms for nonlinear differential equations~\cite{jin2023time, joseph2020koopman, barthe2023continuous, tanaka2024polynomialtimequantumalgorithm, novikau2025quantum}. 
Koopman-von Neumann approach embeds the dynamics into a Hilbert-space evolution 
and seeks conditions under which the resulting generator can be simulated efficiently as a Hamiltonian. In contrast, our approach follows the classical Koopman operator framework and constructs an explicit infinite-dimensional linear representation in a Fourier basis of observables. The resulting linearized system is then truncated and analyzed using tools analogous to those used in the Carleman-linearization literature, allowing us to derive rigorous truncation-error bounds and worst-case complexity estimates for Fourier nonlinearities. For a more detailed
comparison of Koopman and Koopman-von Neumann approaches, see Ref.~\cite{lin2022koopman}.

While many of the tools in the error
analysis and circuit construction are 
inspired by the contemporary works based on Carleman linearization,
we adapt them successfully to the problem of tackling Fourier nonlinearity
and offer improvements along the way. There are three 
improvements we highlight in particular:
\begin{enumerate}
    \item Similar to Carleman linearization, the success of Koopman linearization 
    approach is attributed to rigorous bounds
    on errors arising from truncation of the infinite-dimensional linearized ODE.
    In the Carleman literature, the conditions under which such error bounds hold, typically known
    as dissipative conditions, are derived in terms of standard $2$-norm of the initial state
    and operator $2$-norm of the input matrices \cite{krovi2023improved, liu2021efficient, costa2025further}. Ref.~\cite{costa2025further}
    also derived such bounds in terms of the infinity norm. We derive these bounds in terms of 
    $p$-norm and operator $p$-norm respectively for any $p \in [1,\infty)$.
    The overall complexity of solving a given system of nonlinear ODEs is also 
    formulated in terms of these more general norms. 
    Since the dissipative conditions for different values of $p$ are not equivalent,
    the flexibility in the choice of $p$ increases the domain of applicability of 
    our algorithm. 
    \item Much of the literature on Carleman linearization is dedicated to solving
    dissipative nonlinear ODEs \cite{krovi2023improved, liu2021efficient, costa2025further}. This regime is favoured as 
    bounds on Carleman truncation error lead to an efficient quantum algorithm
    for arbitrarily large final time $T$.
    Here we apply a new technique to derive a bound on truncation error
    in the non-dissipative regime where the final time is
    bounded by a quantity expressed in terms of ${\bf u}_0, G_0$ and $G_1$. 
    The truncation error is exponentially small in the truncation 
    level $N$, so the truncated system
    approximates the original linearized system exponentially well in $N$ as desired.
    By employing this bound, we construct an efficient algorithm that does not require
    dissipative conditions but applies only for short final times $T$ quantified 
    rigorously.
    \item Finally, most works in the literature focus on the task of generating a
    quantum state that encodes the solution of the ODE at the final time \cite{krovi2023improved, liu2021efficient, costa2025further}.
    In this work, we go one step further to integrate the readout of a classical
    quantity from the solution at the final time in our quantum algorithm. 
    Crucially, this integration of readout is enabled by 
    our derivation of the truncation error bounds
    on higher-order components in the linearized differential equation.
\end{enumerate}

The organization of the rest of the paper is as follows.
In \cref{sec:contribution}, we provide a non-technical summary of the 
computational problem and the complexity of our algorithm.
\cref{sec:approach} provides a non-technical description of our approach
to the design of the algorithm. 
In \cref{sec:problem-description}, we provide a detailed statement of
the computational problem tackled in this paper.
\cref{sec:design} contains a detailed outline of the design of
our algorithm. In \cref{sec:error-analysis}, we derive bounds on
errors arising from truncation of the linearized ODE as well as
from the truncation of the Taylor series used in mapping the ODE
to a system of linear equations. \cref{sec:complexity} explains 
the construction of the circuit implementing our algorithm 
and its complexity. In \cref{sec:final-results}, we combine the
error analysis from \cref{sec:error-analysis} and 
complexity from \cref{sec:complexity} to derive final overall
complexities. Finally, \cref{sec:conclusion} concludes this paper with a 
summary and discussion of our results as well as open problems.
Proofs of some intermediate lemmas are deferred to the Appendices.

\subsection{Summary of results}\label{sec:contribution}

We now present an informal summary of the main results of the paper. 
We begin by introducing notation for various vector and matrix norms used.
For $p \in [1,\infty)$, the $p$-norm $\|{\bf a}\|_p$ of a vector $\mathbf{a} \in \mathbb{C}^n$  
and the induced (or operator) $p$-norm $\|A\|_p$ of a matrix $A \in \mathbb{C}^{n\times n}$
are given by
\begin{equation}
    \|{\bf a}\|_p = \left(\sum_j|a_j|^p\right)^{1/p},
    \quad \|A\|_p = \sup_{{\bf y} \in \mathbb{C}^n} 
    \frac{\norm{A{\bf y}}_p}{\norm{{\bf y}}_p},
\end{equation}
respectively. By convention, we use $\|\cdot \|=\|\cdot \|_2$.
We also use the maximum row $p$-norm of a matrix, which is defined to be
\begin{equation}
    \|A\|_{{\rm row},p} = \max_{j \in \{1,\dots,n\}} 
    \norm{A_j}_p,    
\end{equation}
where $\norm{A_j}_p$ denotes the (vector) $p$-norm of 
$A_j = [A_{j1} ,\dots, A_{jm}]$, which is the $j$-th row of $A$.\\

The computational problem considered in this paper is the following.

{\noindent \bf Problem \ref{prob:problem_statement_formal}} 
(Fourier ODE problem, informal){\bf .}
We are given a nonlinear ordinary differential equation of the form 
\begin{equation}\label{eq:raw_problem_1}
     \frac{d{\bf u}}{dt} = G_0 + G_1e^{i{\bf u}},
    \end{equation}
along with an initial condition ${\bf u}(0) = {\bf u}_0$,
where 
$G_0 \in \mathbb{C}^{n}$ and $G_1 \in \mathbb{C}^{n \times n}$ are time-independent coefficient matrices, ${\bf u}(t)\in \mathbb{C}^n$ is the solution vector as a function of $t$ and $e^{i{\bf u}}=(e^{iu_1},\cdots , e^{iu_n} )^{\rm T}\in \mathbb{C}^n$.
We are also given a function $g:\mathbb{C}^n\rightarrow \mathbb{C}$ 
expressible as a degree-$K$ Fourier series of the form
\begin{align}
     g({\bf u}) = \sum_{|{\bf j}| \le K } d_{{\bf j}}e^{i{\bf u}\cdot{\bf j}}, \quad d_{{\bf j}} \in \mathbb{C} \quad \forall {\bf j}\in \mathbb{Z}_+^n,
\end{align}
where $\mathbb{Z}_+^n$ denotes $n$-tuples of non-negative integers. 
The goal is to compute $g({\bf u}(T))$ for a given final 
time $T$ to a given accuracy $\epsilon$.

The coefficient matrices for the ODE are given by appropriate 
quantum oracles, as detailed in \cref{prob:problem_statement_formal}.
Where necessary, upper bounds on the norms of these matrices 
are also assumed to be given as input.
Circuits preparing states proportional to $e^{i{\bf u}_0}$
and ${\bf d}$ respectively are assumed to be given as 
input to the problem, where ${\bf d}$ is the vector with 
entries $d_{{\bf j}}$. We point out that 
$g({\bf u})$ is assumed to be have zero coefficients for all Fourier
terms with negative degree for any variable as well as for
all Fourier terms with total degree higher than $K$. 
While relaxing the latter assumption affects the complexity of our algorithm,
relaxing the former assumption would require more substantial changes 
to the linearization.

We construct quantum algorithms for solving the Fourier ODE problem
in dissipative and non-dissipative regimes, which we define ahead. 
The following informal statement of the main result
reports the overall query complexity of our algorithm in the dissipative regime.\\

{\noindent \bf Theorem \ref{theorem:dissipative_result}} 
(Fourier ODE algorithm under dissipative conditions, informal){\bf .}
Let $\{p,q\} \in [1,\infty)$ be such that $1/p+1/q=1$.
Suppose $G_0, G_1$ and ${\bf u}_0$ in the Fourier ODE problem satisfy  
    \begin{equation}
    \label{eq:dissipativity}
        \tilde{\mu}_0 := \min_j \Im{(G_0)_j}\ge 0 \quad \text{and} \quad
         R_p := \frac{\|G_1\|_{{\rm row},q} \|e^{i{\bf u}_0}\|_p}{\tilde{\mu}_0}< \min \left\{1,\frac{\|e^{i{\bf u}_0}\|_p}{\|e^{i{\bf u}_0}\|_2}\right\}.
    \end{equation}
Then, for a fixed $K$ and $R_p$ bounded from above by a fixed constant less than $1$, 
a quantum circuit can be constructed that 
solves the Fourier ODE problem by making 
\begin{equation}
    \tilde{\mathcal{O}}\left(\frac{T^{3/2}\norm{G_0}_{\infty}^{3/2}
     \norm{{\bf d}}}{\epsilon}\left(1+\frac{\norm{G_1}_2}{\|G_1\|_{\rm row, q}}\right)
     \left(1+ \left(\frac{\tilde{\mu}_0}{\|G_1\|_{\rm row,q}}\right)^K\right)
    \right)
\end{equation}
queries, where
$\tilde{\mathcal{O}}$ suppresses polylogarithmic factors in
\begin{equation}
    {\rm polylog}
     \left(n,T,\norm{G_0}_{\infty},\norm{{\bf d}},1/\epsilon,
     \left(1+ \left(\frac{\tilde{\mu}_0}{\|G_1\|_{\rm row,q}}\right)^K\right)
     \right).
\end{equation}
Here, the quantity $R_p$ serves the same purpose as the quantity $R$ defined in 
References~\cite{krovi2023improved, liu2021efficient, costa2025further}, 
which is to quantify the degree of nonlinearity in the ODEs.
The overall query complexity of our algorithm 
scales as ${\rm polylog}(n)$ provided 
\begin{equation}
    T, \norm{G_0}_{\infty}, \norm{{\bf d}},
    \frac{1}{\epsilon}, \frac{\norm{G_1}_2}{\|G_1\|_{\rm row, q}},
     \frac{\tilde{\mu}_0}{\|G_1\|_{\rm row,q}} \in {\rm polylog}(n).
\end{equation}
Therefore, under these conditions our algorithm is exponentially  
faster than classical algorithms for the same problem in terms of $n$.
The complexity of our algorithm with respect to the
final time $T$ scales as $\tilde{O}(T^{3/2})$, which is suboptimal. 
 We believe this inefficiency arises from our analysis
of the linear system at an intermediate stage. Superior techniques
for similar analysis have been developed for related 
problems~\cite{An2026fastforwarding},
which could be applicable here.\\

The above theorem requires the parameters of the ODE to satisfy
the dissipative condition in \cref{eq:dissipativity}. We also construct a quantum 
algorithm for solving the Fourier ODE problem in the absence of the dissipative
condition as long as the final time is sufficiently short. The complexity of this
algorithm is summarized in the next theorem.\\

{\noindent \bf Theorem \ref{thm:non-dissipative-result}} 
(Fourier ODE algorithm without dissipative conditions, informal){\bf .}
Let $p \in [2,\infty)$ and $q=p/(p-1)$. In the Fourier ODE problem,
let $\|e^{i{\bf u}_0}\|_p > 1$, $K \in \mathcal{O}(1)$, and let 
$r,\ell$ be positive real numbers satisfying $r\ge e$ and $\ell \in \mathcal{O}(1)$.
Then for 
\begin{equation}
        0 < T \le \min\left\{\frac{\ell}{\max\left\{\norm{G_0}_{\infty},
    e^{\ell} r\|e^{i{\bf u}_0}\|_p\norm{G_1}_{{\rm row},q}\right\}(1+1/r)}, 
    \frac{\ln{r/e}}{\norm{G_0}_{\infty}+
    e^{\ell} r\|e^{i{\bf u}_0}\|_p\norm{G_1}_{{\rm row},q}}\right\},
    \end{equation}
a quantum circuit can be constructed that 
solves the Fourier ODE problem by making 
\begin{equation}
    \tilde{O}\left(\frac{r^{K\ell}\|\mathbf d\|_{q}^{1+\ell}
    \|e^{i{\bf u}_0}\|_p^{K(\ell +1)}
    }{\epsilon^{1+\ell}}\right).
\end{equation}
queries, where
$\tilde{\mathcal{O}}$ suppresses factors in
${\rm polylog}\left(r,\|e^{i{\bf u}_0}\|_p,\norm{{\bf d}}_q,1/\epsilon\right)$.

The overall query complexity is in ${\rm polylog}(n)$ provided
\begin{equation}
    \norm{{\bf d}}_q, \|e^{i{\bf u}_0}\|_p, \frac{1}{\epsilon} \in {\rm polylog}(n),
\end{equation}
and $r$ is also chosen such that $r \in {\rm polylog}(n)$.
Interestingly, the complexity of our algorithm does not depend on 
$T$, $\norm{G_0}_{\infty}$ and $\norm{G_1}_{{\rm row},q}$. The independence with respect to $T$ is
because the final time is required to be bounded in terms of other quantities.
The independence with respect to $\norm{G_0}$ and $\norm{G_1}_{{\rm row},q}$ is due
to the fact that the upper bound on $T$ is inversely proportional to both
$\norm{G_0}$ and $\norm{G_1}_{{\rm row},q}$. For $\ell \ll 1$, 
our algorithm achieves complexity quasilinear in $1/\epsilon$, but at the cost
of an exponential reduction with respect to $\ell$ in the maximum time $T$.

\subsection{Approach}
\label{sec:approach}

We now describe our approach to the design of quantum 
algorithms for solving these computational problems. 
Similar to previous works based on Carleman linearization~\cite{costa2025further},  
we begin by rescaling the nonlinear ODE to ensure an efficient readout. 
We then map the finite-dimensional nonlinear ODE to an infinite-dimensional linear 
ODE, which is subsequently truncated to order $N$. The linearization procedure used here is the 
Carleman-Fourier linearization method introduced in~\cite{chen2024carleman}, 
which can be viewed as a 
special case of the more general Koopman linearization framework~\cite{mauroy2020koopman}.\\ 

{\noindent \bf Rescaling}: In previous works based on Carleman linearization, the original ODE was
rescaled to obtain the solution state with high probability. In our formulation of the problem,
the need for rescaling relates to efficient readout. If we directly linearize the given ODE, then,
since we do not in general have the 2-norm of $e^{i{\bf u}_0}$ being less than unity, our readout step
would yield $\text{poly}(n)$ scaling. For this reason, we 
define new variable by adding a common imaginary component $i\ln(\nu)$ to all dependent variables, i.e. $x_j = u_j + i\ln(\nu)$,
where $\nu$ is a sufficiently large positive number that ensures the 2-norm of $e^{i{\bf x}}$ is less than unity.
The rescaled variable ${\bf x}$ then satisfies the ODE $d{\bf x}/dt = F_0 + F_1e^{i{\bf x}}$,
where $F_0 = G_0$ and $F_1 = \nu G_1$. Similarly, we define a function $f$ such that $f({\bf x}) = g({\bf u})$ so that readout of $f$ on the rescaled solution vector solves \cref{prob:problem_statement_formal}.\\

{\noindent \bf Carleman-Fourier linearization}: The crucial step in the design of our algorithm is
the application of Koopman linearization.
The idea can be illustrated using a simple single variable nonlinear problem.
Suppose $x(t)$ satisfies the ODE
$dx/dt = F_0 + F_1 e^{ix}$, where $F_0,F_1 \in \mathbb{C}$.
In Carleman linearization, we construct an array 
of monomials $\Psi_{\rm Carl} = [x\ x^2 \dots]^{\rm T}$, and then derive
an equation that governs the evolution of this array.
This is achieved by expressing the derivative of each component
of $\Psi_{\rm Carl}$ as a linear combination of all its components.
For example, we have 
\begin{align}
    \frac{dx}{dt} &= F_0 + F_1 + iF_1x - F_1x^2/2 + \dots, \nonumber\\
    \frac{dx^2}{dt} &= 2xF_0 + 2xF_1 + 2iF_1x^2 - F_1x^3 + \dots
\end{align}
for the first two components. While the ODE can be linearised in this way, the resultant
coefficient matrix of the linearized ODE will not be sparse, as is desired
for developing efficient quantum algorithms. Not only does this hinder
the construction of a block encoding of the coefficient matrix, but
it makes derivation of bounds on truncation errors significantly harder.
While one may approximate this
coefficient matrix by a sparse matrix by truncating the number of terms in
the series expansion of $e^{ix}$, an extra step is then needed to bound the error
arising from such an approximation, and the techniques for the same have 
not been developed in the literature to our knowledge. 
A better approach is to use a different basis to monomials 
for the linearization. This technique is known as Koopman linearization \cite{brunton2016koopman, mauroy2020koopman}. For our ODE, an obvious choice is
the basis of trigonometric monomials of positive degree, 
$\Psi_{\rm Koop} = [e^{ix}\ e^{2ix} \dots]^{\rm T}$. 
Using this basis for the linearization, we obtain
\begin{align}
    \frac{d(e^{ix})}{dt} &= iF_0e^{ix} + iF_1e^{2ix}, \nonumber\\
    \frac{d(e^{2ix})}{dt} &= 2iF_0e^{2ix} + 2iF_1e^{3ix} 
\end{align}
for the first two components. Continuing this process for 
all trigonometric monomials of positive degree, we obtain an infinite-dimensional
homogeneous linear ODE, which we denote as $d\Psi/dt = {{\bf L}}^{\rm T}\Psi$, 
where the matrix of coefficients ${{\bf L}}^{\rm T}$ is sparse 
as it has only two non-zero entries in each row.
The above example demonstrates Koopman linearization for 
an ODE in a single variable, but the idea extends to multiple variables.
In the latter case, $\Psi$ is expressed in terms of the block components 
$\Psi_{\rm Koop} = [e^{i{\bf x}}\ (e^{i{\bf x}})^{\otimes 2} \dots]^{\rm T}$.

In the next step, we truncate the infinite-dimensional 
homogeneous linear ODE $d\Psi_{\rm Koop}/dt = {{\bf L}}^{\rm T}\Psi_{\rm Koop}$ 
to the first $N$ components, denoted by 
$d\Psi^{(N)}_{\rm Koop}/dt = {{\bf L}}_N^{\rm T}\Psi^{(N)}_{\rm Koop}$. This truncation
introduces inevitable errors, which we bound using techniques known
from the literature on Carleman-Fourier linearization. Some such techniques
have been studied independently in Mathematics. We borrow a technique
from the latter works that allows us to obtain bounds on truncation error
in the absence of dissipative conditions. 
We prove and use truncation error bounds not just on the first, lowest component
of $\Psi_{\rm Koop}^{(N)}$, namely $\Psi_1^{(N)} = e^{i{\bf x}}$, 
but also on $K-1$ higher components. These truncation bounds on the higher components 
of $\Psi_{\rm Koop}^{(N)}$ are required for the purpose of readout 
of $g({\bf u}(T))$ in the final step.

Here we remark that the basis of observables we use in our
algorithm is equivalent to $\Psi^{(N)}_{\rm Koop}$ up to padding by some trivial
(zero) observables. This padding allows efficient generation of the initial state
proportional to $\Psi(0)$ and efficient block encoding of ${{\bf L}}_N^{\rm T}$.\\

{\noindent \bf Solution of the homogeneous linear ODE}: Once we have obtained a finite-dimensional system of homogeneous linear ODEs, we convert it to a system of linear equations. For this conversion we make
use of the truncated Taylor series approach in which the total time interval
$[0,T]$ is divided into smaller subintervals of step size $h$ and the propagator
for each time step is approximated by a truncated Taylor series. \\

{\noindent \bf Integrated readout}: At this point it is possible to generate a solution history state for the ODE by using a known algorithm
for solving systems of linear equations. However, our ultimate goal is
not to generate the solution state, but to calculate $g({\bf u}(T))$. It is possible to 
calculate $g({\bf u}(T))$ using multiple copies of the solution state
via the swap test approach, however this would incur a suboptimal $\mathcal{O}(1/\epsilon^2)$ scaling
for the query complexity. We therefore take a  different approach of expressing
$g({\bf u}(T))$ as an expectation value of a unitary operator which allows us to achieve 
$\tilde{O}(1/\epsilon)$ scaling.

\section{Problem description }
\label{sec:problem-description}

In this section, we present a detailed statement for the 
computational problem considered in this paper. 
We begin by introducing notation, reviewing the definition of block encoding
and then provide the complete problem statement.\\

We denote by $ \mathbb{Z}_+^n $ the set of $n$-tuples of non-negative integers.
The ordered set $\mathbb{Z}_{k,+}^n$ is defined as follows using a tensor-style 
enumeration convention \footnote{Our definition of $\mathbb{Z}_{k,+}^n$ differs from the conventional definition, 
according to which $\mathbb{Z}_{k,+}^n$ is the set of \textit{distinct} $n$-tuples 
of non-negative integers summing to $k$. Consequently, under the conventional definition
$|\mathbb{Z}_{k,+}^n| = \binom{n + k - 1}{k}$.}. Consider the set $\{1,\dots,n\}^k$, which is the set of
$k$-tuples ${\bf l} = (l_1,\dots,l_k)$ with each entry in $\{1,\dots,n\}$.
The set $\mathbb{Z}_{k,+}^n$ is obtained by replacing each $k$-tuple
${\bf l}\in \{1,\dots,n\}^k$ by the $n$-tuple 
${\bf j}=(j_1,\dots,j_n) =: {\rm count}(l_1,\dots,l_k)$, where $j_1$ counts the number 
of times $1$ appears in $(l_1,\dots,l_k)$, $j_2$ counts the number 
of times $2$ appears in $(l_1,\dots,l_k)$, etc. 
By definition, each $n$-tuple in $(j_1,\dots,j_n) \in \mathbb{Z}_{k,+}^n$ 
satisfies $\sum_{i=1}^n j_i = k$. However, note that several tuples
appear more than once in the set $\mathbb{Z}_{k,+}^n$. For example,
\begin{align*}
    \mathbb{Z}^3_{2,+} &= {\rm count}\left(\{1,2,3\}^2\right) \\
    &= {\rm count}\left(\{(1,1),(1,2),(1,3),(2,1),(2,2),(2,3),(3,1),(3,2),(3,3)\}\right) \\
    &= \{(2,0,0),(1,1,0),(1,0,1),(1,1,0),(0,2,0),(0,1,1),(1,0,1),(0,1,1),(0,0,2)\}.
\end{align*}
We further identify each $ \mathbf{j}=(j_1,\dots,j_n) \in \mathbb{Z}_{k,+}^n $ with a vector 
in $\mathbb{C}^n$ with integer-valued coordinates $(j_1,\dots,j_n)$, i.e.
$\mathbf{j} = \sum_{l=1}^{n} j_l\mathbf{e}_{l}$, 
and $\{\mathbf{e}_{1},\dots,\mathbf{e}_{n}\}$ denote the standard orthonormal basis
of $\mathbb{C}^n$ with $\mathbf{e}_l = (0,\dots,1_l,\dots,0)^{\rm T}$.
According to our definition of $\mathbb{Z}_{k,+}^n$, its cardinality is
$\left| \mathbb{Z}_{k,+}^n \right| = n^k$.\\

In our formulation of the computational problem, we use the following definition of block encoding.
\begin{definition}[Block-encoding {\cite[Definition~43]{gilyen2019quantum}}]\label{def:block_encoding}
For $\alpha, \varepsilon \in \mathbb{R}_+$ 
we say that the $(s+a)$-qubit unitary $U$ is an $(\alpha, \epsilon)$-block-encoding of the $s$-qubit operator $A$ if
\begin{equation}
\left\lVert A - \alpha \big( (\bra{0^a}\otimes I)\, U \,(\ket{0^a} \otimes I) \big) \right\rVert \leq \varepsilon.
\end{equation}
\end{definition}

We are now ready to state the computational problem considered in this work.
\begin{prob}[Fourier ODE problem]
\label{prob:problem_statement_formal}
We are given a nonlinear ordinary differential equation of the form 
\begin{equation}\label{eq:raw_problem_1}
     \frac{d{\bf u}}{dt} = G_1e^{i{\bf u}} + G_0,\quad {\bf u}(0) = {\bf u}_0
    \end{equation}
where $G_1 \in \mathbb{C}^{n \times n}$ and $G_0 \in \mathbb{C}^{n}$ are time-independent coefficient matrices, ${\bf u}(t)\in \mathbb{C}^n$ is the solution vector as a function of $t$ and $e^{i{\bf u}}=(e^{iu_1},\cdots , e^{iu_n} )^{\rm T}\in \mathbb{C}^n$.
The coefficient matrices $G_0$ and $G_1$ of the ODE are specified by unitary oracles 
$U_{G_0}$ and $U_{G_1}$ respectively.
$U_{G_0}$ is a binary oracle with the 
action $U_{G_0}\ket{0}\ket{j} = \ket{[(G_0)_j]}\ket{j}$,
where $[\bullet]$ denotes binary representation of its argument.  
$U_{G_1}$ is a $(\beta,0)$-block encoding of $G_1$ for a given $\beta \in \mathbb{R}_+$. 
We are also given $\tilde{\mu}_0 = \min_j \Im{(G_0)_j} \in \mathbb{R}^+$ and 
a constant $\alpha \geq \max_j |(G_0)_j|$.
The initial state is specified by a unitary circuit $U_{u_0}$ with the action $U_{u_0}\ket{0}=\ket{e^{i{\bf u}_0}} = e^{i{\bf u}_0}/\|e^{i{\bf u}_0}\|$ and $\|e^{i{\bf u}_0}\|$ is given.\\ 

\noindent We are also given a function $g:\mathbb{C}^n\rightarrow \mathbb{C}$ which can be expressed as a Fourier series comprising only positive frequency components of degree $\le K$, i.e.,
\begin{align}
     g({\bf u}) = \sum_{{\bf j} \in \bigcup_{l=1}^K\mathbb{Z}^n_{l,+}} d_{{\bf j}}e^{i{\bf u}\cdot{\bf j}}, \quad 
     \nonumber d_{{\bf j}} \in \mathbb{C}.
\end{align}
The function $g$ is specified by a unitary circuit $U_{d}$ with the action $U_{d}\ket{0}={\bf d}/\|{\bf d}\|$, and
we are also given $\|{\bf d}\|$.\\

\noindent The goal is to compute $g({\bf u}(T))$ for a given final time $T\in \mathbb{R}_+$ to a given accuracy $\epsilon \in \mathbb{R}_+$.
\end{prob}

While the oracles assumed to provide input are standard, we comment about 
several key points. We use a binary oracle input for $G_0$ instead of a 
circuit preparing a state proportional to $G_0$, as this simplifies our circuit construction
without losing generality. It is worth highlighting that $U_{u_0}$ is assumed to 
prepare a state proportional to $e^{i{\bf u}_0}$ rather than ${\bf u}_0$. 
While constructing an efficient circuit for $U_{u_0}$ is not straightforward
in general, it is possible if  
$e^{i{\bf u}_0}$ is well-approximated by a sparse vector.
The unitary $U_{d}$ prepares the vector of Fourier coefficients for the function $g$.
How to construct such a $U_{d}$ for an arbitrary function $g$ is an interesting problem
that we leave for future work. Finally, for later use we introduce the block components 
${\bf d}_l=(d_{\bf j})_{{\bf j}\in  \mathbb{Z}^n_{l,+}}$ of ${\bf d}$ so that
\begin{align}
     g({\bf u}) = \sum_{{\bf j} \in \bigcup_{l=1}^K\mathbb{Z}^n_{l,+}} d_{{\bf j}}e^{i{\bf u}\cdot{\bf j}} = \sum_{l=1}^{K}{\bf d}_l\cdot[e^{i{\bf u}}]^{\otimes l} = {\bf d}\cdot\left(e^{i{\bf u}\cdot(1,\cdots ,0)}, \cdots, e^{i{\bf u}\cdot(0,\cdots ,K)} \right)^{\rm T}.
\end{align}
Observe that ${\bf d}_l\in \mathbb{C}^{(n^l)}$ and  
${\bf d}\in \mathbb{C}^{D}$, where $D=\frac{n^{K+1}-n}{n-1}$.

\section{Design of the algorithm}
\label{sec:design}
Our approach to solving \Cref{prob:problem_statement_formal} proceeds in four stages: Rescaling, linearization, solving the resulting linear system, and readout. While this structure resembles the methodology used in deriving existing Carleman linearization results such as \cite{liu2021efficient, krovi2023improved, costa2025further}, there are key differences. These arise both from the fact that we address a different class of ODEs, namely ODEs with Fourier nonlinearity, 
and from the introduction of a new technique to tackle the same, namely Koopman linearization.

\subsection{Rescaling}
In general, the 2-norm of $e^{i{\bf u}_0}$ will be greater than unity and readout of classical quantities from the solution state will exhibit inefficient scaling with $n$. For this reason we introduce 
a rescaling of the variable ${\bf u}$ to the variable {\bf x} given by
\begin{equation}\label{eq:rescaling_definition}
    x_j := u_j + i \ln(\nu),\quad j=1,\dots,n,
\end{equation}
where $\nu\in \mathbb{R}_+$ is the rescaling parameter. 
This change of variable rescales the exponent of the initial state, 
since 
\begin{equation}
    e^{i{\bf x}_0} = e^{i{\bf u}_0}/\nu.
\end{equation}
By choosing $\nu$ to be sufficiently large, we can ensure that 
\begin{equation}
\label{eq:rescaling_norm_input}
    \gamma := \norm{e^{i{\bf x}_0}} = \norm{e^{i{\bf u}_0}}/\nu < 1.
\end{equation}

Under this rescaling,
the Fourier ODE in \cref{prob:problem_statement_formal} maps to 
\begin{equation}\label{eq:problem_1_rescaled}
\frac{d{\bf x}}{dt}= F_1e^{i{\bf x}}+F_0, \quad {\bf x}(0)={\bf x}_0 \quad \text{with} \quad F_1 =\nu G_1, F_0 =G_0.
\end{equation} 
Observe that $R_p$ and $\tilde{\mu}_0$  are invariant under the rescaling defined in 
\cref{eq:rescaling_definition}. We therefore refer to $\tilde{\mu}_0=\min_j \Im{(F_0)_j}$ and  $R_p = \frac{\|F_1\|_{{\rm row},q} \|e^{i{\bf x}_0}\|_p}{\tilde{\mu}_0}$ interchangeably with \cref{eq:dissipativity}. \\

Similarly, we can also rescale the coefficients $d_{\bf j}$ in the Fourier series
of the function $g$ in \cref{prob:problem_statement_formal} as
\begin{equation}\label{eq:c_j}
    c_{{\bf j}} := \nu^{\abs{{\bf j}}}d_{{\bf j}},
\end{equation}
with the vector $\textbf{c}$ containing all elements $c_\textbf{j}$. Then the function $f:\mathbb{C}^n\rightarrow \mathbb{C}$ defined by
\begin{equation}\label{eq:rescaled_output_function}
    f({\bf x}) := \sum_{{\bf j} \in \bigcup_{l=1}^K\mathbb{Z}^n_{l,+}} 
c_{{\bf j}}e^{i{\bf x}\cdot{\bf j}}, 
\end{equation}
satisfies \begin{equation}
    f({\bf x}) = g({\bf u}).
\end{equation} The solution to the rescaled ODE, ${\bf x}(t)$, can therefore be used to solve \cref{prob:problem_statement_formal}.

\subsection{The Koopman representation of nonlinear differential equations}
\label{sec:koopman}

The Koopman representation of nonlinear ODEs presents an approach for 
mapping a nonlinear ODE onto an infinite-dimensional linear ODE.
In this section, we review the theory of linearization by Koopman 
representation and also review how Carleman linearization can be
seen as a special case of Koopman linearization. \\

Consider the following first order nonlinear ODE problem 
\begin{equation}  \label{ODE}
\frac{d{\bf x}}{dt}={\bf F}({\bf x}),\quad {\bf x}(t=0):={\bf x}_{0},
\end{equation}
where ${\bf x}(t)\in \mathbb{C}^n, {\bf F}: \mathbb{C}^n\rightarrow \mathbb{C}^n$.
Whereas solving this initial value problem is equivalent to obtaining an explicit
description for the trajectory of the configuration ${\bf x}(t)$,
we are often interested only in calculating some 
function (observable) of the final configuration, $g({\bf x}(t))$, 
where $g:\mathbb{C}^n\rightarrow \mathbb{C}$. 
The Koopman approach proceeds by deriving dynamic equations for the 
evolution of a space $\mathcal{G}$ 
of observables $\psi:\mathbb{C}^n \rightarrow\mathbb{C}$, 
which contains the observable $g$ of interest.
Note that the observables in $\mathcal{G}$ are not necessarily linear,
i.e., $\psi({\bf x}_1+{\bf x}_2) \ne \psi({\bf x}_1)+\psi({\bf x}_2)$. Following the time evolution of trajectories ${\bf x} \to {\bf x}(t)$,
we define the time evolution of observables. For an observable $\psi \in \mathcal{G}$,
we define $\psi(t)$ as
\begin{equation}
    \psi(t)[\bf{x}]:=\psi[{\bf x}(t)] \quad \forall {\bf x} \in \mathbb{C}^n.
\end{equation}
This definition ensures that evaluating time-evolved observable $\psi(t)$ on an initial
configuration $\bf{x}(0) = \bf{x}$ is equivalent to evaluating 
the observable $\psi$ on the time-evolved configuration ${\bf x}(t)$.
Note that by definition, $\psi(0) = \psi$. \\

Next, we define the Koopman Operator associated with the ODE \cite{mauroy2020koopman}.
\begin{definition}[The Koopman Operator]
For the first order nonlinear ODE problem presented in Eq. (\ref{ODE}), consider a (Banach) space $\mathcal{G}$ of observables $\psi:\mathbb{C}^n \rightarrow\mathbb{C}$ 
on the configuration ${\bf x}(t) \in \mathbb{C}^n$. The finite-time Koopman operator 
$K_t:\mathcal{G} \to \mathcal{G}$ 
associated with the evolution 
${\bf x}(0)\rightarrow{\bf x}(t), t\in\mathbb{R}^+$ is defined as 
\begin{equation}\label{Koopman_operator}
K_t\psi = \psi(t), \quad \forall \psi\in \mathcal{G}.
\end{equation}    
\end{definition}
An immediate consequence of this definition is that
\begin{equation}
(K_t\psi)[{\bf x}(0)] = \psi[{\bf x}(t)], \quad \forall \psi\in \mathcal{G}.
\end{equation} 
In the above definition we assumed that the image of 
$K_t$ is in $\mathcal{G}$, which can be ensured in practice
by carefully selecting a sufficiently large $\mathcal{G}$.
It is easy to observe that $K_t$ is a linear operator on $\mathcal{G}$.
In the limit $t\rightarrow0^+$ we define the infinitesimal Koopman generator  
${\bf L}$ as
\begin{equation}
    {\bf L}(\psi):=\frac{d\psi}{dt}=\lim_{t\rightarrow0^+} \frac{\psi(t) -\psi(0)}{t} 
\end{equation}
In general consider a basis $\Psi = \{\Psi_j\}$ of $\mathcal{G}$.
Let ${\bf L}\Psi_j = \sum_k L_{kj}\Psi_k$. Then 
we can represent ${\bf L}$ in the bases 
$\Psi$ by the operator 
\begin{equation}
    [{\bf L}]_\Psi = 
    \begin{bmatrix}
    \vdots & \vdots & \cdots & \cdots & \cdots \\
    [{\bf L}\Psi_1]_{\Psi} & [{\bf L}\Psi_2]_{\Psi} & \cdots & \cdots & \cdots \\
    \vdots & \vdots & \cdots & \cdots & \cdots 
    \end{bmatrix} = 
    [L_{jk}].
\end{equation}
As is usual, since we work in a fixed basis, 
we identify ${\bf L}$ with its representation $[{\bf L}]_\Psi$
and denote both by ${\bf L}$.
We now have
\begin{equation}
\label{eq:koopman-encoding}
    \frac{d\Psi(t)}{dt}={\bf L}^{\rm T}\Psi, 
\end{equation}
as the dynamical equation that governs the evolution of the observables in the
basis $\Psi$. This is called the Koopman representation of the ODE in~\cref{ODE}.
Notice that ${\bf L}^{\rm T}$, rather than ${\bf L}$, 
governs the evolution of $\Psi(t)$ because we are 
focusing on the evolution of the basis of $\mathcal{G}$ rather than
the representation of a specific vector in $\mathcal{G}$.\\

We may now evaluate this homogeneous linear ODE for an initial configuration ${\bf x}_0$
to obtain
\begin{equation}
\label{eq:koopman-encoding2}
    \frac{d(\Psi(t)[{\bf x}_0])}{dt}={\bf L}^{\rm T}(\Psi(t)[{\bf x}_0]).
\end{equation}
Importantly notice that $\Psi(t)[{\bf x}_0]$, which is sometimes called the ``lifted state'',
is an array of complex numbers that depends on the initial configuration ${\bf x}_0$.
Therefore the right-hand side of  
\cref{eq:koopman-encoding} is amenable to solution on a computer.
Moreover, since $g \in \mathcal{G}$ and $\Psi$ is a basis of $\mathcal{G}$,
the observable of interest $g$ can be expressed
as a sum of observables in $\Psi$. That is $g = {\bf d}\cdot\Psi = \sum_j d_j\Psi_j$,
where ${\bf d}$ is an array with complex entries $\{d_j\}$.
We now have 
\begin{equation}
    g({\bf x}(t)) = {\bf d}\cdot\Psi[{\bf x}(t)] = 
{\bf d}\cdot\Psi(t)[{\bf x}_0].
\end{equation}
Therefore, obtaining the value of $\Psi(t)[{\bf x}_0]$ at time $t$ 
also allows us to obtain $g({\bf x}(t))$.\\

The space $\mathcal{G}$ of observables and therefore
the infinitesimal Koopman operator ${\bf L}$, are often infinite-dimensional. 
In order to use the Koopman operator framework to solve nonlinear ODEs using numerical methods, 
it will be necessary to obtain an approximate, finite-dimensional matrix representation of the 
infinitesimal Koopman generator ${\bf L}$. The next proposition
describes how this is achieved via projection on a subspace of $\mathcal{G}$.
\begin{proposition}[Projected infinitesimal Koopman matrix]\label{prop:infKoop}
Consider an $N$ dimensional subspace $\mathcal{G}_N \subset \mathcal{G}$ spanned by the basis functions $\Psi^{(N)} = \{\Psi_j\}_{j=1}^N$ and let $\Pi: \mathcal{G}\rightarrow \mathcal{G}_N$ be a projection map. An $N$-dimensional matrix representation of the infinitesimal Koopman generator, ${\bf L}_N$, has $j^{\text{th}}$ column $\Pi {\bf L} \Psi_j = \Pi\frac{d\Psi_j}{dt}$ represented in the basis $\{\Psi_j\}_{j=1}^N$.  
\end{proposition}
In the simple case where 
\begin{equation}
    \Pi \Psi_{j} = \left\{
    \begin{array}{lcl}
    \Psi_j & \text{if} & j \le N \\
    0 & \text{if} & j>N
    \end{array}\right.,
\end{equation}
${\bf L}_N^{\rm T}$ is simply the top-left $N\times N$ block of ${\bf L}$.
Next, the following proposition defines approximate evolution of the 
truncated lifted state.
\begin{proposition}[Truncated lifted Dynamics] \label{prop:LifDy}
Define the $N$-dimensional \textit{lifted state} as 
${\Psi}^{(N)}$ with elements ${\Psi}_j^{(N)}(t)$. 
An $N$-dimensional approximate Koopman representation of the ODE in \cref{ODE} is given by
\begin{equation}
\label{eq:truncatedkoopman}
    \frac{d}{dt}{\Psi}^{(N)}={\bf L}_N^{\rm T}{\Psi}^{(N)},
\end{equation}
with initial data $\Psi^{(N)}({\bf x}_0)$.
\end{proposition}
Notice that \cref{eq:truncatedkoopman} is not identical to the
exact ODE $\Pi {\bf L} \Psi_j = \Pi\frac{d\Psi_j}{dt}$.
Therefore, even though ${\Psi}^{(N)}$ and $\Pi\Psi$ coincide
at initial time $t=0$, in general ${\Psi}^{(N)}(t) \ne \Pi\Psi(t)$
for $t>0$. Here $\Psi_j^{(N)}(t)$ can be seen as an approximation to 
$\Psi_j(t)$.\\

\cref{prop:infKoop} and \cref{prop:LifDy} demonstrate (where the proof can be found in sections 1.4.1.1 and 1.4.1.2 of \cite{mauroy2020koopman} respectively) that the nonlinear ODE in \cref{ODE} can be approximately represented as an $N$-dimensional linear ODE (the lifted dynamics). The solution to this linear ODE, ${\Psi}^{(N)}(t)[{\bf x}_0]$ contains the approximate 
trajectories of $\{\Psi_j\}_{j=1}^N$, a basis for~$\mathcal{G}_N$. 
The matrix that defines this linear ODE, ${\bf L}_N^{\rm T}$ is an 
$N$-dimensional approximation of the infinitesimal Koopman matrix in 
the basis $\{\Psi_j\}_{j=1}^N$. Since $\{\Psi_j\}_{j=1}^N$ is a 
basis for $\mathcal{G}_N$, the closest approximation $g_N\in \mathcal{G}_N$ 
to any observable $g\in\mathcal{G}$ can be obtained using $\Psi^{(N)}$ 
through the expansion $g \approx g_N = \sum_{j=1}^Nd_j\Psi_j$. 
As $N\rightarrow\infty$, if $g_N(t)\rightarrow g(t)$, 
then the value of the observable obtained using the truncated lifted dynamics 
$g_N(t)[{\bf x}_0] = g_N({\bf x}(t))\rightarrow g(t)[{\bf x}_0] = g({\bf x}(t))$ converges to the desired value.\\

It is clear from the discussion above that 
the choice of $\mathcal{G}$ dictates 
which observables can be inferred by solving the lifted dynamics,
and the choice of the basis functions, $\Psi$, 
determine the structure of the infinitesimal 
Koopman operator. We now review how Carleman linearization is a special
case of Koopman linearization discussed above.
\begin{exmp}[Carleman Linearization]
Consider an ODE with a single dependent variable $x(t) \in \mathbb{C}$. 
Then Koopman linearization with the choice of polynomial functions as the 
space of observables, i.e. 
$\mathcal{G} = \mathbb{C}[x]$  and 
the choice of monomials as the basis functions, namely
$\Psi = \{\Psi_j:x \to x^j\}$ is equivalent to Carleman linearization. 
For example, consider the simple $n=1$ nonlinear ODE $\frac{dx}{dt} = -x^2$. 
Since $\Psi_j(t)[x] = \Psi_j(0)[x(t)] = (x(t))^j$,
from here onwards we write $\Psi_j =  x^j$.
Note that this simple equation conveys two important pieces of information.
One is the definition of the basis function itself, which maps
any given configuration $x \in \mathbb{C}$ to 
its observable value $x^j \in \mathbb{C}$. Second, it denotes that
$\Psi_j(t)$ acts on any initial configuration $x(0)$ to yield $(x(t))^j$.
Now we have 
\[
\frac{d\Psi_j}{dt} = \frac{dx^j}{dt} = jx^{j-1}\frac{dx}{dt}.
\]
Then using the given ODE, $\frac{dx}{dt} = -x^2$ we obtain
\[
\frac{d\Psi_j}{dt} = -jx^{j+1} = -j\Psi_{j+1}.
\]
We therefore have
\begin{equation} 
{\bf L}\Psi_j = -j \Psi_{j+1}.
\end{equation}
Therefore, the given nonlinear ODE can be expressed as the
infinite-dimensional homogeneous linear ODE
\begin{equation}
\frac{d}{dt} \begin{pmatrix}
    x\\x^2 \\\vdots \\\vdots 
\end{pmatrix} = {\bf L}^{\rm T}\begin{pmatrix}
    x\\x^2 \\\vdots \\\vdots 
\end{pmatrix} = 
\begin{pmatrix}
0 & -1 \\
 & 0 & -2 \\
 & & 0 & \ddots \\
 & & & \ddots  
\end{pmatrix}
\begin{pmatrix}
    x\\x^2 \\\vdots \\\vdots 
\end{pmatrix}.
\end{equation}
After truncation according to \cref{prop:LifDy},
we get an $N$-dimensional linear ODE 
\begin{equation}
\label{Koopman_example1}
    \frac{d}{dt}{\Psi}^{(N)} = {\bf L}_N^{\rm T}{\Psi}^{(N)},
\quad \text{with} 
\quad
{\bf L}_N^{\rm T} = 
\begin{pmatrix}
0 & -1 \\
 & 0 & -2 \\
 & & 0 & \ddots \\
 & & & 0 & -(N{-}1) \\
 &&&&0\\
\end{pmatrix}.
\end{equation}
The truncation to $N$ basis functions now gives rise to error,
resulting in ${\Psi}_j^{(N)}(t) \ne \Psi_j(t)$
for $t>0$ in general.
This is precisely what we obtain by following the 
traditional process of Carleman linearization presented in Ref.~\cite{kowalski1991nonlinear, childs2021high}.
The same conclusion can be drawn for $n>1$, in which case the basis functions 
are chosen to be monomials in $n$ variables.
\end{exmp}

\subsection{Carleman-Fourier Linearization}

The Koopman representation extends Carleman Linearization by providing a procedure to linearize a nonlinear ODE in any choice of basis. Linearization using the Fourier basis $\bar{\Psi}_j({\bf x}) = \left[e^{i\mathbf{x}}\right]^{\otimes j}$ has been referred to as Carleman-Fourier linearization \cite{chen2024carleman} and is the focus of this work. Here $\bar{\Psi}_j$ denotes a set of
basis functions rather than a single basis function. The overhead bar is used to indicate that
these are unpadded observables, and we will later denote observables padded 
with trivial entries without bars; this padding enables easier circuit
construction. We now outline how the Koopman representation with the Fourier basis can be used to represent the rescaled $n$-dimensional nonlinear ODE given in \cref{eq:problem_1_rescaled} as a linear ODE. Following the procedure outlined in \cref{sec:koopman} with 
$\bar{\Psi}_j({\bf x}) =[e^{i{\bf x}}]^{\otimes j}$, we show how to obtain 
$\bar{{\bf L}}_N$ from \cref{prop:infKoop}.\\

Using the Hadamard (element-wise) product of matrices 
defined as $(A \circ B)_{ij} \coloneqq A_{ij} B_{ij}$,
the derivative of \( e^{i\mathbf{x}} \) can be expressed as
    \begin{equation}
    \label{eq:first_deri}
    \frac{d e^{i\mathbf{x}}}{dt}
    = i \begin{pmatrix}
        \dot{x}_1 e^{i x_1} \\
        \dot{x}_2 e^{i x_2} \\
        \vdots
    \end{pmatrix}
    = i\, \dot{\mathbf{x}} \circ e^{i \mathbf{x}}
    = i \left( F_0 + F_1 e^{i \mathbf{x}} \right) \circ e^{i \mathbf{x}}
    = i \begin{pmatrix}
        (F_0)_1 + \sum_j (F_1)_{1j} e^{i x_j} \\
        (F_0)_2 + \sum_j (F_1)_{2j} e^{i x_j} \\
        \vdots
    \end{pmatrix} \circ e^{i \mathbf{x}}.
    \end{equation}
Applying the Hadamard product gives
\begin{align}
 \frac{d e^{i\mathbf{x}}}{dt} &= i\begin{pmatrix}
    (F_0)_1e^{ix_1} + \sum_j (F_1)_{1j}e^{i(x_{j}+x_1)} \\(F_0)_2e^{ix_2} + \sum_j (F_1)_{2j}e^{i(x_{j}+x_2)}\\ \vdots
\end{pmatrix}\nonumber\\
&= i\widetilde{F}_1 [e^{i{\bf x}}]^{\otimes2}+i\widetilde{F}_0e^{i{\bf x}}\nonumber\\
&=i\widetilde{F}_1 \bar{\Psi}_2+i\widetilde{F}_0\bar{\Psi}_1,
\end{align}
where $\widetilde{F}_1$ and $\widetilde{F}_0$ are given by
\begin{equation}
\label{eq:F0_tilde} 
\widetilde{F}_0 = \text{diag}(F_0) =  \begin{pmatrix}
    (F_0)_1 &  & \\ & (F_0)_2 & \\& & \ddots \\ & & & (F_0)_n 
\end{pmatrix} \in \mathbb{C}^{n\times n},
\end{equation} 
\begin{equation}
\label{eq:F1_tilde} 
\widetilde{F}_1 = \left(
\begin{array}{c@{\qquad}c@{\qquad}c@{\qquad}c}
\begin{matrix}(F_1)_{11} & \cdots & (F_1)_{1n}\end{matrix} & & & \\
& \begin{matrix}(F_1)_{21} & \cdots & (F_1)_{2n}\end{matrix} & & \\
& & \ddots & \\
& & & \begin{matrix}(F_1)_{n1} & \cdots & (F_1)_{nn}\end{matrix}
\end{array}
\right) \in \mathbb{C}^{n\times n^2}\end{equation}
Here $\bar{\Psi}_j(t)$ is defined as   
\begin{equation}
    \label{eq:psi_j}
        \bar{\Psi}_j(t) = \left[e^{i\mathbf{x}(t)}\right]^{\otimes j},
    \end{equation}
    where \( \otimes \) denotes the tensor product and the $\bar{\cdot}$ notation is defined in  \cref{eq:padding}. 
    Explicitly, for \( \mathbf{x}(t) \in \mathbb{C}^n \), the vector \( \bar{\Psi}_j \in \mathbb{C}^{n^j} \) is given by:
    \begin{equation}
    \label{eq:tensM}
    \bar{\Psi}_j = \left(
        e^{i j x_1},\,
        e^{i[(j-1)x_1 + x_2]},\,
        \dots,\,
        e^{i[x_1 + (j-1)x_n]},\,
        e^{i[x_2 + (j-1)x_1]},\,
        \dots,\,
        e^{i[(j-1)x_n + x_{n-1}]},\,
        e^{i j x_n}
    \right)^{\rm T}.
\end{equation}
As a concrete example, for \( \mathbf{x} \in \mathbb{C}^2 \) and \( j = 2 \), we obtain:
\begin{equation}
        \bar{\Psi}_2 = \left( e^{i2x_1},\, e^{i(x_1 + x_2)},\, e^{i(x_1 + x_2)},\, e^{i2x_2} \right)^{\rm T} \in \mathbb{C}^{n^2}.
\end{equation}
The $j^{\text{th}}$ column of ${\bf L}_N$ is obtained from  
\begin{align}
\label{eq:Carl_line}
\frac{d\bar{\Psi}_j}{dt} &= \frac{d[e^{i\mathbf{x}}]^{\otimes j}}{dt} \nonumber \\  
&=\frac{de^{i\mathbf{x}}}{dt}\otimes e^{i\mathbf{x}} \otimes \cdots \otimes e^{i\mathbf{x}} + \dots + e^{i\mathbf{x}}\otimes e^{i\mathbf{x}} \otimes \dots \otimes \frac{de^{i\mathbf{x}}}{dt} \nonumber \\ 
&= \left(i\widetilde{F}_1 \bar{\Psi}_2 + \widetilde{F}_0 \bar{\Psi}_1\right) e^{i\mathbf{x}} \otimes e^{i\mathbf{x}} \otimes \dots \otimes e^{i\mathbf{x}} + \dots +  e^{i\mathbf{x}}\otimes e^{i\mathbf{x}} \otimes\cdots \otimes  \left(i\widetilde{F}_1 \bar{\Psi}_2 + \widetilde{F}_0 \bar{\Psi}_1\right).
\end{align}
By splitting this summation into two terms, we get 
\begin{multline}
\frac{d\bar{\Psi}_j}{dt} =
i\widetilde{F}_1 \bar{\Psi}_2  \otimes e^{i\mathbf{x}} \otimes \dots \otimes e^{i\mathbf{x}} + \dots +  e^{i\mathbf{x}}\otimes e^{i\mathbf{x}} \otimes\cdots \otimes  i\widetilde{F}_1 \bar{\Psi}_2 \\
+ \widetilde{F}_0 \bar{\Psi}_1 \otimes e^{i\mathbf{x}} \otimes \dots \otimes e^{i\mathbf{x}} + \dots +  e^{i\mathbf{x}}\otimes e^{i\mathbf{x}} \otimes\cdots \otimes   \widetilde{F}_0 \bar{\Psi}_1.
\end{multline}
By using the definition of $\Psi_j$, we can recast this equation as a recurrence relation
\begin{multline}
\label{eq:recurrence_long}
\frac{d\bar{\Psi}_j}{dt} =
(i\widetilde{F}_1\otimes I^{\otimes (j-1)} + I\otimes i\widetilde{F}_1\otimes I^{\otimes(j-2)} + \dots + I^{\otimes (j-1)}\otimes i\widetilde{F}_1)\bar{\Psi}_{j+1} \\
+ (i\widetilde{F}_0\otimes I^{\otimes (j-1)} + I\otimes i\widetilde{F}_0\otimes I^{\otimes(j-2)} + \dots + I^{\otimes (j-1)}\otimes i\widetilde{F}_0)\bar{\Psi}_j.
\end{multline}
By defining $\bar{B}^{(1)}_{j+1} \in \mathbb{C}^{n^j \times n^{j+1 }}$ and $\bar{B}^{(0)}_{j} \in \mathbb{C}^{n^j \times n^{j}}$ as 
\begin{align}
\label{eq:carleman-assembly2}
\bar{B}^{(1)}_{j+1} &= i\widetilde{F}_1\otimes I^{\otimes (j-1)} + I\otimes i\widetilde{F}_1\otimes I^{\otimes(j-2)} + \dots + I^{\otimes (j-1)}\otimes i\widetilde{F}_1  \nonumber\\
 \bar{B}^{(0)}_{j} &=i\widetilde{F}_0\otimes I^{\otimes (j-1)} + I\otimes i\widetilde{F}_0\otimes I^{\otimes(j-2)} + \dots + I^{\otimes (j-1)}\otimes i\widetilde{F}_0,
\end{align}
\cref{eq:recurrence_long} can be written as 
\begin{equation}\label{eq:recurence_relation}
\frac{d\bar{\Psi}_j}{dt}=   \bar{B}_{j+1}^{(1)}\bar{\Psi}_{j+1} + \bar{B}_{j}^{(0)}\bar{\Psi}_{j},
\end{equation}
Then by \cref{prop:LifDy}, the lifted dynamics is given by the following infinite system of linear ODEs 
\begin{equation}
\label{eq:Lineari_syst}
\frac{d}{dt}\begin{pmatrix}
    \bar{\Psi}_1  \\
    \bar{\Psi}_2   \\
     \vdots  \\
\end{pmatrix} = \begin{pmatrix}
    \bar{B}_0^{(0)} & \bar{B}_1^{(1)}  & 0 & 0 & \dots \\
    0 & \bar{B}_1^{(0)} & \bar{B}_2^{(1)} & 0 & \dots \\
     \vdots &  \vdots & \ddots &  \vdots & \dots 
\end{pmatrix}
\begin{pmatrix}
    \bar{\Psi}_1  \\
    \bar{\Psi}_2   \\
     \vdots 
\end{pmatrix},
\end{equation}
with initial data $\bar{\Psi}({\bf x}_{\text{0}})$.
This procedure can also be viewed as a generalization of the Carleman linearization method used in Ref.~\cite{costa2025further, krovi2023improved, liu2021efficient} where in \cref{eq:Carl_line} we introduce Fourier basis functions instead of monomials as the change of variables.\\

The solution to the order $N$ truncated system 
\begin{equation}
\bar{\Psi}^{(N)}=\begin{bmatrix}
    \bar{\Psi}_1(t) & \bar{\Psi}_2(t) & \dots & \bar{\Psi}_N(t)
\end{bmatrix}^{\rm T},
\end{equation}
 can be used to obtain a good approximation to infinite dimensional system. The truncated lifted state follows the equation
\begin{equation}
\label{eq:Lin_ODE}
\frac{d \bar{\Psi}^{(N)}}{dt}={\bf \bar{L}}_N^{\rm T}\bar{\Psi}^{(N)},
\end{equation}
where
\begin{equation}
\label{eq:LN}
{\bf \bar{L}}_N^{\rm T} = \begin{pmatrix}
    \bar{B}_1^{(0)} & \bar{B}_2^{(1)}  & 0 & 0 & \dots \\
    0 & \bar{B}_2^{(0)} & \bar{B}_3^{(1)} & 0 & \dots \\
     \vdots &  \vdots & \ddots &  \vdots & \dots \\
    0 & 0 & \cdots & 0 & \bar{B}_N^{(0)}
\end{pmatrix},   
\end{equation}
with initial data $\bar{\Psi}^{(N)}({\bf x}_{\text{0}})$.\\

The blocks in adjacent rows of both $\bar{\Psi}^{(N)}$ and ${\bf \bar{L}}_N^{\rm T}$
are of unequal sizes which is unfavorable in the design of a quantum circuit to solve the linear ODE \cref{eq:Lin_ODE}. For this we follow Ref.~\cite{an2023quantum} and introduce states and matrices padded with trivial entries; 
\begin{align}
\label{eq:padding}
    \Psi_j(t)& = {\bf e}_1^{\otimes N-j}\otimes \bar {\Psi}_j(t) \in \mathbb{C}^{n^N},
    \nonumber\\
    B_j^{(0)} & =\mathds{1}^{\otimes N-j} \otimes \bar{B}_j^{(0)}\in\mathbb{C}^{n^N\times n^N},
    \nonumber\\
    B_{j}^{(1)} & =\mathds{1}^{\otimes N-j} \otimes  ({\bf e}_1\otimes \bar{B}_{j}^{(1)})\in\mathbb{C}^{n^N\times n^N},
\end{align}
where ${\bf e}_1$ denotes $n$-dimensional vector ${\bf e}_1 = \begin{bmatrix}
    1 & 0 & \dots \end{bmatrix}^{\rm T}$.
The linearized system in \cref{eq:Lin_ODE}
is then equivalent to the system given by 
\begin{equation}
\label{eq:Lin_ODE_padded}
\frac{d \Psi^{(N)}}{dt}={\bf L}_N^{\rm T}\Psi^{(N)}.
\end{equation}
where 
\begin{equation}
\label{eq:padded_Psi_N}
{\Psi}^{(N)}=\begin{bmatrix}
    {\Psi}_1(t) & {\Psi}_2(t) & \dots & {\Psi}_N(t)
\end{bmatrix}^{\rm T}
\end{equation}
and 
\begin{equation}
\label{eq:Lin_ODE_with_zeros}
    [{\bf L}_N^{\rm T}]_{j,j} =B_j^{(0)}, \quad 
    [{\bf L}_N^{\rm T}]_{j,j+1}=B_{j+1}^{(1)}.
\end{equation}
The next section discusses our strategy for solving \cref{eq:Lin_ODE_padded}.

\subsection{Solution of linear differential equation by truncated Taylor series approach}

After reducing the rescaled nonlinear differential equation, \cref{eq:problem_1_rescaled}, to a linear one, \cref{eq:Lin_ODE_padded},
we now discuss our approach to solving the resulting linear differential equation.
We use the truncated Taylor series approach, which has been explored 
previously~\cite{krovi2023improved}.
Following Ref.~\cite{krovi2023improved}, the $Nn^N$ dimensional linear system of ODEs can be reduced to the  $2mNn^N$ dimensional algebraic system of linear equations, i.e.,
\begin{equation}
\label{eq:linear_algebraic_system}
\mathcal{L}_{N,m}\mathcal{X}_{N,m} = \mathcal{B}_{N,m},   
\end{equation}
where 
\begin{equation}\label{eq:linear_algebraic_system_M}
    \underbrace{\begin{bmatrix}
    \openone & 0 & 0 & 0 & 0 & 0 &  0 \\
    -V & \ddots & 0 & 0 & 0 & 0 &  0 \\
    0 & \ddots & \openone & 0 & 0 & 0 & 0 \\
    0 & 0 & -V & \openone & 0 & 0 & 0 \\
    0 & 0 & 0 & -\openone & \openone & 0 & 0 \\
    0 & 0 & 0 & 0 & \ddots & \ddots & 0 \\
    0 & 0 & 0 & 0 & 0 & -\openone & \openone
\end{bmatrix}}_{\mathcal{L}_{N,m}}
\underbrace{\begin{bmatrix}
\Psi^{(N)}(0) \\ \Psi^{(N)}(h) \\ \vdots \\ \Psi^{(N)}((m-1)h) \\ \Psi^{(N)}(mh)
\\ \Psi^{(N)}(mh) \\ \vdots \\ \Psi^{(N)}(mh)
\end{bmatrix}}_{\mathcal{X}_{N,m}}
=
\underbrace{\begin{bmatrix}
\Psi^{(N)}(0) \\ {\bf 0} \\ \vdots \\ \vdots \\ \vdots \\ \vdots \\ {\bf 0}
\end{bmatrix}}_{\mathcal{B}_{N,m}}.
\end{equation}
Here, $h$ is the small time step, and $m$ is the number of steps such that $mh = T$.
 $m$ is chosen to be an integer power of $2$ and the number of extra copies of the final state $p$ is chosen to be $p=m$. We have 
\begin{equation}
\label{eq:V}
    V = e^{{\bf L}_N^{\rm T}h}.
\end{equation} 
In principle, one can obtain the solution history state 
$\mathcal{X}_{N,m}$
by inverting the matrix $\mathcal{L}_{N,m}$. However, since we do not have
access to $V$, we replace $V$ in $\mathcal{L}_{N,m}$ by 
a truncated Taylor series
\begin{equation}
\label{eq:V_k}
    V_{k} = \sum_{j=0}^{k}\frac{({\bf L}_N^{\rm T}h)^j}{j!},
\end{equation}
with $k$ the order of truncation. 
We denote this new matrix by $\mathcal{L}_{N,m,k}$, and
the corresponding vector approximating the solution history state
${\mathcal{X}}_{N,m,k}$. Then by definition, 
${\mathcal{X}}_{N,m,k}$ satisfies
\begin{equation}
\label{eq:TaylorSystem}
\mathcal{L}_{N,m,k}{\mathcal{X}}_{N,m,k} = \mathcal{B}_{N,m}.  
\end{equation}
We denote the approximated solution in ${\mathcal{X}}_{N,m,k}$ 
at $j$th time step by $\Phi_j$, and therefore $\Phi_j \approx \Psi^{(N)}(jh)$.
Finally, the Taylor-truncated system of equations can be written in the form
\begin{equation}
\label{eq:TaylorSystem_M}
    \underbrace{\begin{bmatrix}
    \openone & 0 & 0 & 0 & 0 & 0 &  0 \\
    -V_k & \ddots & 0 & 0 & 0 & 0 &  0 \\
    0 & \ddots & \openone & 0 & 0 & 0 & 0 \\
    0 & 0 & -V_k & \openone & 0 & 0 & 0 \\
    0 & 0 & 0 & -\openone & \openone & 0 & 0 \\
    0 & 0 & 0 & 0 & \ddots & \ddots & 0 \\
    0 & 0 & 0 & 0 & 0 & -\openone & \openone
\end{bmatrix}}_{\mathcal{L}_{N,m,k}}
\underbrace{\begin{bmatrix}
\Phi_0 \\ \Phi_1 \\ \vdots \\ \Phi_{m-1} \\ \Phi_m
\\ \Phi_m \\ \vdots \\ \Phi_m
\end{bmatrix}}_{{\mathcal{X}}_{N,m,k}}
=
\underbrace{\begin{bmatrix}
\Psi^{(N)}(0) \\ {\bf 0} \\ \vdots \\ \vdots \\ \vdots \\ \vdots \\ {\bf 0}
\end{bmatrix}}_{\mathcal{B}_{N,m}}.
\end{equation}
Here the dependence of the components $\Phi_j$ on $N$, $m$, $h$
is suppressed in the notation for brevity.
For sufficiently high value of $k$, the solution state 
${\mathcal{X}}_{N,m,k}$ is a good approximation of
$\mathcal{X}_{N,m}$.

\subsection{Reduction to expectation value of a unitary operator}\label{sec:Reduction to expectation value of a unitary operator}
We now discuss how to efficiently compute $f({\bf x})$ 
using the linear system solver described in the previous section.
In contrast to the previous work on quantum Carleman linearization, we are going to use the first $K$ components of the solution state of the linear system.
We first redefine the vector ${\bf c}$ from \cref{eq:c_j} as  
\begin{equation}\label{eq:c_def}
    {\bf c} = \bigoplus_{j=1}^{N}{\bf e}_1^{\otimes N-j} \otimes {\bf c}_j  \in \mathbb{C}^{Nn^N},
\end{equation}
where for $j\leq K$ we have ${\bf c}_j=\{c_{\bf j}\,  | \, {\bf j}\in\mathbb{Z}^n_{j,+}\}\in\mathbb{C}^{n^j}$ and for all $j>K$ we have ${\bf c}_j={\bf 0}$. Now ${\bf c}$ has $D=\frac{n^{K+1}-n}{n-1}$ nontrivial terms $c_{\bf j}$ for ${\bf j} \in  \bigcup_{l=1}^K\mathbb{Z}^n_{l,+}$.
We begin by rewriting \cref{eq:rescaled_output_function} in terms of components 
$\Psi_j$ at final time $T$ as
\begin{equation}
\label{eq:exact}
f({\bf x}) = \sum_{j=1}^{K} {\bf c}_j\cdot \Psi_j(T).    
\end{equation}
For sufficiently high Koopman truncation order $N$, we have $\Psi_j \approx \Psi^{(N)}_j$.
Therefore, we can approximate $f({\bf x})$ by 
\begin{equation}
\label{eq:f^N}
f({\bf x}) = \sum_{j=1}^{K} {\bf c}_j\cdot \Psi_j(T) 
\approx \sum_{j=1}^{K} {\bf c}_j\cdot \Psi^{(N)}_j(T) = {\bf c}\cdot\Psi^{(N)}(T)  
\end{equation}
Furthermore, for sufficiently high Taylor truncation order $k$, we have 
$\Psi^{(N)}_j(T) \approx [\Phi_m]_j$. Therefore we can approximate ${\bf c}\cdot\Psi^{(N)}(T)$ by 
${\bf c}\cdot \Phi_m$. \\

To express ${\bf c}\cdot \Phi_m$ in terms of 
the solution history state ${\mathcal{X}}_{N,m,k}$, we first define
\begin{equation}\label{eq:mathcal_C}
\mathcal{C} = \frac{1}{m}\begin{bmatrix}
{\bf 0} & {\bf 0 }& \dots & {\bf 0}_m & {\bf c}
& {\bf c} & \dots & {\bf c}_{2m}
\end{bmatrix}^{\rm T}\in \mathbb{C}^{2mNn^N}.\end{equation} 
Then we have
\begin{equation}\label{eq:f_n_hat}
f({\bf x}) \approx {\bf c}\cdot\Phi_m 
    = \mathcal{C}\cdot {\mathcal{X}}_{N,m,k} = \mathcal{C}\cdot 
    \mathcal{L}_{N,m,k}^{-1}\mathcal{B}_{N,m}. 
\end{equation}

Our quantum circuit computes the expression 
$\mathcal{C}\cdot \mathcal{L}_{N,m,k}^{-1}\mathcal{B}_{N,m}$
by computing the expectation value of a unitary operator. For this purpose, we first express $\mathcal{C}\cdot \mathcal{L}_{N,m,k}^{-1}\mathcal{B}_{N,m}$
as an expectation value. 
Let $\mathcal{C}^{\rm T} = \alpha_C({}_1\!\bra{0}U_C^\dagger)$ 
and $\mathcal{B} = \alpha_\mathcal{B} U_\mathcal{B} \ket{0}_1$
for unitaries $U_\mathcal{C}$ and $U_\mathcal{B}$ acting on $(1+\log_2(m) + \log_2(N)+N\log_2 n)$-qubit register 1. 
Let $U_{\mathcal{L}^{-1}}$ be a block encoding of $\mathcal{L}_{N,m,k}^{-1}/\alpha_{\mathcal{L}^{-1}}$
with ancilla register 2, so that 
$\alpha_{\mathcal{L}^{-1}}({}_2\!\expval{0|U_{\mathcal{L}^{-1}}|0}_2) \approx \mathcal{L}_{N,m,k}^{-1}$.
Then we have
\begin{equation}\label{eq:reducation_to_exp_of_unitary}
    f({\bf x}) \approx \mathcal{C}\cdot \mathcal{L}_{N,m,k}^{-1}\mathcal{B}_{N,m}
    = \alpha_\mathcal{C}\alpha_{\mathcal{L}^{-1}}\alpha_\mathcal{B}
    \bra{0}U_\mathcal{C}^\dagger U_{\mathcal{L}^{-1}} U_\mathcal{B}\ket{0},
\end{equation}
where $\ket{0} = \ket{0}_1\ket{0}_2$. 
In our quantum circuit, we assume access to $U_\mathcal{C}$, and we construct
$U_{\mathcal{L}^{-1}}$ and $U_\mathcal{B}$ by adapting techniques in the literature.

\section{Error analysis}
\label{sec:error-analysis}
In this section we obtain upper bounds on the errors that arise 
from truncation of the Koopman representation as well as from
truncation of the Taylor series in solving the linearized ODE.

\subsection{Analysis of errors arising from truncation of the linearized system
of differential equations}
In this section, we derive bounds on the error due to truncation of the linearized ODE,
first under dissipative conditions and then without dissipative conditions.

\subsubsection{Truncation error under dissipative conditions}
We begin by establishing the conditions under which the rescaled nonlinear ODE problem given in \cref{eq:problem_1_rescaled} is dissipative. 
We derive these conditions
with respect to the $p$-norm for arbitrary $p \in [1,\infty]$.
Subsequently, we perform an error analysis of the linearization, deriving the conditions that ensure the stability of the resulting linearized ODE. The techniques used in this latter 
part are very close to those used in Ref.~\cite{costa2025further}.

\begin{lemma}[Dissipativity in the  $p$-Norm]\label{lemm:dissp}
\label{lem:pnormdecrease}
Let $\{p,q\} \in [1,\infty)$ be such that $1/p+1/q=1$. 
Consider the rescaled ODE from \cref{eq:problem_1_rescaled}  where   
\begin{equation}
    \tilde{\mu}_0 := \min_j \Im{(F_0)_j}\ge 0 \quad \text{and} \quad
     R_p := \frac{\|F_1\|_{{\rm row},q} \|\Psi_1(0)\|_p}{\tilde{\mu}_0}<1.
\end{equation}
Then
\begin{equation}
    \|\Psi_1(t)\|_p\leq \|\Psi_1(0)\|_p.
\end{equation}
\end{lemma}
If the rescaled ODE from \cref{eq:problem_1_rescaled} satisfies 
$\tilde{\mu}_0\geq 0$ and $R_p<1$, then we say that the rescaled ODE
is {\it dissipative in the $p$-norm}.
The proof of \cref{lem:pnormdecrease} can be found in \cref{app:error_analysis}. 
Our proof uses a slightly different technique than the literature,
which allows us to establish dissipativity conditions for $p\ne 2$.

Now we move to the error analysis due to truncation of the Carleman-Fourier linearization.
We first define the error due to truncation of the linearized system.
Consider the Carleman-Fourier linearization given in  \cref{eq:Lin_ODE}
of the rescaled ODE from \cref{eq:problem_1_rescaled}. 
We define the $k$th block component of the Carleman-Fourier truncation error to be  
\begin{equation}
    \eta_k =\Psi_k - \Psi^{(N)}_k, \quad k = 1,\dots,N.
\end{equation}
The Carleman-Fourier truncation error component $\eta_k$ captures the difference between
the $k$th block components of the solution of the exact and the truncated linearized 
ODEs respectively. The next theorem established a bound on the $p$-norm of $\eta_k$.

\begin{theorem}[Infinite time Carleman-Fourier error]\label{lemma:inf_time_error_bound} 
Let the rescaled ODE from \cref{eq:problem_1_rescaled} be dissipative in the 
$p$-norm as given in \cref{lemm:dissp}, then the $k$th component of the Carleman-Fourier truncation error
satisfies
\begin{equation}
    \|\eta_k\|_p\leq  (\|\Psi_1(0)\|_p)^{N+1}  \left(\frac{\|F_1\|_{\rm{row},{q}}}{\tilde{\mu}_0}\right)^{N+1-k}.
\end{equation}
\end{theorem}  

\begin{proof}
From the definition of $\eta_k$ we have
\begin{align}
\frac{d \eta_k(t)}{dt} &= \left(\mathbf{L}^{\rm T}\Psi\right)_k - \left(\mathbf{L}^{\rm T}_N\widetilde{\Psi}\right)_k\nonumber\\ 
&= B_k^{(0)}\Psi_k + B_{k+1}^{(1)}\Psi_{k+1} - B_k^{(0)}\widetilde{\Psi}_k - B_{k+1}^{(1)}\widetilde{\Psi}_{k+1}\nonumber\\
&= B_k^{(0)}\eta_k + B_{k+1}^{(1)}\eta_{k+1}\quad \text{for }  k\in [1,N-1],
\end{align}
where $\mathbf{L}^{\rm T}$ is the Carleman-Fourier matrix before the truncation. Now  for $k=N$, we have
\begin{equation}
\frac{d \eta_N(t)}{dt} = B_N^{(0)}\eta_N + B_{N+1}^{(1)}\Psi_{N+1}.   
\end{equation}
Starting from the last equation, since 
$\eta_k(0)={\bf 0}$ we have
\begin{equation}
\eta_N(t) = \int_0^t e^{B_N^{(0)}(t-s)}B_{N+1}^{(1)}\Psi_{N+1}(s) ds.
\end{equation}
Since $B_{k}^{(0)}$ in \cref{eq:carleman-assembly2} is a diagonal matrix,
we have 
\begin{equation}
\Re{((B^{(0)}_{k})_{ii})} \leq -k\tilde{\mu}_0,\quad \forall j \in \{1, \cdots, n^k\},     
\end{equation}
and 
\begin{equation}
\|e^{B^{(0)}_N(t-s)}\| \leq e^{-k\tilde{\mu}_0(t-s)}.    
\end{equation}

Now we can bound 
$\|\eta_N(t)\|_p$ as
\begin{align}
\label{eq:n_N}
\|\eta_{N}(t)\|_p &\le \big\|B_{N+1}^{(1)}\big\|_p\max_s \big\|\Psi_{N+1}(s)\big\|_p 
\int_0^t \big\|e^{B^{(0)}_N(t-s)}\big\|_p ds\nonumber\\
&\le  \big\|B_{N+1}^{(1)}\big\|_p\max_s \big\|\Psi_{N+1}(s)\big\|_p 
\int_0^t e^{-N\tilde{\mu}_0(t-s)}ds\nonumber\\
&\le \frac{N\big\|\widetilde{F}_1\big\|_p 
 \max_s\|\Psi_1(s)\|_p^{N+1}}{\tilde{\mu}_0 N} \nonumber\\
 &\le \frac{\big\|\widetilde{F}_1\big\|_p \|\Psi_1(0)\|_p^{N+1}}{\tilde{\mu}_0} .
\end{align}
Here we used  $\big\|B_{N+1}^{(1)}\big\|_p\leq N\big\|\widetilde{F}_1\big\|_p$, 
$\norm{\Psi_k}_p = (\norm{\Psi_1}_p)^k$, and \cref{lemm:dissp}. Now for $1\leq k\leq N-1$,
\begin{equation}
    \eta_k(t) = -\int_0^t e^{B_{k}^{(0)}(t-s)} \left(B_{k+1}^{(1)}\eta_{k+1}\right) ds.
\end{equation}
Now we may bound $\|\eta_k(t)\|_p$ as
\begin{align}
   \|\eta_k(t)\|_p &\leq \int_0^t \big\|e^{B_k^{(0)}(t-s)} \big\|_p \left\|B_{k+1}^{(1)}\eta_{k+1} \right \|_p ds \nonumber\\
    &\leq \int_0^t \big\|e^{-k\tilde{\mu}_0(t-s)}\big\| 
  \|B_{k+1}^{(1)}\|_p \|\eta_{k+1}\|_pds \nonumber\\
  &\leq \frac{\big\|\widetilde{F}_1\big\|_p \max_s\|\eta_{k+1}(s)\|_p}{\tilde{\mu}_0}.
\end{align}
We can now similarly bound $\|\eta_{k+1}\|_p$ in terms of 
$\|\eta_{k+2}\|_p$, and so on and so forth, till we obtain a bound in terms
of $\|\eta_N\|$, which is
\begin{align}
\max_t\|\eta_k(t)\|_p&\leq \frac{(\|\widetilde{F}_1\|_p)^{(N-k)} 
\max_t\|\eta_{N}(t)\|_p}{\tilde{\mu}_0^{N-k}}\nonumber\\
&\le \frac{(\|\widetilde{F}_1\|_p)^{N-k}}{\tilde{\mu}_0^{N-k}}\left(\frac{\|\widetilde{F}_1\|_p (\|\Psi_1(0)\|_p)^{N+1}}{\tilde{\mu}_0 }\right)\nonumber\\
&\le \frac{(\|\widetilde{F}_1\|_p)^{N+1-k}(\|\Psi_1(0)\|_p)^{N+1}}{\tilde{\mu}_0^{N+1-k}}\nonumber\\
&=\frac{(\|F_1\|_{\text{row},q})^{N+1-k}(\|\Psi_1(0)\|_p)^{N+1}}{\tilde{\mu}_0^{N+1-k}}.
\end{align}
Here we have used \cref{lemma:F_1_p_F_1_row_q} in the last equality.
\end{proof}

\subsubsection{Truncation error without dissipative conditions}
We now derive a similar truncation-error bound in the absence of dissipative conditions
on the ODE in \cref{eq:problem_1_rescaled}.
We first determine the time interval over which the norm of the solution of the 
rescaled nonlinear ODE in \cref{eq:problem_1_rescaled} 
is upper bounded with respect to the $p$-norm. Within this time interval interval, we then carry out an error analysis of the linearization and derive the conditions required to guarantee the stability of the resulting linearized ODE. The techniques used in this section are 
inspired by those used in
Ref.~\cite{chen2024carleman}.

\begin{lemma}[Upper Bounded in  $p$-Norm]
\label{lemma:finitetimedecreasingnorm}
Let $\{p,q\} \in [1,\infty)$ be such that $1/p+1/q=1$ and $r >1$. Consider the rescaled ODE from \cref{eq:problem_1_rescaled}  where  
\begin{equation}
    \Lambda_p := \max\left\{\|F_0\|_{\infty},\norm{{F}_1}_{{\rm row},q}\right\} \quad \text{and} \quad
     \norm{\Psi_1(0)}_p < 1/r.
\end{equation}
Then for $t\in [0,T_r]$,
\begin{equation}
    \|\Psi_1(t) \|_p \leq \frac{1}{r},
\end{equation}
where
\begin{equation}
    T_{r} := \frac{1}{\Lambda_p(1+1/r)}\ln \left(\frac{1}{r\|\Psi_1(0)\|_p}\right).
\end{equation}
\end{lemma}
\begin{proof} See \cref{app:error_analysis}.
\end{proof}

\begin{theorem}[Finite-time Carleman-Fourier error]
\label{thm:finitetimebound} Consider the rescaled ODE from \cref{eq:problem_1_rescaled}  with its Carleman-Fourier linearization given in  \cref{eq:Lin_ODE}. Let the  ODE be upper bounded on a finite time interval as in \cref{lemma:finitetimedecreasingnorm} i.e $\norm{\Psi_1(0)}_p < 1/r$ so that $\norm{\Psi_1(t)}_p \leq 1/r$ for $t\in[0,T_r]$ with $r>1$ and $T_r$ as given in \cref{lemma:finitetimedecreasingnorm}. Then, the error of the solution is given by $\eta_k =\Psi_k - \Psi^{(N)}_k$ at the order $k \in [N]$ due to Carleman-Fourier truncation $N\geq 2$ and the error vector $\eta =\begin{pmatrix}
    \eta_1, \eta _2, \cdots , \eta_N
\end{pmatrix}^{\rm T}$ is upper bounded by 
\begin{equation}
    \|\eta(t)\|_p \le \frac{1}{r}
    \left(\frac{e^{\left(\|F_0\|_{\infty} + \norm{F_1}_{{\rm row},q}\right)t}}{r}\right)^{N}.
\end{equation}
\end{theorem}

\begin{proof} 
From the definition of $\eta_k$ we have
\begin{align}
\frac{d \eta_k(t)}{dt} &= \left(\mathbf{L}^{\rm T}\Psi\right)_k - \left(\mathbf{L}^{\rm T}_N\widetilde{\Psi}\right)_k\nonumber\\ 
&= B_k^{(0)}\Psi_k + B_{k+1}^{(1)}\Psi_{k+1} - B_k^{(0)}\widetilde{\Psi}_k - B_{k+1}^{(1)}\widetilde{\Psi}_{k+1}\nonumber\\
&= B_k^{(0)}\eta_k + B_{k+1}^{(1)}\eta_{k+1}\quad \text{for }  k\in [1,N-1],
\end{align}
where $\mathbf{L}^{\rm T}$ is the Carleman-Fourier matrix before the truncation. This can be written as
\begin{equation}
     \frac{d\eta}{dt} = {\bf L}_N^{\rm T}\eta + v,
     \quad v = 
     \begin{bmatrix}
         0 \\ \vdots \\ 0 \\ B_{N+1}^{(1)}\Psi_{N+1}
     \end{bmatrix},
\end{equation}
with $\eta(0) = 0$. 
Therefore we have
\begin{equation}
    \eta(t) = \int_{0}^{t} e^{{\bf L}_N^{\rm T}(t-s)}v(s)ds.
\end{equation}
Taking the norm of the equation and by using properties of the norm, we get
\begin{equation}
    \norm{\eta(t)}_p \le \int_{0}^{t} \norm{e^{{\bf L}_N^{\rm T}(t-s)}}_p\norm{v(s)}_pds.
\end{equation}
We have
\begin{equation}
    \norm{e^{{\bf L}_N^{\rm T}(t-s)}}_p \le e^{\norm{{\bf L}_N^{\rm T}}_p(t-s)} \le 
    e^{N\left(\norm{\widetilde{F}_0}_p + \norm{\widetilde{F}_1}_p\right)(t-s)}
\end{equation}
where we have used \cref{lemma:LN_le_F}, and from $\big\|B_{N+1}^{(1)}\big\|_p\leq N\big\|\widetilde{F}_1\big\|_p$ we obtain
\begin{equation}
    \norm{v(s)}_p \le \norm{B_{N+1}^{(1)}}_p\norm{\Psi_{N+1}}_p \le 
    N\norm{\widetilde{F}_1}_p\left(\frac{1}{r}\right)^{N+1}.
\end{equation}
Substituting these in the inequality for $\norm{\eta(t)}_p$ yields
\begin{align}
    \norm{\eta(t)}_p &\le 
    N\norm{\widetilde{F}_1}_p\left(\frac{1}{r}\right)^{N+1}
    \int_{0}^{t}e^{N\left(\norm{\widetilde{F}_0}_p + \norm{\widetilde{F}_1}_p\right)(t-s)}ds \nonumber\\
    &= \frac{\norm{\widetilde{F}_1}_p}{\norm{\widetilde{F}_0}_p+\norm{\widetilde{F}_1}_p}
    \left(\frac{1}{r}\right)^{N+1}
    \left(e^{N\left(\norm{\widetilde{F}_0}_p + \norm{\widetilde{F}_1}_p\right)t}-1\right)\nonumber\\
    &\le \left(\frac{1}{r}\right)^{N+1}
    e^{N\left(\norm{\widetilde{F}_0}_p + \norm{\widetilde{F}_1}_p\right)t}\nonumber\\
    &= \frac{1}{r}
    \left(\frac{e^{\left(\norm{\widetilde{F}_0}_p + \norm{\widetilde{F}_1}_p\right)t}}{r}\right)^{N}.
\end{align}
The final result follows from \cref{lemma:F_1_p_F_1_row_q} and using $\|\widetilde{F}_0\|_{p}=\|F_0\|_{\infty}$.
\end{proof}

\begin{remark}
 The truncation bound in \cref{thm:finitetimebound} converges for large $N$ if $t < \frac{\ln{r}}{\norm{\widetilde{F}_0}_p + \norm{\widetilde{F}_1}_p}$, however, this time can occur outside the time interval in which the bound is valid i.e there exists values for $r$ such that $\frac{\ln{r}}{\norm{\widetilde{F}_0}_p + \norm{\widetilde{F}_1}_p} > T_r$. Therefore, the finite time Carleman-Fourier linearization is valid for $t\in[0,T_{\max}]$ where
\begin{align}
    T_{\max} &= \min\left\{T_r, \frac{\ln{r}}{\norm{\widetilde{F}_0}_p + \norm{\widetilde{F}_1}_p}\right\} \nonumber\\
    &= \min\left\{\frac{1}{\max\left\{\norm{F_0}_{\infty},
    \norm{F_1}_{{\rm row},q}\right\}(1+1/r)}\ln \left(\frac{1}{r\|\Psi_1(0)\|_p}\right), 
    \frac{\ln{r}}{\norm{F_0}_{\infty}+
    \norm{F_1}_{{\rm row},q}}\right\}.
\end{align}
\end{remark}

\subsection{Errors from truncation of the Taylor series}

In this section we derive bounds on errors due to truncation of the 
Taylor series to order $k$, equivalently due to replacement of $V$
in \cref{eq:V} by $V_k$. We begin by defining a quantity that tells us how the solution of a linear ODE grows. 
\begin{definition}\label{def:C(A)} For $p \in [1,\infty)$ and $A\in \mathbb{C}^{n\times n}$,
    \begin{equation} 
        C_p(A)=\sup_{t\in[0,T]} \norm{\exp(At)}_p.
    \end{equation}
\end{definition}
For $p=2$, we also use the notation $C_2(A) = C(A)$. Next, we bound the remainder due to truncation of the Taylor series.
\begin{lemma}[Taylor remainder bound]\label{lemma:V_ke^-Lh_bound}
   For \cref{eq:TaylorSystem} let $\|{\bf L}_N^{\rm T}\|_ph\leq 1$ and $me^2/(k+1)!\leq 1$. Then for all $j\leq m,$
    \begin{equation}
        \norm{\left(V_k^{j}-e^{{\bf L}_N^{\rm T}hj}\right)e^{-{\bf L}_N^{\rm T}hj} }_p  \leq \frac{(e-1)je^2}{(k+1)!},
    \end{equation}
where $V_k$ is a $k$th-order truncated Taylor series of $V = e^{{\bf L}_N^{\rm T}h}$, i.e., 
\begin{equation}
V_{k} = \sum_{i=0}^{k}\frac{({\bf L}_N^{\rm T}h)^i}{i!}.
\end{equation}
\end{lemma}
The proof can be found in \cref{app:error_analysis}. 
Finally, in the next lemma we use \cref{lemma:V_ke^-Lh_bound} to bound the error 
in the solution state at the final time step due to truncation of the Taylor series.
\begin{lemma}[Taylor truncation error]\label{lemma:taylor_error_bound}
     For ${\bf L}_N^{\rm T}, h, m, k, \Phi_{j}$ and $\Psi^{(N)}(jh)$ be as defined
     in \cref{eq:TaylorSystem} let $\|{\bf L}_N^{\rm T}\|_ph\leq 1$ and $me^2/(k+1)!\leq 1$. 
     Then for all $j\leq m$, 
    \begin{equation}
        \norm{\Phi_{j}-\Psi^{(N)}(jh)}_p \le \frac{(e-1)e^2j}{(k+1)!}
        C_p({\bf L}_N^{\rm T})\norm{\Psi^{(N)}(0)}_p.
    \end{equation}
\end{lemma}

\begin{proof}
\begin{align}
\norm{\Phi_{j}-\Psi^{(N)}(jh)}_p &= \norm{(V_k^{j}-e^{\bar{j}h{\bf L}_N^{\rm T}})\Psi^{(N)}(0)}_p \\ \nonumber
&\leq \norm{(V_k^{j}-e^{jh{\bf L}_N^{\rm T}})e^{-jh{\bf L}_N^{\rm T}}}_p\norm{e^{jh{\bf L}_N^{\rm T}}}_p\norm{\Psi^{(N)}(0)}_p\\ \nonumber
&\leq \frac{(e-1)je^2}{(k+1)!}C_p({\bf L}_N^{\rm T})\norm{\Psi^{(N)}(0)}_p.
\end{align}  
where we have used \cref{lemma:V_ke^-Lh_bound},  $\norm{V_k^{j}}_p\leq \norm{e^{hj{\bf L}_N^{\rm T}}}_p$ and the definition of $C_p({\bf L}_N^{\rm T})$.
\end{proof}

\section{Complexity of the algorithm}
\label{sec:complexity}
In this section, we derive the complexity of implementing 
the unitary operators $U_{\mathcal{C}}$, $U_{\mathcal{L}^{-1}}$
and $U_{\mathcal{B}}$ in 
\cref{eq:reducation_to_exp_of_unitary} 
by constructing explicit circuits for the same.
These circuits and complexities 
are used in our ODE solvers in computing an approximation of $f(\mathbf{x})$
by expectation value estimation as discussed in \cref{sec:Reduction to expectation value of a unitary operator}. The techniques employed in this section are closely related to those used in Refs.~\cite{krovi2023improved,costa2025further}.\\

We begin by constructing a block encoding of ${\bf L}_N^{\rm T}$ in \cref{eq:LN}.
This construction uses several preliminary lemmas that we prove in \cref{App:Supp_res}.

\begin{lemma}[Block encoding of ${\bf L}_N^{\rm T}$]\label{lemma: complexity_L_N_T} Given a real number $\alpha \geq \max_j|(G_0)_j|$, $\nu\in\mathbb{R}_+$ for the rescaling given in \cref{eq:rescaling_definition}, access to an oracle $U_{G_0}$ with the action $U_{G_0}\ket{0^{b_1}}\ket{j}=\ket{[(G_0)_j]}\ket{j}$ and a unitary $U_{G_1}$ that is a $(\beta,0)$-block encoding of $G_1$, an exact block encoding of the matrix ${\bf L}_N^{\rm T}$ defined in \cref{eq:Lin_ODE_with_zeros} can be constructed with scaling factor $\alpha_{{\bf L}_N^{\rm T}}=N(\alpha+\nu\beta)$ using one call each 
to $U_{G_0}$ and $U_{G_1}$.
\end{lemma}

\begin{proof}
Notice that ${\bf L}_N^{\rm T}$ can be expressed as
\begin{equation}
    {\bf L}_N^{\rm T} = \sum_{j=1}^{N} \ket{j}\bra{j}\otimes B_j^{(0)}
    + \sum_{j=1}^{N-1} \ket{j}\bra{j+1}
    \otimes B_{j+1}^{(1)},
\end{equation}
Our strategy to block encode ${\bf L}_N^{\rm T}$ is to construct block encoding
of the two terms in this expression individually and then to combine them
using \cref{lemma:block_encoding_arithmetic}.

Let us construct a block encoding for the first tem. 
Recall from \cref{eq:padding}, we have
\begin{align}
 B^{(0)}_{j} &=\mathds{1}^{\otimes N-j} 
    \otimes \left(i\widetilde{F}_0\otimes I^{\otimes (j-1)} + I\otimes i\widetilde{F}_0\otimes I^{\otimes(j-2)} + \dots + I^{\otimes (j-1)}\otimes i\widetilde{F}_0\right)\nonumber\\
    &= \mathds{1}^{\otimes N-j} 
    \otimes \left(\sum_{l=N-j+1}^{N} \text{SWAP}_{lN}
    \left(\mathds{1}^{\otimes N-1}\otimes i\widetilde{F}_0\right)
    \text{SWAP}_{lN}\right),
\end{align}
where $\text{SWAP}_{lN}$ swaps the $l$th subsystem with the $N$th subsystem;
here each subsystem is $n$-dimensional.
We define a controlled swap operator C-SWAP as
\begin{equation}
    \text{C-SWAP} = \sum_{l=1}^{N}\ket{l}\bra{l} 
    \otimes \text{SWAP}_{lN},
\end{equation}
It follows that
\begin{equation}
\label{eq:intermediateB0}
    W_0 := \sum_{l=1}^{N}\ket{l}\bra{l} 
    \otimes \left(\text{SWAP}_{lN}
    \left(\mathds{1}^{\otimes N-1}\otimes i\widetilde{F}_0\right)
    \text{SWAP}_{lN}\right) = \\
    \text{C-SWAP}
    \left(\sum_{l=1}^{N}\ket{l}\bra{l}\otimes \mathds{1}^{\otimes N-1}\otimes i\widetilde{F}_0\right)
    \text{C-SWAP}.
\end{equation}
The operator $W_0$ allows us to block encode 
$B_j^{(0)}$ by using 
\begin{equation}
    B_j^{(0)} = j\bra{+_j}W_0\ket{+_j}, \quad 
    \ket{+_j}:=\frac{1}{\sqrt{j}}\sum_{l=N-j+1}^{N}\ket{l}.
\end{equation}
However to reduce queries $U_{G_0}$, we will not construct individual block
encodings of $B_j^{(0)}$. Rather,
we will construct them in superposition by using the idea
of inequality testing.
We define an inequality-testing unitary operator 
$U_{>}$ acting on a space $\mathbb{C}^2\otimes \mathbb{C}^N \otimes \mathbb{C}^N$
as follows:
\begin{equation}
    U_{>}\ket{0}\ket{k}\ket{j} = \ket{k>j}\ket{k}\ket{j}
\end{equation}
where $\ket{k>j} = \ket{1}$ if $k>j$ and $\ket{k>j} = \ket{0}$ if $k\le j$.
Now observe that
\begin{equation}
    \ket{N-j}\bra{N-j} \otimes B_j^{(0)}= N\bra{0}\bra{+_N}U_{>} \left(
    \ket{0}\bra{0}\otimes\ket{N-j}\bra{N-j} \otimes
    W_0\right)\ket{+_N}\ket{0}.
\end{equation}
This equality follows from the fact that $\bra{0}U_{>}$ annihilates any term
in the summation of $W_0$ in \cref{eq:intermediateB0} for which $l\le N-j$.
Now by linearity, we get 
\begin{align}
    \sum_{j=1}^{N}\ket{N-j}\bra{N-j} \otimes B_j^{(0)}&= N\bra{0}\bra{+_N}U_{>} \left(
    \ket{0}\bra{0}\otimes\left(\sum_{j=1}^{N}\ket{N-j}\bra{N-j}\right) \otimes
    W_0\right)\ket{+_N}\ket{0}\nonumber\\
    &=N\bra{0}\bra{+_N}U_{>} \left(
    \ket{0}\bra{0}\otimes\mathds{1}_N \otimes
    W_0\right)\ket{+_N}\ket{0}
\end{align}
Finally, a unitary $\text{NOT}_N:\ket{j} \to \ket{N-j}$ can be used to obtain
\begin{align}
    \sum_{j=1}^{N}\ket{j}\bra{j} \otimes  B_j^{(0)}=N\left(\text{NOT}_N\bra{0}\bra{+_N}U_{>} \left(
    \ket{0}\bra{0}\otimes\mathds{1}_N \otimes
    W_0\right)\ket{+_N}\ket{0}\text{NOT}_N\right).
\end{align}
Moreover, an $(\alpha,0)$ block encoding of $W_0$ can be obtained with one query to
$U_{G_0}$. Consequently, an $(N\alpha,0)$ block
encoding of 
$\sum_{j=1}^{N}\ket{j}\bra{j} \otimes B_j^{(0)}$
can be constructed using a single query to $U_{G_0}$.

The construction of a block encoding of 
$\sum_{j=1}^{N-1}\ket{j}\bra{j+1} \otimes B_{j+1}^{(1)}$
follows a similar technique. 
We use 
\begin{multline}
 B^{(1)}_{j+1} =
 \mathds{1}^{\otimes N-j-1} \otimes \left(\sum_{l=N-j+1}^{N} \text{CYCLE}_{N-j,l}\text{SWAP}_{lN}\text{SWAP}_{l-1,N-1}
    \left(\mathds{1}^{\otimes N-2}\otimes \mathbf{e}_1\otimes i\widetilde{F}_1\right)
    \text{SWAP}_{lN}\text{SWAP}_{l-1,N-1}\right),
\end{multline}
where $\text{CYCLE}_{N-j,l}$ permutes subsystems $N-j,N-j+1,\dots,l$ in a cycle
so that $\mathbf{e}_1$ is positioned correctly in the front.
The block encoding of the desired term can be constructed using the operators
\begin{equation}
    \text{C-SWAP}^{(2)} = \sum_{l=1}^{N}\ket{l}\bra{l} 
    \otimes \text{SWAP}_{lN}\text{SWAP}_{l-1,N-1}
\end{equation}
and 
\begin{equation}
    \text{C-CYCLE} = \sum_{j=1}^{N}\sum_{l=1}^{N}\ket{j}\bra{j}\otimes \ket{l}\bra{l} 
    \otimes \text{CYCLE}_{jl}
\end{equation}
together with $U_{>}$ and $U_{G_1}$. This block encoding circuit makes one use of $U_{G_1}$.
The scaling factor for the
block encoding is $(N-1)\nu\beta$, since there are $N-1$ terms in the summation
and a $(\nu\beta,0)$ block encoding of $i\widetilde{F}_1$ can be 
constructed from $U_{G_1}$. 

Finally, using \cref{lemma:block_encoding_arithmetic}, we can obtain a
$(N\alpha + (N-1)\nu\beta,0)$ block encoding of ${\bf L}_N^{\rm T}$,
which makes one query to $U_{G_0}$ and $U_{G_1}$ each.

\end{proof}

The next lemma shows that the operator $\mathcal{A}$ defined therein can be used to block-encode $\mathcal{L}^{-1}$. 

\begin{lemma}[Block encoding of $\mathcal{L}^{-1}$]\label{lemma:complexity_L_inv} Let 

\begin{align}\label{eq:mathcal_M}
\mathcal{A} &= I-\mathcal{M} \nonumber\\
\mathcal{M} &= \sum_{j=0}^{m-1} \ket{j+1}_1\bra{j}  \otimes\, \mathcal{M}_{2}(I-\mathcal{M}_{1})^{-1}
+ \sum_{j=m}^{2m-2} \ket{j+1}\bra{j} \,\otimes\, I \nonumber\\
\mathcal{M}_1 &= \sum_{i=0}^{k-1} \ket{i+1}_2\bra{i}\otimes \frac{{\bf L}_N^{\rm T}h}{i+1} \nonumber\\
\mathcal{M}_2 &= \sum_{i=0}^{k} \ket{0}_2\bra{i} \otimes I. 
\end{align}
where the first register indexed by $i$ contains $\lceil\log _2(k+1)\rceil$ qubits, the second indexed by $j$ contains $\lceil1+\log _2m \rceil$ qubits  and the third contains $\lceil N\log _2(n)+\log_2(N)\rceil$ qubits. Then 
\begin{equation}
    {}_2\!\bra{0}\mathcal{A}^{-1}\ket{0}_2 = \mathcal{L}^{-1},
\end{equation}
for $\mathcal{L} = \mathcal{L}_{N,m,k}$ defined in \cref{eq:TaylorSystem}.
\end{lemma}

\begin{proof}
Following \cite{krovi2023improved} Section 4, 
for $\ket{l}_2\ket{v}_3$ an arbitrary standard basis state for the first and third register, 
\begin{align}
(I - \mathcal{M}_1)^{-1} \ket{l}_2\ket{v}_3 
&= \sum_{i=0}^k \mathcal{M}_1^i \ket{l}_2\ket{v}_3 \nonumber\\
&= \sum_{i=0}^{k-l} \ket{l + i}_2 \otimes 
   \frac{l! \,({\bf L}_N^{\rm T}h)^i}{(l + i)!} \ket{v}_3,
\end{align}
since $\mathcal{M}_1$ is nilpotent and $\mathcal{M}_1^{k+1}=0$. Then,
\begin{equation}\label{eq:M_2_M_1_action}
    \mathcal{M}_{2}(I-\mathcal{M}_{1})^{-1} \ket{l}_2\ket{v}_3
    = \ket{0}_2\otimes \sum_{i=0}^{k-l} \frac{l! \, ({\bf L}_N^{\rm T}h)^i}{(l+ i)!} \ket{v}_3 = \ket{0}_2\otimes W_{\ell,k} \ket{v}_3,
\end{equation}
where
\begin{equation}\label{def:W_l}
    W_{l,k} \coloneqq \sum_{i=0}^{k-l} \frac{l! \, ({\bf L}_N^{\rm T}h)^i}{(l + i)!}.
\end{equation}
For $l=0$, we get
\begin{equation}
    \mathcal{M}_{2}(I-\mathcal{M}_{1})^{-1} \ket{0}_2\ket{v}_3= \ket{0}_2\otimes W_{0,k} \ket{v}_3 = \ket{0}_2 V_k\ket{v}_3.
\end{equation}
Now observe that
\begin{equation}
\mathcal{M}\ket{0}_2\ket{j}_1\ket{v}_3=
\begin{cases}
\ket{0}_2\ket{j+1}_1V_k\ket{v}_3, & 0 \le j \le m-1 \\[2pt]
\ket{0}_2\ket{j+1}_1\ket{v}_3, & m-1 < j \le 2m-2 \\[2pt]
0, & j=2m-1 
\end{cases}.
\end{equation} 
By comparing with \cref{eq:TaylorSystem}, we observe that
\begin{equation}
    \mathcal{M}\ket{0}_2\ket{j}_1\ket{v}_3 = \ket{0}_2(\mathds{1} -
    \mathcal{L}_{N,m,k})\ket{j}_1\ket{v}_3,
\end{equation}
which holds for any $\ket{j}_1$ and any $\ket{v}_3$.
In effect, provided the second register is in $\ket{0}_2$ state, $\mathcal{M}$ 
acts trivially on the second register and acts as $(\mathds{1} - \mathcal{L}_{N,m,k})$
on the first and the third register regardless of their state. Consequently, 
$\mathcal{M}^j$ acts as $(\mathds{1} - \mathcal{L}_{N,m,k})^j$ on the first and 
the third register
for any non-negative integer $j$.
Now observe that $\mathcal{M}^{2m}=0$, i.e. $\mathcal{M}$ is nilpotent, so 
$\mathcal{A} = (\mathds{1} - \mathcal{M})$ is invertible and its inverse is given by
$\mathcal{A}^{-1} = (\mathds{1} - \mathcal{M})^{-1} = \sum_{j=0}^{2m-1}\mathcal{M}^j$.
Therefore,
\begin{align}
    \mathcal{A}^{-1}\ket{0}_2\ket{j}_1\ket{v}_3 &= 
    (\mathds{1}-\mathcal{M})^{-1}\ket{0}_2\ket{j}_1\ket{v}_3 \nonumber\\
    &= \ket{0}_2\left(\mathds{1}-(\mathds{1} 
    - \mathcal{L}_{N,m,k})\right)^{-1}\ket{j}_1\ket{v}_3 \nonumber\\
    &= \ket{0}_2 \mathcal{L}_{N,m,k}^{-1}\ket{j}_1\ket{v}_3,
\end{align}
which implies
\begin{equation}
    {}_2\!\bra{0}\mathcal{A}^{-1}\ket{0}_2 = \mathcal{L}_{N,m,k}^{-1}.
\end{equation}
\end{proof}

In contrast to previous work on Carleman linearisation which block-encodes $\mathcal{A}$ and employs the QLSA to perform inversion and obtain the normalized solution state \cite{krovi2023improved}, our approach directly block-encodes $\mathcal{A}^{-1}$ since we are concerned with the task of extracting classical quantities from the solution state which can be reduced to computing the expectation value of a unitary operator. 

\begin{lemma}[Block encoding of $\mathcal{A}^{-1}$]\label{lemma:complexity_A_inv} Let $h\|{\bf L}_N^{\rm T}\|_2\leq 1$, $me^2/(k+1)!\leq 1$ and 
\begin{equation}
    \tau\leq  \frac{1}{4e^2mC({\bf L}_N^{\rm T})\sqrt{k+1}}.
\end{equation}
Then, given a real number $\alpha \geq \max_j|(G_0)_j|$, $\nu\in\mathbb{R}_+$ for the rescaling given in \cref{eq:rescaling_definition}, access to an oracle $U_{G_0}$ with the action $U_{G_0}\ket{0^{b_1}}\ket{j}=\ket{[(G_0)_j]}\ket{j}$ and a unitary $U_{G_1}$ that is a $(\beta,0)$-block encoding of $G_1$, a $(\alpha_{\mathcal{A}^{-1}}, \mathcal{E}(\sigma))$-block encoding of $\mathcal{A}^{-1}$ for $\mathcal{A}$ defined in \cref{eq:mathcal_M} can be constructed using 
\begin{equation}
\mathcal{O}\left(Nk^{7/2} T(\alpha+\nu\beta)C({\bf L}_N^{\rm T}) \cdot \log k \log \left(\frac{1}{\sigma}\right)\log \left(\frac{1}{\tau }\right)\right),
\end{equation}
calls to $U_{G_0}$ and $U_{G_1}$, where
\begin{align}
\label{eq:alpha_E_sigma}
   \alpha_{\mathcal{A}^{-1}} &= 8e^2mC({\bf L}_N^{\rm T})\nonumber\\
    \mathcal{E}(\sigma) &= \sigma +8e^4m^2\sqrt{k+1}\tau C({\bf L}_N^{\rm T})^2.
\end{align}
\end{lemma} 

\begin{proof} The block encoding can be constructed using the following steps. 

\begin{enumerate}
    
    \item A $\left(\frac{\alpha_{{\bf L}_N^{\rm T}}h}{i+1},0\right)$-block encoding of $\ket{i+1}\bra{i}\otimes\frac{{\bf L}_N^{\rm T}h}{i+1}$ can be constructed using $\mathcal{O}(1)$ calls to a $(\alpha_{{\bf L}_N^{\rm T}}, 0)$-block encoding of ${\bf L}_N^{\rm T}$ from \cref{lemma:block_encoding_arithmetic} and \cref{lemma:block_encoding_sparse}.

\item A $(\mu,0)$-block encoding of $I-\mathcal{M}_1$ can be constructed  using $\mathcal{O}(k)$ calls to a $(\alpha_{{\bf L}_N^{\rm T}}, 0)$-block encoding of ${\bf L}_N^{\rm T}$ from \cref{lemma:block_encoding_arithmetic} and with $\mu=1+\sum_{i=0}^{k-1}\frac{\alpha_{{\bf L}_N^{\rm T}}h}{i+1} =1+ \alpha_{{\bf L}_N^{\rm T}}hH_k$ where $H_k$ is the $k^{\text{th}}$ harmonic number.
    
    \item  Let $\omega\geq \|(I-\mathcal{M}_1)^{-1}\|$. Using \cref{lemma:block_encoding_inversion} a $\left(2\omega,\tau\right)$-block encoding of $(I-\mathcal{M}_1)^{-1}$ can be constructed using $\mathcal{O}(\omega\mu\log \frac{1}{\tau})$ queries to the block encoding of $I-\mathcal{M}_1$.

\item  A $(\sqrt{k+1},0)$-block encoding of $\mathcal{M}_2$  can be constructed with no calls to $U_{G_0}$ or $U_{G_1}$ from \cref{lemma:block_encoding_sparse}.

    \item A $(2\omega \sqrt{k+1}, \sqrt{k+1}\tau)$- block encoding of $\mathcal{M}_2(I-\mathcal{M}_1)^{-1}$ can be constructed using $\mathcal{O}(1)$ query to the block encoding of $(I-\mathcal{M}_1)^{-1}$ and $\mathcal{M}_2$.

\item A $(2\omega \sqrt{k+1}, \sqrt{k+1}\tau)$-block encoding of $\sum_{j=0}^{m-1} \ket{j+1}\bra{j} \otimes\mathcal{M}_2(I-\mathcal{M}_1)^{-1}$ can be constructed using $\mathcal{O}(1)$ call to the block encoding of $\mathcal{M}_2(I-\mathcal{M}_1)^{-1}$ from \cref{lemma:block_encoding_arithmetic}
 and \cref{lemma:block_encoding_sparse}.

\item A $(2(1+\omega \sqrt{k+1}), \sqrt{k+1}\tau)$-block encoding of $\mathcal{A}$ can be constructed using $\mathcal{O}(1)$ call to the block encoding of $\mathcal{M}_2(I-\mathcal{M}_1)^{-1}$ from \cref{lemma:block_encoding_arithmetic}
 and \cref{lemma:block_encoding_sparse}.

 \item From the assumptions stated in the Lemma we have $\sqrt{k+1}\tau\leq  \frac{1}{2 \|\mathcal{A}^{-1}\|}$. Then with $\phi \geq \|\mathcal{A}^{-1}\|$, a $(4\phi, \sigma + 2\sqrt{k+1}\tau \|\mathcal{A}^{-1}\|^2)$ block encoding of $\mathcal{A}^{-1}$ can be constructed using $\mathcal{O}((1+\omega \sqrt{k+1})\phi \log\frac{1}{\sigma })$ calls to the block encoding of $\mathcal{A}$ from \cref{lemma:block_encoding_inversion_with_error}.
\end{enumerate}

The lemma is proven by adding the complexities in step 1 - 8 to obtain the total number of calls to $U_{G_0}$ and $U_{G_1}$ of
\begin{equation}
\mathcal{O}\left((1+k\sqrt{k+1})e^2mC({\bf L}_N^{\rm T}) \log\frac{1}{\sigma } \cdot k(1+\alpha_{{\bf L}_N^{\rm T}}hH_k) \log \frac{1}{\tau}\cdot k\right),
\end{equation}
which follows from choosing $\phi  = 2e^2mC({\bf L}_N^{\rm T})$ from \cref{lemma:A_norm},  $\|\mathcal{A}^{-1}\|^2\leq4e^4m^2C({\bf L}_N^{\rm T})^2$ from \cref{lemma:A_norm}, $\omega  = k$ from \cref{lemma:I-M_1_norm} and \cref{lemma: complexity_L_N_T}, and using
$mh = T$.

\end{proof}

\begin{remark}\label{remark:block_enc_L}
    Since \cref{lemma:complexity_L_inv} establishes that $\mathcal{A}^{-1}$ is a $(1,0)$-block encoding of $\mathcal{L}^{-1}$, the circuit described in \cref{lemma:complexity_A_inv} is a $(\alpha_{\mathcal{L}^{-1}},\mathcal{E})$-block encoding of $\mathcal{L}^{-1}$ with $\alpha_{\mathcal{L}^{-1}}=\alpha_{\mathcal{A}^{-1}}$.
\end{remark}

\begin{lemma}[Initial state preparation]\label{lemma:complexity_inital_state} Given $\nu\in\mathbb{R}_+$ for the rescaling defined in \cref{eq:rescaling_definition}, $\widetilde{\gamma}:=\|e^{i{\bf u}_0}\|$ and a unitary $U_{u_0}$ which prepares the state $\ket{e^{i{\bf u}_0}}$ encoding the normalized vector
$e^{i{\bf u}_0}/\widetilde{\gamma}$ in its amplitudes, a $\lceil 1+\log_2m+\log_2N+N\log_2n\rceil$-qubit unitary $U_\mathcal{B}$, which prepares the state $\ket{\mathcal{B}_{N,m}}$ encoding the normalized vector $\mathcal{B}_{N,m}/\alpha_{\mathcal{B}}$ in its amplitudes for $\mathcal{B}_{N,m}$ defined in \cref{eq:linear_algebraic_system} and $\alpha_{\mathcal{B}}:=\|\Psi^{(N)}(0)\|$ for $\Psi^{(N)}$ given in \cref{eq:Lin_ODE} can be constructed using $\mathcal{O}(N)$ calls to $U_{u_0}$.
\end{lemma}

\begin{proof}  Since from the definition of the rescaling, \cref{eq:rescaling_definition}, $e^{i{\bf x}}=e^{i{\bf u}}/\nu$, we have $U_{u_0}\ket{0}=\ket{e^{i{\bf x}_0}}$. \\

We first construct a $\lceil\log_2N+N\log_2n\rceil$ qubit unitary, $U_{S}$, which prepares a state proportional to $\ket{\Psi^{(N)}(0)}$,
\begin{equation}
   U_{S}\ket{0}=\frac{\ket{\Psi^{(N)}(0)}}{\|\Psi^{(N)}(0)\|} = \frac{1}{\sqrt{Q}}
\sum_{j=1}^{N}\ket{j}\ket{\Psi_1(0)}^{\otimes j}\ket{0}^{\otimes N-j},
\end{equation}
where $Q=\sum_{j=1}^N\gamma^{2j} = \frac{\gamma^2(1-\gamma^{2N})}{1-\gamma^2}$ for $\gamma=\|e^{i\textbf{x}_0}\|=\widetilde{\gamma}/\nu$.\\

Following \cite{liu2021efficient} Lemma 5, $U_{S}$ can be constructed using $\mathcal{O}(N)$ calls to $U_{u_0}$ by initializing a $\log_2 N$ qubit ancilla in 
a state proportional $\sum_{j=1}^{N}\gamma^{j}\ket{j}$
and applying controlled rotations on the target register of the form
$\ket{0}\bra{0} \otimes \mathds{1} + \ket{1}\bra{1}\otimes U_{u_0}^l$.\\

Let $U_{\rm idx}$ be a $\lceil1+\log_2m\rceil$ qubit unitary with the action $U_{\rm idx}\ket{x} = \ket{0}$. Then the $\lceil1+\log_2m+\log_2N+N\log_2n\rceil$-qubit unitary $U_{\rm idx}\otimes U_{S}$ prepares $\ket{\mathcal{B}_{N,m}}$ and can be constructed using 1 call to $U_S$.
\end{proof}

\begin{remark}\label{remark:rescaling}
If $\gamma = \|e^{i{\bf x}_0}\|>1$ then as $N$ grows $\alpha_{\mathcal{B}}=\|\Psi^{(N)}(0)\| = \sqrt{\frac{\gamma^2(1-2\gamma^{2N})}{1-\gamma^2}}=\mathcal{O}(\gamma^{N})$. However, from \cref{eq:rescaling_norm_input} the rescaling paramter $\nu$ from \cref{eq:rescaling_definition} can be chosen such that $\gamma\le1/\sqrt{2}$ and therefore $\alpha_{\mathcal{B}}\leq \frac{\gamma}{\sqrt{1-\gamma^2}} =\mathcal{O}(1)$ as $N$ grows.
\end{remark}

\begin{lemma} [Fourier coefficient preparation]\label{lemma:fourier_coefficient_preperation}
Let the components of the  vector ${\bf d}\in\mathbb{C}^D$ for $D=\frac{n^{K+1}-n}{n-1}$ be denoted $d_{\bf j}$ for ${\bf j} \in \mathcal{S}_K =\bigcup_{l=1}^K\mathbb{Z}^n_{l,+}$. Suppose the $D$ nonzero components of the vector ${\bf c}\in \mathbb{C}^{Nn^N}$ defined in \cref{eq:c_def} are given by $c_{\bf j}=\nu^{\abs{{\bf j}}}d_{\bf j}$ for ${\bf j} \in \mathcal{S}_K$ and some $\nu\in \mathbb{R}_+$. Then, given a unitary $U_{d}$ that prepares the state $\ket{d}$ encoding the normalised vector ${\bf d/}\|{\bf d}\|$ in its amplitudes, we can construct a unitary circuit $U_\mathcal{C}$ such that $\mathcal{C}^{\rm T} = \alpha_\mathcal{C}(\!\bra{0}U_\mathcal{C}^\dagger)$ for $\mathcal{C}$ defined in \cref{eq:mathcal_C} using $\mathcal{O}(1)$ calls to $U_d$ and with 
\begin{equation}
    \alpha_\mathcal{C}=\frac{\max{\{\nu,\nu^K\}} \|{\bf d}\|}{\sqrt{m}}
    = \mathcal{O}\left(\frac{(1+\nu^K) \|{\bf d}\|}{\sqrt{m}}\right).
\end{equation}
\end{lemma}

\begin{proof}

Let $M\in\mathbb{R_+}^{(Nn^N\times D)}$ be the rectangular matrix such that ${\bf c}=M{\bf d}$. We can construct an $(\alpha_M,0)$-block encoding of $M$ using $\nu$ and additional single and two qubit gates. Using standard properties of the operator 2 norm we can choose $\alpha_M = \max{\{\nu,\nu^K\}}$ so that
\begin{equation}  \bra{0^a}U_MU_d\ket{0}\ket{0^a}=\frac{{\bf c}}{\max{\{\nu,\nu^K\}}\|{\bf d}\|}.
\end{equation}

The circuit $U_\mathcal{C}$ can then be constructed to prepare the state $\frac{\sqrt{m}}{\alpha_M\|{\bf d}\|}\mathcal{C} $ using one call to $U_M$ and $U_d$. This is achieved by preparing an index register $\frac{1}{\sqrt{m}}\sum_{t=m}^{2m-1}\ket{t}$ and using controlled two-qubit gates to map this index register into the correct offset of the data register. We therefore have $\alpha_\mathcal{C}= \max{\{\nu,\nu^K\}} \|{\bf d}\|/\sqrt{m}$.

\end{proof}

The last piece we need for constructing our algorithm is the
following lemma that states the complexity of estimating 
$\expval{0|U|0}$ given access to the controlled unitary $cU$.
\begin{lemma}[Complexity of $\bra{0}U\ket{0}$ {\cite[Theorem~5]{alase2022tight}}]
\label{lem:readout}
    Given access to $cU$ for a unitary $U$, the expectation value $\expval{0|U|0}$ 
    can be estimated to an additive error $\epsilon$ with high probability 
    by making $\mathcal{O}(1/\epsilon)$ uses of $cU$. 
\end{lemma}

\section{Final results}
\label{sec:final-results}
In this section we combine the results from \cref{sec:error-analysis} and
\cref{sec:complexity} to derive algorithms and their complexity
for solving the Fourier ODE problem under dissipative conditions and in the non-dissipative 
regime. 

\subsection{Algorithm and complexity under dissipative conditions}
We first state a lemma that discusses the conditions under which the linearized ODE is stable. The proof can be found in \cref{app:rescaling_C(L)<1}. We then use this result to state and prove the final theorem for the case where the ODE in \cref{eq:raw_problem_1} is dissipative.

\begin{lemma}[Stability of the linearized ODE]\label{lemma:rescaling_C(L)<1} Let $\{p,q\} \in [1,\infty)$ be such that $1/p+1/q=1$ and for \cref{prob:problem_statement_formal} suppose that $G_0, G_1$ and ${\bf u}_0$  satisfy  
    \begin{equation}
        \tilde{\mu}_0 := \min_j \Im{(G_0)_j}\ge 0 \quad \text{and} \quad
         R_p := \frac{\|G_1\|_{{\rm row},q} \|e^{i{\bf u}_0}\|_p}{\tilde{\mu}_0}< \min \left\{1,\frac{\|e^{i{\bf u}_0}\|_p}{\|e^{i{\bf u}_0}\|_2}\right\},
    \end{equation}  
Then the rescaling of the variable ${\bf u}$ given in \cref{eq:rescaling_definition} with $\nu =\|e^{i{\bf u}_0}\|_p/R_p$  produces a rescaled differential equation with
\begin{equation}
    C({\bf L}_N^{\rm T})\leq 1 \quad \text{and}\quad \gamma:=\|e^{i{\bf x}_0}\|<1,
\end{equation}
where ${\bf L}_N^{\rm T}$ is defined in \cref{eq:LN} on the rescaled problem. 
\end{lemma}

\begin{remark} An improved stability condition $\nu\leq \tilde{\mu}_0/\|G_1\|_{\rm row, \infty}$ follows by adapting the proof of \cref{lemma:rescaling_C(L)<1} at \cref{eq:sum_off_diag_blocks_bound} to use $p=1$. The condition in \cref{lemma:rescaling_C(L)<1} has been used to simplify the statement of \cref{theorem:dissipative_result}.
\end{remark}

We now state and prove the final theorem for the case where the ODE is dissipative. 
\begin{theorem}[Fourier ODE solver]\label{theorem:dissipative_result}      
For \cref{prob:problem_statement_formal}, let $\{p,q\} \in [1,\infty)$ be such that $1/p+1/q=1$ and suppose $G_0, G_1$ and ${\bf u}_0$  satisfy  
    \begin{equation}
        \tilde{\mu}_0 := \min_j \Im{(G_0)_j}\ge 0 \quad \text{and} \quad
         R_p := \frac{\|G_1\|_{{\rm row},q} \|e^{i{\bf u}_0}\|_p}{\tilde{\mu}_0}< \min \left\{1,\frac{\|e^{i{\bf u}_0}\|_p}{\|e^{i{\bf u}_0}\|_2}\right\}
    \end{equation}
for some given $\tilde{\mu}_0$ and $R_p$.
Suppose we are also given $\|G_1\|_{{\rm row},q}$,
$\|e^{i{\bf u}_0}\|_p$, and $\|{\bf d}\|_q$ in addition to
the inputs specified in \cref{prob:problem_statement_formal}.
Then a quantum circuit can be constructed 
that returns $g({\bf u}(T))$ for any $T\in \mathbb{R}_+$ 
to an additive accuracy $\epsilon\in \mathbb{R}_+$ by making 
\begin{align}
\label{eq:complexityG0G1-dissipative}
\mathcal{O}\left( \frac{1}{\epsilon}  
N^{3/2}
k^{7/2}(T\alpha)^{3/2} z\left(1+\frac{\tilde{\mu}_0\beta}{\alpha\|G_1\|_{\rm row, q}}\right)\cdot \log k \log^2 \left(\frac{(NT\alpha)^{3/2}\sqrt{k}z}{\epsilon} \right)\right) &\quad \text{calls to } U_{G_0}, U_{G_1} \nonumber \\
\mathcal{O}\left(\frac{N^{3/2}\sqrt{T\alpha}z}{\epsilon}\right)  &\quad \text{calls to } U_{u_0} \nonumber\\
\mathcal{O}\left(\frac{\sqrt{NT\alpha}z}{\epsilon}\right) &\quad \text{calls to } U_{d},
\end{align}
where 
\begin{align}
N &= \mathcal{O}\left(\frac{\log(Ks\|{\bf d}\|_q/\epsilon)}{\log(R_p^{-1})}\right) \nonumber\\
  k &= \mathcal{O} \left( \log_2 \left(\frac{NT\alpha z}{\epsilon }\right)\right)\nonumber\\
     z &=\mathcal{O}\left( \frac{\|{\bf d}\|s R_p}{\sqrt{1-R_p^2}}\right)\nonumber\\
     s &= \mathcal{O}\left(1+\left(\frac{\tilde{\mu}_0}{\|G_1\|_{\rm row,q}}\right)^K\right).
\end{align}

\end{theorem}

\begin{proof}The proof is divided in four parts: rescaling, the algorithm, its correctness and the complexity. \\

{\noindent \bf Rescaling:}
For the rescaling of the variable ${\bf u}$ given in \cref{eq:rescaling_definition} parametrised by $\nu\in \mathbb{R}_+$, choose $
\nu =\frac{\norm{e^{i{\bf u}_0}}_p}{R_p}=\frac{\tilde{\mu}_0}{\|G_1\|_{\rm row,q}}$. Then the rescaled differential equation in the variable ${\bf x}(t)$ given in \cref{eq:problem_1_rescaled} satisfies $C({\bf L}_N^{\rm T})\leq 1$  and $\gamma=R_p<1$ from \cref{lemma:rescaling_C(L)<1}.\\

We then construct a quantum algorithm to estimate $f({\bf x}(T))$ for $f$ given in \cref{eq:rescaled_output_function} since $f({\bf x}(T))= g({\bf u}(T))$.\\

{\noindent \bf Algorithm:}
The algorithm for estimating $f({\bf x}(T))$ is as follows. 
\begin{enumerate}
    \item Select the parameters of the algorithm, $N,k,\sigma,\tau,\delta,m$;

\begin{align}\label{eq:dissipative_parameters}
 N&= \left\lceil
\frac{
\log \frac{4Ks\|{\bf d}\|_q}{\epsilon}
}{
\log R_p^{-1}
} \right\rceil
 \quad \nonumber\\ 
k &= \max\left\{\left\lceil \log_2 \left( \frac{4e^3}{\epsilon} s\norm{{\bf d}}m \frac{R_p}{\sqrt{1-R_p^2}} \right)\right\rceil , \lceil \log_2 me^2  \rceil\right\} \nonumber\\ 
\sigma &= c_1-c_2\tau \nonumber\\ 
\tau &= \min \left\{\frac{c_1}{1+c_2}, \frac{1}{4e^2m\sqrt{k+1}}\right\}\nonumber\\ 
    c_1 &= \frac{\epsilon\sqrt{m}}{4\|{\bf d}\|s }\frac{\sqrt{1-R_p^2}}{R_p}\nonumber\\  
c_2 &=  8e^4m^2\sqrt{k+1}\nonumber\\ 
    \delta &= 
    \frac{\epsilon}{\sqrt{m}\|{\bf d}\|s }\frac{\sqrt{1-R_p^2}}{R_p} \nonumber\\ 
    m &= \lceil TN\left(\alpha + \tilde{\mu}_0 \right)\rceil \nonumber\\
    s &= \max \left\{\frac{\tilde{\mu}_0}{\|G_1\|_{\rm row,q}}, 
    \left(\frac{\tilde{\mu}_0}{\|G_1\|_{\rm row,q}}\right)^K\right\}.
\end{align}

    \item Construct a $(\alpha_{\mathcal{L}^{-1}}, \mathcal{E} )$-block encoding 
$U_{\mathcal{L}^{-1}}$ of $\mathcal{L}^{-1}$ using \cref{lemma:complexity_A_inv} and \cref{lemma:complexity_L_inv}.

\item Construct a $d= \lceil 1+\log_2(m) + \log_2(N)+N\log_2(n)\rceil$-qubit unitary $U_\mathcal{B}$ acting on $d$-qubit register 1
    such that $\mathcal{B}_{N,m} = \alpha_{\mathcal{B}} U_\mathcal{B} \ket{0}_1$ using \cref{lemma:complexity_inital_state}.
    
\item Construct a unitary $U_\mathcal{C}$ such that 
    $\mathcal{C}^{\rm T} = \alpha_\mathcal{C}({}_1\!\bra{0}U_\mathcal{C}^\dagger)$ for $\mathcal{C}$ defined in \cref{eq:mathcal_C} using \cref{lemma:fourier_coefficient_preperation}.
    \item Use \cref{lem:readout} to compute $w$, a
    $\delta$-additive approximation to $\bra{0}U_\mathcal{C}^\dagger U_{\mathcal{L}^{-1}} U_\mathcal{B}\ket{0}$,
    where $\ket{0} = \ket{0}_1\ket{0^a}$.
    \item Return $\alpha_\mathcal{C}  \alpha_{\mathcal{L}^{-1}}  \alpha_\mathcal{B}  w$ as the estimate for
    $f({\bf x}(T))$.
\end{enumerate}

{\noindent \bf Correctness (error analysis):}
There are four sources of error: 1. Error from
truncation the Koopman variables,
2. Error from truncation of the Taylor series, 3. Error in the block encoding of $\mathcal{L}^{-1}$ and 4. Error 
in calculating the expectation value. 
In the notation introduced in \cref{sec:approach}, 
the total error in the approximation of $f({\bf x}(T))$ is bounded by the sum of these four sources;
\begin{multline} \underbrace{\abs{f({\bf x}(T))-\alpha_\mathcal{C} \alpha_{\mathcal{L}^{-1}} \alpha_\mathcal{B} w}}_{\text{total error}} \le \underbrace{\abs{f({\bf x}(T))- {\bf c}\cdot \Psi^{(N)}(T)}}_{\text{Koopman truncation error}} \\
    + \underbrace{\abs{ {\bf c}\cdot \Psi^{(N)}(T) - {\bf c}\cdot \Phi_m}}_{\text{Taylor truncation error}} \\
    + \underbrace{\abs{ {\bf c}\cdot \Phi_m - \alpha_\mathcal{C} \alpha_{\mathcal{L}^{-1}}\alpha_\mathcal{B} \bra{0}U_\mathcal{C}^\dagger U_{\mathcal{L}^{-1}} U_\mathcal{B}\ket{0}}}_{\text{block-encoding error}} \\
    + \underbrace{\abs{\alpha_\mathcal{C} \alpha_{\mathcal{L}^{-1}}\alpha_\mathcal{B} \bra{0}U_\mathcal{C}^\dagger U_{\mathcal{L}^{-1}} U_\mathcal{B}\ket{0}
    -\alpha_\mathcal{C}\alpha_{\mathcal{L}^{-1}}\alpha_\mathcal{B} w}}_{\text{expectation value calculation error}}.
\end{multline}

We now bound each error term. 
\begin{enumerate}
    \item Koopman truncation error: We can ensure the Koopman truncation error 
    stays below $\epsilon/4$ by choosing sufficiently large $N$ as follows. Using \cref{eq:f^N} we can write
    \begin{align}
    \abs{f({\bf x}(T))- {\bf c}\cdot \Psi^{(N)}(T)}
    &= \abs{\sum_{l=1}^{K} {\bf c}_l\cdot(\Psi_l(T)-\Psi^{(N)}_l(T))}\nonumber \\
    &\le \sum_{l=1}^{K}\norm{{\bf c}_l}_q\norm{\eta_l(T)}_p.
    \end{align}
    Now we can use \cref{lemma:inf_time_error_bound}, which is valid since $\tilde{\mu}_0$ and $R_p$ are invariant under the rescaling  
    \begin{equation}
     \sum_{l=1}^{K}\norm{{\bf c}_l}_q\norm{\eta_l(T)}_p  \leq  
     \norm{{\bf c}}_q(\|e^{i{\bf x}_0}\|_p)^{N+1} \sum_{l=1}^{K}  \left(\frac{\|F_1\|_{\rm{row},{q}}}{\tilde{\mu}_0}\right)^{N+1-l}  \leq \frac{\epsilon}{4}.  
    \end{equation}
    To ensure that the truncation error stays below $\epsilon/4$, it suffices to choose
    $N$ according as
    \begin{equation}
     (\|e^{i{\bf x}_0}\|_p)^{N+1}\leq \frac{\epsilon}{4K\|{\bf c}\|_q}   \Longrightarrow
N \ge
\left\lceil
\frac{
\log \frac{4Ks\|{\bf d}\|_q}{\epsilon}
}{
\log R_p^{-1}
}
\right\rceil.
    \end{equation}
Here we have used
$\frac{\tilde{\mu}_0}{\|F_1\|_{\rm row,q}} = \frac{\tilde{\mu}_0}{\nu|\|G_1\|_{\rm row,q}}=1$, $\|e^{i{\bf x}_0}\|_p=\frac{1}{\nu}\|e^{i{\bf u}_0}\|_p=R_p$ and the following bound on $\|{\bf c}\|_q$ which follows from the definition of the $q$-norm,
 \begin{equation}
    \|{\bf c}\|_q^q=\sum_{{\bf j}} \abs{ c_{\bf j}}^q=\sum_{{\bf j}}   \nu^{q\abs{{\bf j}}} \abs{d_j}^q \leq \max_{\bf j} \nu^{q\abs{{\bf j}}}\sum_{\bf j} \abs{d_{\bf j}}^q = \left(s\|{\bf d}\|_q\right)^q.
    \end{equation}

    \item Taylor truncation error: We can ensure the Taylor truncation  error 
    stays below $\epsilon/4$ by choosing sufficiently large Taylor order $k$ as follows. 
    Since $\alpha \geq \norm{F_0}_\infty$  and $\norm{F_1}_{\rm row,q}=\nu\norm{G_1}_{\rm row,q}=\tilde{\mu}_0$, our choice of $h = T/m$ satisfies 
    \begin{align}
 h \leq \frac{1}{N\left(\alpha + \tilde{\mu}_0 \right)} \le \frac{1}{N\left(\norm{F_0}_\infty + \norm{F_1}_{\rm row,q}\right)}.
\end{align}
We therefore have $\norm{{\bf L}_N^{\rm T}h}\le 1$.
    Furthermore, since $k \ge \log_2(me^2)$, we have
\begin{align}
   me^2 \le 2^k \implies m\leq \frac{(k+1)!}{e^2},
\end{align}
where we used $(k+1)!\geq 2^{k}$. 
Then by using \cref{lemma:taylor_error_bound}, with \cref{eq:f^N,eq:f_n_hat}, we get
\begin{align}
\abs{{\bf c}\cdot \Psi^{(N)}(mh) - {\bf c}\cdot \Phi_m}
&\le 
\norm{{\bf c}}\norm{\Psi^{(N)}(mh)-\Phi_m}\nonumber\\
&\le
\norm{{\bf c}}\frac{(e-1)e^2m}{(k+1)!}
C({\bf L}_N^{\rm T})\norm{\Psi^{(N)}(0)}   .
\end{align}
To ensure that the Taylor truncation error is less than $\epsilon/4$, we choose 
the truncation order $k$ as
\begin{equation}
\norm{{\bf c}}\frac{(e-1)e^2m}{(k+1)!}
C({\bf L}_N^{\rm T})\norm{\Psi^{(N)}(0)}   \le \frac{\epsilon}{4}\,\Longrightarrow(k+1)! \geq \frac{4e^3}{\epsilon} s\norm{{\bf d}}m \frac{R_p}{\sqrt{1-R_p^2}},   
\end{equation}
where we have used $C({\bf L}_N^{\rm T})\leq 1$, 
$\norm{\Psi^{(N)}(0)}\leq \frac{R_p}{\sqrt{1-R_p^2}}$ since $\gamma = R_p <1$, 
and $\|{\bf c}\|_p\leq s\|{\bf d}\|_p$.
      Therefore it suffices to choose
      \begin{align}
       k \geq \left\lceil \log_2 \left( \frac{4e^3}{\epsilon} s\norm{{\bf d}}m \frac{R_p}{\sqrt{1-R_p^2}}\right)\right\rceil. 
    \end{align}

    \item Block encoding error: We can ensure the block encoding error of the truncated Taylor series
    stays below $\epsilon/4$ by choosing sufficiently a small error, $\sigma$, in the block encoding of $\mathcal{A}^{-1}$ as follows,
    \begin{align}
\abs{{\bf c}\cdot \Phi_m - \alpha_\mathcal{C} \alpha_{\mathcal{L}^{-1}}\alpha_\mathcal{B} \bra{0}U_\mathcal{C}^\dagger U_{\mathcal{L}^{-1}} U_\mathcal{B}\ket{0}}
&=
 \abs{\alpha_\mathcal{C} \alpha_\mathcal{B}({}_1\!\bra{0}U_\mathcal{C}^\dagger \mathcal{L}^{-1} U_\mathcal{B}\ket{0}_1) - \alpha_\mathcal{C}\alpha_{\mathcal{L}^{-1}}\alpha_\mathcal{B} \bra{0}U_\mathcal{C}^\dagger U_{\mathcal{L}^{-1}} U_\mathcal{B}\ket{0}}\nonumber \\
&\le \alpha_\mathcal{C} \alpha_\mathcal{B}\mathcal{E}(\sigma),
\end{align}
where we have used \cref{lemma:complexity_A_inv} and \cref{remark:block_enc_L} in the last step to obtain the error in the block encoding $U_{\mathcal{L}^{-1}}$ of $\mathcal{L}^{-1}$. Now using \cref{eq:alpha_E_sigma} in \cref{lemma:complexity_A_inv},
we obtain an upper bound for the error $\sigma$ of the block encoding of $\mathcal{A}^{-1}$
such that the error of the block encoding stays below $\epsilon/4$. We get  
\begin{equation}
\alpha_\mathcal{C} \alpha_\mathcal{B}\mathcal{E}(\sigma) \leq \frac{\epsilon}{4}\Longrightarrow \sigma  \leq \frac{\epsilon\sqrt{m}}{4\|{\bf d}\| s }\frac{\sqrt{1-R_p^2}}{R_p} - 8e^4m^2\sqrt{k+1}\tau,
\end{equation}
where we have used $\alpha_\mathcal{B}\leq\frac{R_p}{\sqrt{1-R_p^2}}$, \cref{lemma:fourier_coefficient_preperation} and $C({\bf L}_N^{\rm T})\leq 1$. The choice for the parameters $\sigma,\tau$ given in \cref{eq:dissipative_parameters} ensures that $\sigma >0$ and $\tau \leq\frac{1}{4e^2mC({\bf L}_N^{\rm T})\sqrt{k+1}}$ so that \cref{lemma:complexity_A_inv} is valid.
    
    \item Error in calculating the expectation value of the unitary operator from \cref{sec:Reduction to expectation value of a unitary operator}: We can set the error in calculating the expectation value 
    to be bounded from above by $\epsilon/4$ using \cref{lemma:complexity_A_inv}, $\alpha_\mathcal{B}\leq\frac{R_p}{\sqrt{1-R_p^2}}$, \cref{lemma:fourier_coefficient_preperation} and $C({\bf L}_N^{\rm T})\leq 1$ as follows;
    \begin{align}
        \abs{\alpha_\mathcal{C} \alpha_{\mathcal{L}^{-1}}\alpha_\mathcal{B} \bra{0}U_\mathcal{C}^\dagger U_{\mathcal{L}^{-1}} U_\mathcal{B}\ket{0}
    -\alpha_\mathcal{C} \alpha_{\mathcal{L}^{-1}}\alpha_\mathcal{B} w} \le 
    \alpha_\mathcal{C} \alpha_{\mathcal{L}^{-1}}\alpha_\mathcal{B} \delta \le \frac{\epsilon}{4}\\
    \Longrightarrow \delta  \leq \frac{\epsilon}{\sqrt{m}\|{\bf d}\|s} \frac{\sqrt{1-R_p^2}}{R_p} .
    \end{align}

\end{enumerate}

{\noindent \bf Complexity:}
 The complexity of the algorithm comes from the computation of $\expval{0|U_\mathcal{C}^\dagger U_{\mathcal{L}^{-1}} U_\mathcal{B}|0}$ which, from \cref{lem:readout} makes $\mathcal{O}(1/\delta)$ calls to each of
$U_\mathcal{C}$, $U_{\mathcal{L}^{-1}}$ and $U_\mathcal{B}$. The complexity of constructing  $U_{\mathcal{L}^{-1}}, U_\mathcal{B}$  and $U_\mathcal{C}$ is given in \cref{lemma:complexity_A_inv}, \cref{lemma:complexity_inital_state} and \cref{lemma:fourier_coefficient_preperation} respectively. The final result follows from the choice for the parameters $\delta,m$ and $\mathcal{O}(1/\sigma) = \mathcal{O}(1/\tau) = \mathcal{O}(c_2/c_1)$. 
\end{proof}

With additional constraints placed on the size of block-encoding scaling factors of $\widetilde{F}_0,\widetilde{F}_1$,  $R_p$ and the Fourier order of the readout function $K$, the total complexity result of \cref{theorem:dissipative_result} admits a simplified form, the proof of which is given in \cref{app:dissipative-result}.

\begin{corollary} [Complexity with dissipative conditions] \label{colrollary:dissipative-result} 
Suppose the Fourier ODE \cref{prob:problem_statement_formal} also satisfies
\begin{equation}
    \alpha \in \mathcal{O}(\norm{G_0}_{\infty}),\quad 
    \beta \in \mathcal{O}(\norm{G_1}),\quad
    K\in \mathcal{O}(1), \quad
    (1-R_p)^{-1} \in \mathcal{O}(1). 
\end{equation}
Then the total query complexity for \cref{theorem:dissipative_result} is 
\begin{equation}
\tilde{\mathcal{O}}\left(\frac{T^{3/2}\norm{G_0}_{\infty}^{3/2}
     \norm{{\bf d}}}{\epsilon}\left(1+\frac{\norm{G_1}}{\|G_1\|_{\rm row, q}}\right)
     \left(1+ \left(\frac{\tilde{\mu}_0}{\|G_1\|_{\rm row,q}}\right)^K\right)
    \right)
\end{equation}
where $\tilde{\mathcal{O}}$ suppresses factors in
\begin{equation}
    {\rm polylog}
     \left(n,T,\norm{G_0}_{\infty},\norm{{\bf d}},1/\epsilon,
     \left(1+ \left(\frac{\tilde{\mu}_0}{\|G_1\|_{\rm row,q}}\right)^K\right)
     \right).
\end{equation}. 
\end{corollary}

\subsection{Algorithm and complexity without dissipative conditions}
  Finally, we present our algorithm for the case where the ODE does not satisfy  
  dissipative conditions but the solution can only be obtained at a short final time. The stability conditions given in \cref{lemma:rescaling_C(L)<1} subsumed the dissipative condition, therefore in order to generalize to the non-dissipative regime we first obtain an upper bound on the solution growth in this regime.

\begin{lemma}[Solution growth bound]\label{lemma:sol_growth_bound}
For $p\in[1,\infty)$ and $\nu \in \mathbb{R}_+$ for the rescaling given in \cref{eq:rescaling_definition}, we have
\begin{equation}
   C_p({\bf L}_N^{\rm T})=\sup_{t\in[0,T]}\norm{\exp({\bf L}_N^{\rm T}t)}_p \leq \exp\left(T\left(N\nu\|G_1\|_{\rm row, q} + \max\{-\tilde{\mu}_0, -N\tilde{\mu}_0\}\right)\right) := \Gamma_p(N,T,\nu) 
\end{equation}
for ${\bf L}_N^{\rm T}$ defined in \cref{eq:LN} and $\tilde{\mu}_0=\min \Im{(G_0)_j}$.
\end{lemma}

The proof can be found in \cref{app:norm_stabODE}. 
We now present the more general final result that does not require the problem to be dissipative. In the following we use $\Gamma(N,T,\nu) = \Gamma_2(N,T,\nu)$.

\begin{theorem}[Complexity without dissipative conditions]
\label{thm:non-dissipative-result}
For \cref{prob:problem_statement_formal}, let $\{p,q\} \in [1,\infty)$ be such that $1/p+1/q=1$.
Given $\tilde{\mu}_0$, $\|G_1\|_{{\rm row},q}$, and  $\|\mathbf d\|_{q}$, 
for $T\in[0, \widetilde{T}_{\rm max}]$ where
\begin{equation}
    \widetilde{T}_{\rm max} = \min\left\{\frac{1}{\max\left\{\alpha,
    \nu\norm{G_1}_{{\rm row},q}\right\}(1+1/r)}\ln \left(\frac{\nu}{r\|e^{i{\bf u}_0}\|_p}\right), 
    \frac{\ln{(r/e)}}{\alpha+
    \nu\norm{G_1}_{{\rm row},q}}\right\},
\end{equation}
and for any chosen $r,\nu$ satisfying
\begin{align}
    r \ge e, \nonumber\\ 
    \nu > \max\{r\|e^{i{\bf u}_0}\|_p,\sqrt{2}\|e^{i{\bf u}_0}\|\},
\end{align}
a quantum circuit can be constructed that returns  
$g({\bf u}(T))$ to an accuracy $\epsilon\in \mathbb{R}_+$ by making
\begin{align}
\label{eq:complexity-nondissipative}
\mathcal{O}\left( \frac{1}{\epsilon}  
N^{3/2}
k^{7/2} \ln^{3/2}(r)z
\left( 1 + \frac{\beta}{\|G_1\|_{\rm row ,q}}\right) 
\cdot \log k \log^2 \left(\frac{N^{3/2}\sqrt{k}\ln^{3/2}(r)z\Gamma}{\epsilon}\right)\right) &\quad \text{calls to } U_{G_0}, U_{G_1}, \\
    \mathcal{O}\left(\frac{N^{3/2}\sqrt{\ln(r)}z}{\epsilon} \right)  &\quad \text{calls to } U_{u_0}, \\
    \mathcal{O}\left(\frac{\sqrt{N\ln(r)}z}{\epsilon}\right) &\quad \text{calls to } U_{d},
\end{align}
where 
\begin{align}\label{eq:params_non_dissipative}
     N \in \mathcal{O} \left(\log\left(1+\dfrac{4K\,s\,\|\mathbf d\|_{q}}{r\epsilon}\right)\right),\\
    k \in \mathcal{O}\left(\log_2 \left( \frac{1}{\epsilon} N\ln(r) z \right)\right),\\
    z = \|{\bf d}\|s\Gamma,\\
    s = \max\{\nu,\nu^K\},\\
    \Gamma = \Gamma(N,T,\nu).
\end{align}

\end{theorem}

\begin{proof}The proof is divided into four parts: rescaling, the algorithm, its correctness and finally the complexity. \\

{\noindent \bf Rescaling:}
For the rescaling of the variable ${\bf u}$ given in \cref{eq:rescaling_definition} parametrised by $\nu\in \mathbb{R}_+$, with $\nu >\sqrt{2}\norm{e^{i{\bf u}_0}}$, the rescaled differential equation in the variable ${\bf x}(t)$ given in \cref{eq:problem_1_rescaled} satisfies $\gamma<\frac{1}{\sqrt{2}}$ from \cref{eq:rescaling_norm_input}.\\

We then construct a quantum algorithm to estimate $f({\bf x}(T))$ for $f$ given in \cref{eq:rescaled_output_function} since $f({\bf x}(T)= g({\bf u}(T))$.\\

{\noindent \bf Algorithm:}
The algorithm for estimating $f({\bf x}(T))$ is identical to \cref{theorem:dissipative_result} with the parameters $N,k,\sigma,\tau,\delta,m$ chosen to be
\begin{align}\label{eq:params_non_dissipative}
    N =
\left\lceil
\frac{
\log\left(\max\left(\dfrac{4Ks\|\mathbf d\|_{p}}{r\,\epsilon}, 1\right)\right)
}{
\log\left(\dfrac{r}{e^{(\alpha+\nu\|G_1\|_{\mathrm{row},q})T}}\right)
}
\right\rceil\\
    k = \max \left\{\left\lceil \log_2 \left( \frac{4e^3}{\epsilon} s \norm{{\bf d}}m \Gamma \right)\right\rceil , \lceil \log_2 me^2\rceil \right\} \\
   \sigma = c_1-c_2\tau,\\
    \tau = \min \left\{\frac{c_1}{1+c_2}, \frac{1}{4e^2m\sqrt{k+1}\Gamma}\right\},\\
    c_1 = \frac{\epsilon\sqrt{m}}{4\|{\bf d}\| s }, \\c_2 =  8e^4m^2\sqrt{k+1}\Gamma^2,\\
    \delta =\frac{\epsilon}{\sqrt{m}\|{\bf d}\|s \Gamma} \\ 
    m = \lceil TN\left(\alpha + \nu\norm{G_1}_{\rm row,q} \right)\rceil
    \end{align}

{\noindent \bf Correctness (error analysis):}
The four sources of error are identical to \cref{theorem:dissipative_result}. We now bound each error term. 
\begin{enumerate}
    \item Koopman truncation error:
    We can ensure the Koopman truncation error 
    stays below $\epsilon/4$ by choosing sufficiently large $N$ as follows.
    Since we have $T\leq\widetilde{T}_{\rm max}\leq T_{\max}$,
    by \cref{thm:finitetimebound}, we have
    \begin{align}
    \abs{f({\bf x}(T))- {\bf c}\cdot \Psi^{(N)}(mh)} =
        \abs{\sum_{l=1}^{K} {\bf c}_l\cdot(\Psi_l(T)-\Psi^{(N)}_l(T))}
        \le s\|{\bf d}\|_q\sum_{j=l}^{K}\norm{\eta_l(T)}_p \\
        \leq  s\|{\bf d}\|_q\frac{1}{r} \sum_{l=1}^{K} \left(\frac{e^{\left(\|F_0\|_{\infty} + \norm{F_1}_{{\rm row},q}\right)T}}{r}\right)^{N} \\
        \leq \frac{s\|{\bf d}\|_qK}{r}  \left(\frac{e^{\left(\|G_0\|_{\infty} + \nu\norm{G_1}_{{\rm row},q}\right)T}}{r}\right)^{N}\le \frac{\epsilon}{4}.
        \end{align}
    This is achieved by choosing an $N$ satisfying
    \begin{equation}
        \left(\frac{r}{e^{\left(\|G_0\|_{\infty} + \nu\norm{G_1}_{{\rm row},q}\right)T}}\right)^{N}
        \ge \frac{4s\|{\bf d}\|_qK}{r\epsilon}
    \end{equation}
Notice that $\frac{r}{e^{\left(\|G_0\|_{\infty} + \nu\norm{G_1}_{{\rm row},q}\right)T}} \ge 1$
is guaranteed by the definition of $\widetilde{T}_{\max}$. Now an $N$ satisfying this condition
is given by
    \begin{align}
N =
\left\lceil
\frac{
\log\left(\max\left(\dfrac{4Ks\|\mathbf d\|_{q}}{r\,\epsilon}, 1\right)\right)
}{
\log\left(\dfrac{r}{e^{(\alpha+\nu\|G_1\|_{\mathrm{row},q})T}}\right)
}
\right\rceil,
\end{align}
where we used $\|{\bf c}\|_q \leq  s\|{\bf d}\|_q$ and $\alpha \ge \|G_0\|_{\infty}$.\\

Furthermore, by definition of $\widetilde{T}_{\max}$, we have
\begin{equation}
    T \le \widetilde{T}_{\max} \le \frac{\ln{(r/e)}}{\alpha+
    \nu\norm{G_1}_{{\rm row},q}}
    \implies 
    \dfrac{r}{e^{(\alpha+\nu\|G_1\|_{\mathrm{row},q})T}} \ge e.
\end{equation}
This implies that 
\begin{equation}
\label{eq:Nscaling}
    N \in \mathcal{O}\left(\log\left(1+\dfrac{4K\,s\,\|\mathbf d\|_{q}}{r\epsilon}\right)\right).
\end{equation}

    \item Taylor truncation error: We can ensure the Taylor truncation  error 
    stays below $\epsilon/4$ by choosing sufficiently large $k$ as follows. Using \cref{lemma:taylor_error_bound} which is valid since the choice of $m$ and $k$ ensure that
\begin{align}
 h \leq  \frac{1}{N\left(\norm{F_0}_\infty + \norm{F_1}_{\rm row,q}\right)},\\ m\leq \frac{(k+1)!}{e^2}
\end{align}
since $\alpha \geq \|G_1\|_{\rm row, q}$, we have,
    \begin{align}
        \abs{{\bf c}\cdot \Psi^{(N)}(mh) - {\bf c}\cdot \Phi_m}\le 
        \norm{{\bf c}}\norm{\Psi^{(N)}(mh)-\Phi_m}\\
\le
       \norm{{\bf c}}\frac{(e-1)e^2m}{(k+1)!}
     C({\bf L}_N^{\rm T})\norm{\Psi^{(N)}(0)}   \le \frac{\epsilon}{4}\\
      \Longrightarrow(k+1)! \geq \frac{4e^3}{\epsilon} s\norm{{\bf d}}m \Gamma   \\
      \Longrightarrow k \geq \left\lceil \log_2 \left( \frac{4e^3}{\epsilon} s\norm{{\bf d}}m \Gamma \right)\right\rceil, 
    \end{align}
where we have used  \cref{lemma:sol_growth_bound} with  $T\leq \widetilde{T}_{\rm max}$, $\norm{\Psi^{(N)}(0)}\leq 1$ since $\gamma<1/\sqrt{2}$, $(k+1)!\geq 2^{k}$ and $\|{\bf c}\|_q\leq \max\{\nu,\nu^K\}\|{\bf d}\|_q$.

    \item Block encoding error: We can ensure the block encoding error 
    stays below $\epsilon/4$ by choosing sufficiently a small error in the block encoding of $\alpha_{\mathcal{L}^{-1}}, \sigma$. From \cref{lemma:complexity_A_inv} we have,
    \begin{align}
\abs{{\bf c}\cdot \Phi_m - \alpha_\mathcal{C} \alpha_{\mathcal{L}^{-1}}\alpha_\mathcal{B} \bra{0}U_\mathcal{C}^\dagger U_{\mathcal{L}^{-1}} U_\mathcal{B}\ket{0}} =
 \abs{\alpha_\mathcal{C} \alpha_\mathcal{B}({}_1\!\bra{0}U_\mathcal{C}^\dagger \mathcal{L}^{-1} U_\mathcal{B}\ket{0}_1) - \alpha_\mathcal{C}\alpha_{\mathcal{L}^{-1}}\alpha_\mathcal{B} \bra{0}U_\mathcal{C}^\dagger U_{\mathcal{L}^{-1}} U_\mathcal{B}\ket{0}} 
        \\
        \le \alpha_\mathcal{C} \alpha_\mathcal{B}\mathcal{E}(\sigma) \leq \frac{\epsilon}{4}\\
\Longrightarrow \sigma  \leq \frac{\epsilon\sqrt{m}}{4\|{\bf d}\| s } - 8e^4m^2\sqrt{k+1}\tau \Gamma^2
    \end{align}
    where we have used $\alpha_\mathcal{B}\leq1$, \cref{lemma:fourier_coefficient_preperation} and \cref{lemma:sol_growth_bound} with $T\leq \widetilde{T}_{\rm max}$. The choice for the parameters $\sigma,\tau$ given in \cref{eq:dissipative_parameters} ensures that $\sigma >0$ and $\tau \leq\frac{1}{4e^2mC({\bf L}_N^{\rm T})\sqrt{k+1}}$ so \cref{lemma:complexity_A_inv} is valid.
    
    \item Error in calculating the expectation value: We can set the error in calculating the expectation value 
    to be bounded from above by $\epsilon/4$ using \cref{lemma:complexity_A_inv}, $\alpha_\mathcal{B}\leq1$, \cref{lemma:fourier_coefficient_preperation} and \cref{lemma:sol_growth_bound} with $T\leq \widetilde{T}_{\rm max}$ as follows,  
    \begin{align}
        \abs{\alpha_\mathcal{C} \alpha_{\mathcal{L}^{-1}}\alpha_\mathcal{B} \bra{0}U_\mathcal{C}^\dagger U_{\mathcal{L}^{-1}} U_\mathcal{B}\ket{0}
    -\alpha_\mathcal{C} \alpha_{\mathcal{L}^{-1}}\alpha_\mathcal{B} w} \le 
    \alpha_\mathcal{C} \alpha_{\mathcal{L}^{-1}}\alpha_\mathcal{B} \delta \le \frac{\epsilon}{4}\\
    \Longrightarrow \delta  \leq \frac{\epsilon}{\sqrt{m}\|{\bf d}\|s \Gamma} .
    \end{align}

\end{enumerate}

{\noindent \bf Complexity:}
Follows \cref{theorem:dissipative_result}. 
For simplifying the complexity expressions,
we used 
\begin{equation}
    \widetilde{T}_{\rm max} \le   \frac{\ln{(r/e)}}{\alpha+\nu\norm{G_1}_{{\rm row},q}}
    \implies 
    \widetilde{T}_{\rm max}(\alpha+\nu\norm{G_1}_{{\rm row},q}) \in \mathcal{O}(\ln(r))
\end{equation}
and 
\begin{equation}
    \frac{\alpha+\nu\beta}{\alpha+ \nu\|G_1\|_{\rm row ,q}} \in 
    \mathcal{O}\left( 1 + \frac{\beta}{\|G_1\|_{\rm row ,q}}\right).
\end{equation}
\end{proof}

When $p\geq2$, the total complexity result of \cref{thm:non-dissipative-result} simplifies; a proof is given in \cref{app:non-dissipative-result}.

\begin{corollary}[Total complexity without dissipative conditions] \label{colrollary:non-dissipative-result} 
For $p \ge 2$, $K\in \mathcal{O}(1)$, 
\begin{equation}
       \ell = \ln \left(\frac{\nu}{r\|e^{i{\bf u}_0}\|_p}\right) \in \mathcal{O}(1) 
    \end{equation}
and $\|e^{i{\bf u}_0}\|_p > 1$,
the total query complexity in \cref{thm:non-dissipative-result} is
\begin{equation}
\tilde{O}\left(\frac{r^{K\ell}\|\mathbf d\|_{q}^{1+\ell}
    \|e^{i{\bf u}_0}\|_p^{K(\ell +1)}
    }{\epsilon^{1+\ell}}\right).
\end{equation}
\end{corollary}

\begin{remark}
    The query complexity and the maximum time $\widetilde{T}_{\max}$ in \cref{thm:non-dissipative-result}
    have complicated dependence on the parameter $r$. For $r=e$, we have 
    $\widetilde{T}_{\max}=0$ as the second factor in its definition vanishes.
    As $r$ increases, $\widetilde{T}_{\max}$ increases until some value $r_0$
    but then reaches a maximum and eventually decreases due to $1/r$-dependence
    of the denominator of the first factor. 
    The value of $r$ appears in 
    the complexities in \cref{eq:complexity-nondissipative} directly and 
    also indirectly through $\nu$, which is required to
    be greater than $r\norm{e^{i{\bf u}_0}}_p$.

\end{remark}

\section{Conclusion}
\label{sec:conclusion}
In this work we applied the Koopman linearization method to design efficient quantum algorithms
for solving nonlinear ODEs with Fourier nonlinearity.
We developed two algorithms for this problem. 
The first algorithm applies under dissipative conditions which we derive. 
Under these conditions, our Koopman linearization-based approach is expected to 
be significantly more efficient than the existing methods based on Carleman linearization,
which would require approximating the Fourier nonlinear term by a polynomial
function incurring further errors. 
Our second algorithm removes the dissipativity assumption and demonstrates that efficient solution extraction remains possible on a short
%, finite 
time interval. We also discussed how to extract classical quantities from the resulting quantum state representing the solution. \\

Like Carleman linearization, the design 
of our algorithm based on Koopman linearization
also begins by rescaling to obtain a nonlinear ODE
in terms of the rescaled variable ${\bf x}$ 
with some desired properties. 
While Carleman linearization yields
an ODE governing the evolution of an array of monomials the form
$\Psi_{\rm Carl} = [{\bf x}\ {\bf x}^{\otimes 2} \dots]^{\rm T}$,
Koopman linearization applied to the nonlinear ODE under consideration yields
a linear ODE governing the evolution of an array of trigonometric monomials of the form
$\Psi_{\rm Koop} = [e^{i{\bf x}}\ (e^{i{\bf x}})^{\otimes 2} \dots]^{\rm T}$.
Naturally, our algorithm is well suited for generating a state proportional
to $e^{i{\bf x}}$, but not amenable to generating a state proportional
to ${\bf x}$, which is typically called the solution state in the literature.
Whether generation of $e^{i{\bf x}}$ instead of ${\bf x}$ is acceptable will
depend on the application of the ODE solver. Nevertheless, 
as we show here, access to $\Psi_{\rm Koop}(T)$
allows computation of a large class of classical observables systematically,
and these observables being trigonometric polynomials of the entries of 
${\bf x}$ are complementary to those 
accessible via Carleman linearization based algorithms, namely 
polynomials of the entries of ${\bf x}$.\\

We remark that as far as the problem of 
generating the state $e^{i{\bf x}(T)}$ is concerned, the particular
ODE we consider in this work can be tackled by Carleman-based approaches with 
a substitution of variables $\Psi_1  = e^{i{\bf u}}$. 
However this is a feature of the specific ODE considered in this paper.
There are other classes of nonlinear ODEs, including $d{\bf u}/dt = 
G_0 + G_1 e^{i{\bf u}} + G_{-1} e^{-i{\bf u}}$, which do not 
map to polynomial nonlinear ODEs with any straightforward 
substitution of variables and yet can potentially be tackled by 
Koopman linearization.   It remains an open question how well direct application of
 Carleman linearization (i.e., without substitution of variables) would perform on the ODE considered in this paper. As discussed in \cref{sec:approach}, 
 this question can be tackled by truncating the Taylor series of 
the nonlinear term and bounding the resulting error, which then allows direct application of
Carleman linearization. However, such an approach is likely to 
lead to different dissipativity conditions than the Koopman linearization approach,
pointing to complementary strengths of the two approaches.\\

Another interesting direction opened up by the Koopman linearisation framework is the study of lower bounds and no-go theorems. The Koopman linearisation approach does not aim to generate the solution state $\textbf{x}$, rather it tracks the evolution of a carefully chosen set of observables. Consequently, the no-go theorems and lower bounds established in the literature do not apply directly. Importantly it is a priori unclear whether the techniques used to establish such results extend to Koopman linearization–based algorithms, as these methods rely on the rapid separation of configuration trajectories—a feature that may not persist at the level of observable arrays. This raises a natural question: can Koopman linearization–based algorithms enable the solution of nonlinear ODEs that are otherwise intractable using Carleman linearization? Or more broadly, could they facilitate the study of PDEs exhibiting chaos or turbulence on a quantum computer? \\

Beyond these exciting prospects for taming nonlinearity, the Koopman linearization 
framework introduced here opens new directions in the design of quantum algorithms
for nonlinear dynamics. Rather than relying exclusively on Taylor expansions 
to treat differential equations with arbitrary nonlinearities in the dependent variable, 
one can, in cases where the nonlinearities exhibit particular structure, 
exploit alternative functional bases such as trigonometric, 
Chebyshev, or Hermite polynomials to obtain sparse and structured 
linearized ODEs. This flexibility suggests that nonlinear systems 
previously thought to be intractable on quantum computers may be 
efficiently solvable when linearized in a suitable basis. 
However, our work also indicates that the 
design of such efficient algorithms requires advances on several
fronts. For instance, we need to study how to choose a basis of observables
that generates a linearized ODE with a generator that is easy to 
process on a quantum computer (e.g. via block encoding techniques),
and moreover, allows truncation with exponentially fast convergence.
Notably, novel techniques are likely to be required for bounding
truncation errors arising in linearization using new bases.
Similarly, preparation of a state proportional to the initial
value of the observable array would also require new ideas for
state generation.\\

One immediate direction for future work is to extend our approach to ODEs with 
nonlinear terms comprising trigonometric polynomials of entries of ${\bf x}$ 
of positive as well as negative degrees. Such nonlinearities 
are encountered in the Kuramoto model of coupled oscillators as demonstrated in~\cite{chen2024carleman}. Another direction is to construct 
Koopman linearization based algorithm for time-dependent nonlinear ODEs,
which requires the construction of a time-dependent infinitesimal Koopman matrix. 

\section*{Acknowledgements}
We thank Pedro C.S. Costa for many insightful discussions
and extensive feedback, including 
improvements in notation, simplifying technical explanations,
restructuring of sections, and
rewording of some paragraphs, lemmas and theorems, thereby
improving the readability of this manuscript. 
We also thank Nalini Joshi for many valuable discussions and feedback. 
AA is grateful for discussions with Nana Liu, Nathan Wiebe, Dominic Berry 
and Laura Henderson. 
AA acknowledges the support from a New Faculty Start-up Grant by Concordia University. 
GM acknowledges funding from BTQ Technologies Corp.

\bibliography{references_abbrev} 

\appendix

\section*{Appendices}

\section{Logarithmic norm and bounds on the matrix exponential}\label{appendix:log_norm}

The growth of the exponential of a matrix can be bounded using
the logarithmic norm of the matrix, which is defined as follows \cite{strom1975logarithmic}.
\begin{definition}
    The logarithmic $p$-norm of a matrix $A$ is defined as
    \begin{equation}
        \mu_p(A) = \lim_{h \to 0^+}\frac{\norm{\mathds{1}+hA}_p-1}{h}.
    \end{equation}
\end{definition}
For the special case $p=2$, we denote $\mu_2(A) = \mu(A)$, and we have
\begin{equation}
\label{eq:sec_def}
\mu(A) = \max_{\|\mathbf{v}\|=1}\left(\mathbf{v}^{\dagger}H(A)\mathbf{v}\right)
= \sup_{\|\mathbf{v}\|=1}\text{Re}\{\mathbf{v}^{\dagger}A\mathbf{v}\}=
\max\Bigg\{\lambda\; |\; \det\left(\frac{A^{\dagger}+A}{2}-\lambda I\right) =0 \Bigg\},
\end{equation}
where $H(A)$ is the hermitian part of the matrix $A$.

The following properties of logarithmic $p$-norm are well established.
\begin{lemma} For $p \in [1,\infty)$ and 
matrices $A, B \in \mathbb{C}^{n\times n}$, the following statements hold.
\begin{enumerate} 
    \item 
Triangle inequality for the logarithmic norm: 
  \begin{equation}
      \mu_p(A+B) \leq \mu_p(A) + \mu_p(B).\label{lemma:Trian_logN}
  \end{equation}  
\item For $\alpha(A) = \max\Big\{\text{Re}(\lambda)\; |\; \det(A-\lambda I) =0\Big\}$,
    \begin{equation} \label{lemma:alpa_leq_mu}
        \alpha(A) \leq \mu_p(A).
    \end{equation}
\item Moreover, we have
    \begin{equation}\label{lemma:muA} 
        \mu_p(A) \leq \|A\|_p.
    \end{equation}
\end{enumerate}
\end{lemma}

Another important lemma is between the exponents of the quantities above
\begin{lemma}\label{lemma:expA}
Let $A \in \mathbb{C}^{n\times n}$ and $t\geq 0$, then
\begin{equation}
    \|\exp{(At)}\|_p \leq \exp{\mu_p(A)t}.
\end{equation}
\end{lemma}

\section{Properties of the linearized differential equation}\label{app:truncation_bound_proofs}

The following Lemma establishes equality between $\norm{\widetilde{F}_1}_p$ and 
$\norm{F_1}_{{\rm row},q}$ by using the fact that $\widetilde{F}_1$ is obtained by stacking 
the rows of $F_1$ diagonally, as made precise by Eq.~\eqref{eq:F1_tilde}.
\begin{lemma}\label{lemma:F_1_p_F_1_row_q}
    For $\widetilde{F}_1$ defined in Eq.~\eqref{eq:F1_tilde} and for $p,q \in [1,\infty)$ satisfying
    $1/p + 1/q = 1$, we have
    \begin{equation}
        \norm{\widetilde{F}_1}_p = \norm{F_1}_{{\rm row},q}.
    \end{equation}
\end{lemma}

\begin{proof}
    We will prove this lemma by proving that $\norm{\widetilde{F}_1}_p \le \norm{F_1}_{{\rm row},q}$
    and $\norm{\widetilde{F}_1}_p \ge \norm{F_1}_{{\rm row},q}$. 
    To prove the first direction, recall by definition
    \begin{equation}
        \norm{\widetilde{F}_1}_p = \sup_{\norm{\Psi}_p=1}\norm{\widetilde{F}_1 \Psi}_p,
    \end{equation}
    where $\Psi \in \mathbb{C}^{n^2}$.
    We can express $\Psi$ as a block vector
    $\Psi = (\Psi_1 \dots  \Psi_n)^{\rm T}$ where each 
    component $\Psi_j \in \mathbb{C}^{n}$. 
    Using the definition of $\widetilde{F}_1$ in Eq.~\eqref{eq:F1_tilde}, we have 
    $\widetilde{F}_1 \Psi = ((F_1)_1\Psi_1 \dots (F_1)_n\Psi_n)$,
    where $(F_1)_j$ denotes the $j$-th row of $F_1$. Now
    by applying Holder's inequality, we get
     \begin{align}
        \norm{\widetilde{F}_1}_p &=  \sup_{\norm{\Psi}_p=1}\left(\sum_{j=1}^{n}\abs{(F_1)_j\Psi_j}^p\right)^{1/p}
        \nonumber \\
        &\le \sup_{\norm{\Psi}_p=1}\left(\sum_j\norm{(F_1)_j}_q^{\,p}\norm{\Psi_j}_p^{\,p}\right)^{1/p}.
    \end{align}
    Let $k \in \{1,\dots,n\}$ be the index of the row of $F_1$ with maximum $q$-norm, i.e.
    $\norm{(F_1)_k}_q = \sup_j \norm{(F_1)_j}_q$. If there are multiple rows with 
    the same maximum norm, then we choose $k$ to be the index of any one of these rows. Then
    \begin{align}
        \norm{\widetilde{F}_1}_p 
        &\le \sup_{\norm{\Psi}_p=1}\left(\norm{(F_1)_k}_q^{\,p}\sum_j\norm{\Psi_j}_p^{\,p}\right)^{1/p}
        \nonumber\\
        &=\left(\norm{F_1}_{{\rm row},q}^p\right)^{1/p}\nonumber\\
        &=\norm{F_1}_{{\rm row},q}.
    \end{align}
    Here we used $\sum_j\norm{\Psi_j}_p^{\,p} = 1$ which follows from $\norm{\Psi}_p = 1$.
    We have therefore proved $\norm{\widetilde{F}_1}_p \le \norm{F_1}_{{\rm row},q}$.

    To prove the other direction, namely $\norm{\widetilde{F}_1}_p \ge \norm{F_1}_{{\rm row},q}$,
    our strategy is to explicitly construct a $\Psi \in \mathbb{C}^{n^2}$ with $\norm{\Psi}_p=1$ 
    such that 
    $\norm{\widetilde{F}_1\Psi}_p = \norm{F_1}_{{\rm row},q}$. Then the 
    statements follows from the fact that $\norm{\widetilde{F}_1}_p \ge \norm{\widetilde{F}_1\Psi}_p$.
    Consider the vector $\Psi$ with entries
    \begin{equation}
        \Psi_{jl} = \left\{
        \begin{array}{lcl}
        0 &\text{if} & j\ne k \\
        \left(\frac{\abs{F_{1kl}}}{\norm{(F_1)_k}_q}\right)^{q/p}e^{-i\arg{(F_1)_{kl}}}
        &\text{if} & j=k.
        \end{array}\right.
    \end{equation}
The condition $\norm{\Psi}_p=1$ follows from
\begin{equation}
    \norm{\Psi}_p^{\,p} = \sum_{j,l}\abs{\Psi_{jl}}^p 
    = \sum_{l} \left(\frac{\abs{(F_1)_{kl}}}{\norm{(F_1)_k}_q}\right)^{q} 
    = \frac{\sum_{l} \abs{(F_1)_{kl}}^q}{\norm{(F_1)_k}_q^{\,q}}
    =1.
\end{equation}
To prove $\norm{\widetilde{F}_1\Psi}_p = \norm{F_1}_{{\rm row},q}$,
we first note that $\widetilde{F}_1\Psi$ has only one non-zero entry at $k$-th position
with value $(F_1)_k\Psi_k$. Then
\begin{align}
    \norm{\widetilde{F}_1\Psi}_p &= \abs{(F_1)_k\Psi_k} \nonumber\\
    &= \sum_l F_{1kl}\left(\frac{\abs{F_{1kl}}}{\norm{(F_1)_k}_q}\right)^{q/p}e^{-i\arg{(F_1)_{kl}}}
    \nonumber\\
    &= \frac{\sum_l \abs{(F_1)_{kl}}^{q/p+1}}{(\sum_l\abs{(F_1)_{kl}}^q)^{1/p}}.
\end{align}
Using $1/p+1/q = 1$, we have $q/p+1=q$, which yields
\begin{align}
    \norm{\widetilde{F}_1\Psi}_p &= \frac{\sum_l \abs{(F_1)_{kl}}^{q}}{(\sum_l\abs{(F_1)_{kl}}^q)^{1/p}}\nonumber\\
    &= \left(\sum_l \abs{(F_1)_{kl}}^{q}\right)^{1-1/p}\nonumber\\
    &= \left(\sum_l \abs{(F_1)_{kl}}^{q}\right)^{1/q} \nonumber\\
    &= \norm{(F_1)_{k}}_q \nonumber\\
    &= \norm{F_1}_{{\rm row},q}.
\end{align}
This proves $\norm{\widetilde{F}_1}_p \ge \norm{F_1}_{{\rm row},q}$ and therefore completes the proof.
\end{proof}

The next Lemma establishes an upper bound on the $p$-norm of the linearized and truncated approximation to \cref{eq:problem_1_rescaled}.
\begin{lemma}\label{lemma:LN_le_F} For ${\bf L}_N^{\rm T}$ defined in \cref{eq:LN} and $\widetilde{F}_1, \widetilde{F}_0$ defined in \cref{eq:F1_tilde} and \cref{eq:F0_tilde} respectively, we have
    \begin{equation}
    \|{\bf L}_N^{\rm T}\|_p\le N\left(\norm{\widetilde{F}_0}_p + \norm{\widetilde{F}_1}_p\right).\end{equation}
\end{lemma}

\begin{proof}

 We start by rewriting the matrix $\mathbf{L}_N^{\rm T}$ as follow 
    \begin{equation}
        {\bf L}_N^{\rm T} = \begin{pmatrix}
    B_0^{(0)} & B_1^{(1)}  & 0 & 0 & \dots \\
    0 & B_1^{(0)} & B_2^{(1)} & 0 & \dots \\
     \vdots &  \vdots & \ddots &  \vdots & \dots \\
    0 & 0 & \cdots & 0 & B_N^{(0)}
\end{pmatrix} = H_0 + H_1,   
    \end{equation}
where 
\begin{equation}
H_0 = \sum_{j=1}^{N} | j \rangle \langle j | \otimes B_j^{(0)} , \quad 
H_1 = \sum_{j=1}^{N-1} | j \rangle \langle j + 1 | \otimes B_j^{(1)}.
\end{equation}
In the decomposition above we have $\{\ket{j}\}_{j=0}^{N+1}$ as the basis of $\mathbb{C}^{N+1}$.\\

Using subadditivity of the norm we have
\begin{equation}
    \|{\bf L}^{\rm T}_N\|_p \le \|H_0\|_p+\|H_1\|_p
\end{equation}.

We then have for any $x\in\mathbb{C}^{\text{dim}({\bf L}_N^{\rm T})}$
\begin{align*}
\|H_0 x\|_p^p
  \le \sum_{j=0}^{N} \big(\|B_j^{(0)}\|_p^p\,\|x_j\|_p^p\big)\le \left(\max_{0\le j\le N}\|B_j^{(0)}\|_p\right)^p
      \sum_{j=0}^{N} \|x_j\|_p^p
   \\=  \left(\max_{0\le j\le N}\|B_j^{(0)}\|_p\right)^p\|x\|_p^p \leq (N\|\widetilde{F}_0\|_p)^p\|x\|_p^p,
\end{align*}

and

\begin{align*}
\|H_1 x\|_p^p \le \sum_{j=0}^{N-1} \big(\|B_j^{(1)}\|_p^p\,\|x_j\|_p^p\big) \le \left(\max_{0\le j\le N-1}\|B_j^{(1)}\|_p\right)^p
      \sum_{j=0}^{N-1} \|x_j\|_p^p
   \\=  \left(\max_{0\le j\le N-1}\|B_j^{(1)}\|_p\right)^p\|x\|_p^p \leq ((N-2)\|\widetilde{F}_1\|_p)^p\|x\|_p^p \le N\|\widetilde{F}_1\|_p^p\|x\|_p^p
\end{align*}
where we have used $\|B^{(1)}_{j+1}\| \leq j\|\widetilde{F}_1\|_p, \|B^{(0)}_{j}\| \leq j\|\widetilde{F}_0\|_p$.\\

Then from the definition of the operator $p$-norm we have 
\begin{equation}
    \|H_0\|_p = \sup_{x\neq0} \frac{\|H_0x\|_p}{\|x\|_p} \leq N \|\widetilde{F}_0\|_p
\end{equation}

and

\begin{equation}
    \|H_1\|_p = \sup_{x\neq0} \frac{\|H_1x\|_p}{\|x\|_p} \leq N \|\widetilde{F}_1\|_p
\end{equation}.

\end{proof}

\section{Proofs of lemmas for the error analysis} \label{app:error_analysis}

{\noindent \bf \cref{lemm:dissp}}
(Dissipativity in the $p$-Norm){\bf .}
Let $\{p,q\} \in [1,\infty)$ be such that $1/p+1/q=1$. 
Consider the rescaled ODE from \cref{eq:problem_1_rescaled}  where   
\begin{equation}
    \tilde{\mu}_0 = \min_j \Im{(F_0)_j}\ge 0 \quad \text{and} \quad
     R_p = \frac{\|F_1\|_{{\rm row},q} \|\Psi_1(0)\|_p}{\tilde{\mu}_0}<1.
\end{equation}
Then
\begin{equation}
    \|\Psi_1(t)\|_p\leq \|\Psi_1(0)\|_p.
\end{equation}

\begin{proof}
From the definition given by $R_p$, 
we have 
\begin{equation}
    \norm{\Psi_1(0)}_p = \|e^{i{\bf x}_0}\|_p
<\frac{\tilde{\mu}_0}{\|F_1\|_{{\rm row},q}} = 
\frac{\tilde{\mu}_0}{\|\widetilde{F}_1\|_p}.
\end{equation}
Since ${\bf x}(t)$ is a continuous function, there exists a $\delta>0$ 
such that 
\begin{equation}
\label{eq:ineq_cond}
\|\widetilde{F}_1\|_p\norm{e^{i\mathbf{x}(t)}}_p<\tilde{\mu}_0,\quad t\in[0,\delta].  
\end{equation}
We will now show that 
$d\|\Psi_1(t)\|_p/dt < 0$ for $t \in [0,\delta]$, from which
the lemma follows by repeated application.

To show that 
$d\|\Psi_1(t)\|_p/dt < 0$ for $t \in [0,\delta]$, 
first note that $|e^{ix_j}|^2 = e^{ix_j}e^{-ix_j^*} = e^{-2\Im(x_j)}$. Therefore
$d |e^{ix_j}|^2/dt = -2e^{-2\Im(x_j)} \Im(dx_j/dt) = -2|e^{ix_j}|^2 \Im(\dot{x_j})$.
Then
\begin{equation}
    \frac{d |e^{ix_j}|^p}{dt} = 
    \frac{d \left(|e^{ix_j}|^2\right)^{p/2}}{dt}=
    \frac{p}{2}\left(|e^{ix_j}|^2\right)^{p/2-1}
    \frac{d |e^{ix_j}|^2}{dt}
    = -p|e^{ix_j}|^{p}\Im(\dot{x_j}).
\end{equation}
Next, using \cref{eq:problem_1_rescaled}, i.e.
\begin{equation}
\frac{dx_j}{dt} = \sum_{_l}(F_1)_{jl} e^{ix_l} + (F_0)_{j},    
\end{equation}
we obtain
\begin{align}
\label{eq:pnormdecrease}
    \frac{(\norm{e^{i\textbf{x}}}_p)^p}{dt} &= \sum_j \frac{d |e^{ix_j}|^p}{dt} \nonumber \\
    &= -\sum_j p|e^{ix_j}|^{p}\Im(\dot{x_j}) \nonumber \\
    &= -\sum_j p|e^{ix_j}|^{p}\Im\left(F_0^j + \sum_l F_1^{jl}e^{ix_l}\right)
    \nonumber \\
    &= -\sum_j p|e^{ix_j}|^{p}\left(\Im(F_0^j) + \Im(\sum_l F_1^{jl}e^{ix_l})\right)
\end{align}

Now it suffices to show that the term inside the parentheses is positive.
By observing that $\sum_l (F_1)_{jl}e^{ix_l} = (\widetilde{F}_1(\textbf{e}_j\otimes\Psi_1))_j$,
we get
Therefore
\begin{align}
\Im\left(\sum_l (F_1)_{jl}e^{ix_l}\right)&\leq \Bigg|\sum_l (F_1)_{jl}e^{ix_l}\Bigg| 
\nonumber\\
&= |(\widetilde{F}_1(\textbf{e}_j\otimes\Psi_1))_j|.
\end{align}
For any vector $\textbf{x}$, its $i$th component is less than its $p$-norm, i.e.
$|x_i|\leq \|\mathbf{x}\|_p$. 
\begin{align}
\Im\left(\sum_l (F_1)_{jl}e^{ix_l}\right)
&\le |(\widetilde{F}_1(\textbf{e}_j\otimes\Psi_1))_j| \nonumber\\
&\leq \|\widetilde{F}_1(\textbf{e}_j\otimes\Psi_1)\|_p  \nonumber\\
&\leq  \|\widetilde{F}_1\|_p\|\Psi_1\|_p  \nonumber\\
&< \tilde{\mu}_0,
\end{align}
where in the last equality we used \cref{eq:ineq_cond}.
Consequently, we get $\left(\Im(F_0^j) + \Im(\sum_l F_1^{jl}e^{ix_l})\right) < 0$,
and therefore
\begin{equation}
    \frac{(\norm{e^{i\textbf{x}}}_p)^p}{dt} =
    -\sum_j p|e^{ix_j}|^{p}\left(\Im(F_0^j) + \Im(\sum_l F_1^{jl}e^{ix_l})\right)
    < 0.
\end{equation}
We have thus proven that $d \norm{e^{i\mathbf{x}(t)}}_p/dt \le 0$ for $t\in[0, \delta]$.
The statement of the lemma follows from applying this procedure repeatedly.
\end{proof}

{\noindent \bf \cref{lemma:finitetimedecreasingnorm}} (Upper Bounded in  $p$-Norm) \textit{Let $\{p,q\} \in [1,\infty)$ be such that $1/p+1/q=1$ and $r >1$. Consider the rescaled ODE from \cref{eq:problem_1_rescaled} where  
\begin{equation}
    \Lambda_p := \max\left\{\|F_0\|_{\infty},\norm{{F}_1}_{{\rm row},q}\right\} \quad \text{and} \quad
     \norm{\Psi_1(0)}_p < 1/r.
\end{equation}
Then for $t\in [0,T_r]$,
\begin{equation}
    \|\Psi_1(t) \|_p \leq \frac{1}{r}
\end{equation}
where
\begin{equation}
    T_{r} := \frac{1}{\Lambda_p(1+1/r)}\ln \left(\frac{1}{r\|\Psi_1(0)\|_p}\right).
\end{equation}}

\begin{proof}
By continuity there exists some $S$ such that for all $t\in [0,S]$ we have 
\[\|\Psi_1(t) \|_p\leq 1/r\]
and $\|\Psi_1(S) \|_p=1/r$.

We now show $T_r\leq S$ so that for $t\in[0,T_r]$ we have $\|\Psi_1(t) \|_p\leq 1/r$. We have that $\Psi_1$ satisfies  
\begin{align}
    \dot{\Psi}_1 &= B_1^{(0)}\Psi_1+B_2^{(1)}\Psi_{2} 
\end{align}
By integrating this equation, we get
\begin{align}
    \Psi_1(t) &=\Psi_0(0)+ \int^{t}_0   
    B_1^{(0)}\Psi_1(s)+B_2^{(1)}\Psi_{2}(s) ds.
\end{align}
Using subadditivity of the norm we have
\begin{align}
    \|\Psi_1(t)\|_p &\leq \|\Psi_1(0)\|_p+ \int^{t}_0\|B_1^{(0)}\Psi_1\|_p+\|B_2^{(1)}\Psi_{2}\|_p ds \nonumber\\
    &\leq \|\Psi_1(0)\|_p+ \int^{t}_0 \Lambda_p(\|\Psi_1\|_p+\|\Psi_{2}\|_p) ds\nonumber\\
    &\leq \|\Psi_1(0)\|_p+ \int^{t}_0 \Lambda_p(\|\Psi_1\|_p+\|\Psi_{1}\|_p^2) ds\nonumber\\
    &= \|\Psi_1(0)\|_p+ \int^{t}_0 \Lambda_p\|\Psi_1\|_p(1+\|\Psi_{1}\|_p) ds,
\end{align}
where we have used  $\|B^{(1)}_{j+1}\|_p \leq j\|\widetilde{F}_1\|_p, \|B^{(0)}_{j}\|_p \leq j\|\widetilde{F}_0\|_p$, \cref{lemma:F_1_p_F_1_row_q}, $\|\widetilde{F}_0\|_{p}=\|F_0\|_{\infty}$ and $\norm{\Psi_k}_p = (\norm{\Psi_1}_p)^k$ which follows from $\Psi_k = (\Psi_1)^{\otimes k}$.

For all $t\in [0,S]$,  we have $\|\Psi_1(t) \|_p \le 1/r$, therefore
\begin{equation}
    \|\Psi_1(t)\|_p \leq 
    \|\Psi_1(0)\|_p+ \Lambda_p (1+1/r)\int^{t}_0 \|\Psi_1(s)\|_p  ds.
\end{equation}
Applying the integral form of Gronwalls inequality gives
\begin{equation}
    \|\Psi_1(t)\|_p \leq \|\Psi_1(0)\|_pe^{\Lambda_p(1+1/r)t} \quad \forall t\in [0,S].
\end{equation}
We have $\|\Psi_1(S)\|_p = 1/r$, which implies
\begin{equation}
    \frac{1}{r} \leq \|\Psi_1(0)\|_p e^{\Lambda_p(1+1/r)S}.
\end{equation}
Therefore
\begin{equation}
    S\geq T_{r} := \frac{1}{\Lambda_p(1+1/r)}\ln \left(\frac{1}{r\|\Psi_1(0)\|_p}\right).
\end{equation}
\end{proof}

{\noindent \bf \cref{lemma:V_ke^-Lh_bound}}
(Taylor remainder bound)
\textit{For \cref{eq:TaylorSystem} let $\|{\bf L}_N^{\rm T}\|_ph\leq 1$ and $me^2/(k+1)!\leq 1$. Then for all $j\leq m,$
    \begin{equation}
        \norm{\left(V_k^j-e^{{\bf L}_N^{\rm T}hj}\right)e^{-{\bf L}_N^{\rm T}hj} }_p  \leq \frac{(e-1)je^2}{(k+1)!}
    \end{equation}
where $V_k$ is the truncated Taylor series given in \cref{eq:V_k}.}
\begin{proof} Let $X=V_ke^{-{\bf L}_N^{\rm T}h}$. We have,
\begin{equation}
    \norm{\left(V_k^j-e^{{\bf L}_N^{\rm T}hj}\right)e^{-{\bf L}_N^{\rm T}hj} }_p   
    = \norm{X^j -I }_p = \norm{(I-X)(I+X+\cdots+X^{j-1})}_p \leq \norm{I-X}_p
    \sum_{i=0}^{j-1}\norm{X}_p^i.
\end{equation}
Since $\|{\bf L}_N^{\rm T}\|_ph\leq 1$ we have,
\begin{align}
\norm{I-X}_p = \norm{\left(e^{{\bf L}_N^{\rm T}h}-V_k\right)\left(e^{-{\bf L}_N^{\rm T}h}\right)}_p  &\leq  \norm{e^{{\bf L}_N^{\rm T}h}-V_k}_p\norm{e^{-{\bf L}_N^{\rm T}h}}_p\nonumber \\ 
&\leq \norm{\sum_{i=k+1}^{\infty} \frac{({\bf L}_N^{\rm T}h)^i}{i!}}_p   e^{\norm{{\bf L}_N^{\rm T}}_ph}\nonumber \\ 
&\leq \frac{e^{\norm{{\bf L}_N^{\rm T}}_ph}\left(\norm{{\bf L}_N^{\rm T}}_ph\right)^{k+1}}{(k+1)!}e^{\norm{{\bf L}_N^{\rm T}}_ph}\nonumber \\ 
&\leq  \frac{e^2}{(k+1)!}.
\end{align}

Hence,
\begin{equation}
    \|X\|_p \leq  1+\|X-I\|_p \leq 1+\frac{e^2}{(k+1)!}.
\end{equation}
Then, from the last two inequalities
\begin{align}
    \norm{X^j -I }_p \leq \frac{e^2}{(k+1)!} \sum_{i=0}^{j-1}\left(1+\frac{e^2}{(k+1)!}\right)^i &=\left(1+\frac{e^2}{(k+1)!}\right)^{j}-1 \\ \nonumber &= \exp\left(j \log\left(1+e^2/(k+1)!\right)\right)-1\\ \nonumber
    &\leq \exp\left({\frac{je^2}{(k+1)!}}\right) -1  \\ \nonumber &\leq \frac{je^2(e-1)}{(k+1)!},
\end{align}
where we have used  $\ln(1+x)\leq x$ for $x\geq 0$,  $\frac{e^x-1}{x}$ being increasing for $0\leq x \leq 1$ and the assumption $je^2/(k+1)!\leq 1$.
\end{proof}

\section{Proofs of lemmas and supporting results for the algorithm complexity analysis}\label{App:Supp_res}

In this Appendix, we prove several preliminary lemmas required in the proof of 
\cref{lemma:complexity_L_inv}, including 
construction of block encodings of matrices $B_j^{(0)}$ and $B_j^{(1)}$
and other lemmas for bounding the error in the block encoding of $\mathcal{L}^{-1}$. We begin by reviewing known results on block encodings.
\begin{lemma}[Block-encoding of sparse-access matrices {\cite[Lemma~48]{gilyen2019quantum}}]\label{lemma:block_encoding_sparse}
Let $A$ be an $s_r$-row-sparse and $s_c$-column-sparse $s$-qubit matrix where each element of $A$ has absolute value at most $1$. Suppose we have access to the positions and values of the nonzero entries of $A$ with 3 oracles, $O_r,O_c,O_A$. Then, 
a $(\sqrt{s_rs_c}, 0)$-block encoding of $A$ can be constructed using $\mathcal{O}(1)$ calls to $O_r,O_c$ and $O_A$.\end{lemma}

The following lemmas review standard algebraic operations on block-encoded operators. 
Unless stated otherwise, a query to $U_A$ counts any call to $U_A$, $U_A^\dagger$, or c$U_A$. In the following lemmas we specify the exact circuit used, in other sections we will adopt the shorthand “calls to $U_A$” to mean any of these.

\begin{lemma}[Block encoding inversion {\cite[Theorem~41]{gilyen2019quantum}, \cite[Lemma~11]{low2024quantum}}]\label{lemma:block_encoding_inversion}
Let $U_A$ be a $(\alpha_A,0)$-block encoding of the operator $A$ where
$\alpha_A \geq \| A \|$. Then with  $\alpha_{A^{-1}} \geq \| A^{-1} \|$ a  $(2\alpha_{A^{-1}},\varepsilon)$-block encoding of $A^{-1}$ can be constructed using
$
\mathcal{O}\!\left( \alpha_A \alpha_{A^{-1}} \log\!\left(\tfrac{1}{\varepsilon}\right) \right)
$
queries to $cU_A$ and its inverse.
\end{lemma}

\begin{lemma}[Arithmetic using block-encoded matrices {\cite[Lemma~3]{krovi2023improved}, \cite[Lemma~53]{gilyen2019quantum}, \cite[Proposition~4]{takahira2021quantum} }]\label{lemma:block_encoding_arithmetic}
If $U$ is an $(\alpha,\epsilon)$-block encoding of the $s$-qubit operator $A$ and $V$ is a $(\beta,\delta)$-block encoding of the $w$-qubit operator $B$, then 

\begin{enumerate}
    \item an $(\alpha + \beta , \alpha \delta +\beta \epsilon)$-block encoding of $A+B$ can be constructed using $\mathcal{O}(1)$ calls to $cU$ and $cV$

    \item an $(\alpha \beta , \alpha \delta +\beta \epsilon)$-block encoding of $A B$ can be constructed using $\mathcal{O}(1)$ calls to $U$ and $V$

    \item an $(\alpha \beta , \alpha \delta +\beta \epsilon)$-block encoding of $A\otimes B$ can be constructed using $\mathcal{O}(1)$ calls to $U$ and $V$ 
\end{enumerate}
\end{lemma}
We next prove a lemma that states the complexity of inverting a $\delta$ accurate block encoding. 

\begin{lemma}[Block encoding inversion with error]\label{lemma:block_encoding_inversion_with_error} Let $U_A$ be an $(\alpha_A,\delta)$-block encoding of the operator $A$ where
$\alpha_A \geq \| A \|$, $\delta \leq \frac{1}{2\|A^{-1}\|}$. Then with  $\alpha_{A^{-1}} \geq \| A^{-1} \|$ a  $(4\alpha_{A^{-1}},\epsilon+2\delta \|A^{-1}\|^2)$-block encoding of $A^{-1}$ can be constructed using
$\mathcal{O}\left(\alpha_{A}\alpha_{A^{-1}}\log\left(\frac{1}{\epsilon}\right)\right)$ queries to c$U_{A}$ and its inverse. 
\end{lemma}

\begin{proof}
Let $\tilde{A}=A+\delta A$ with $\|A-\tilde{A}\|=\|\delta A\|\leq \delta$. From \cref{def:block_encoding} for $U_A$ we can choose 
$\tilde{A}=\alpha_A(\bra{0}\otimes I) U_A(\ket{0}\otimes I)$
so that $U_{A}$ is an $(\alpha_{A},0)$-block encoding of $\tilde{A}$. Then from \cref{lemma:block_encoding_inversion}, with $\beta\geq \|\tilde{A}^{-1}\|$ a $(2\beta, \epsilon)$-block encoding of $\tilde{A}^{-1}$, $V$ can be constructed using $\mathcal{O}\left(\alpha_{A}\beta\log\left(\frac{1}{\epsilon}\right)\right)$ queries to c$U_{A}$ and its inverse. Using \cref{def:block_encoding}
 and the reverse triangle inequality we have
\begin{equation}
\left\|\tilde{A}^{-1} - 2\beta \big( (\bra{0}\otimes I)\, V\,(\ket{0} \otimes I) \big) \right\| \leq \epsilon  
\,\Rightarrow\,\left\| A^{-1}  -2\beta\big( (\bra{0}\otimes I)\, V\,(\ket{0} \otimes I) \big) \right\| \leq \epsilon +\|\tilde{A}^{-1}-A^{-1}\|
\end{equation}
so that $V$ is a $(2\beta, \epsilon+\|\tilde{A}^{-1}-A^{-1}\|)$-block encoding of $A^{-1}$. 

Since $\|\delta A\|\leq \delta \leq \frac{1}{2\|A^{-1}\|}$, we have
\begin{align}
\|\tilde{A}^{-1}\| &\leq \|A^{-1}\|\|(I+A^{-1}\delta A)^{-1}
\|\nonumber\\ 
&= \|A^{-1}\| \left\|\sum_{k=0}^\infty(-A^{-1}\delta A)^{k}
\right\|\nonumber\\
&\leq  \frac{\|A^{-1}\|}{1-\|A^{-1}\|\|\delta A\|}\nonumber\\
&\leq 2\|A^{-1}\|\nonumber\\
&\leq 2\alpha_{A^{-1}}.
\end{align}
We can therefore choose $\beta=2\alpha_{A^{-1}}$. We also have
\begin{align}
    \|A^{-1}-\tilde{A}^{-1}\|&\leq\|A^{-1}\|\|I-(I+A^{-1}\delta A)^{-1}\|\nonumber\\
    &\leq \|A^{-1}\| \left\|  \sum_{k=1}^\infty (-1)^{k+1}(A^{-1}\delta A)^k\right\|\nonumber\\
    &\leq 2\delta\|A^{-1}\|^2.
\end{align}
Then $V$ is a $(4\alpha_{A^{-1}},\epsilon+2\delta\|A^{-1}\|)$-block encoding of $A^{-1}$.
\end{proof}

We now construct block-encodings of the block matrices $\widetilde{F}_0$ and $\widetilde{F}_1$ resulted from the linearization of \cref{eq:problem_1_rescaled}.

\begin{lemma}[Block encoding of $\widetilde{F}_0$]\label{lemma:complexity_F_0_tilde}
    Given a real number $\alpha \geq \max_j|(G_0)_j|$ and access to an oracle $U_{G_0}$ with the action $U_{G_0}\ket{0^{b_1}}\ket{j}=\ket{[(G_0)_j]}\ket{j}$, an $(\alpha,0)$-block encoding of the matrix
    $\widetilde{F}_0$ defined in \cref{eq:F0_tilde} for the rescaling given in \cref{eq:rescaling_definition} can be constructed using $\mathcal{O}(1)$ calls to $U_{G_0}$. 
\end{lemma}
\begin{proof}
The circuit $U_{\widetilde{F}_0} = X_a \, U_{G_0, (b_1,j)}^\dagger \, cR_y\!\left(2\arcsin\!\left(\tfrac{(G_0)_j}{\alpha}\right)\right)_{( b_1\rightarrow a)} \, U_{G_0, (b,j)}$ is an exact block encoding of $\widetilde{F}_0/\alpha$. This is clear from its action on $\ket{0}\ket{0^{b_1}}\ket{j}$.

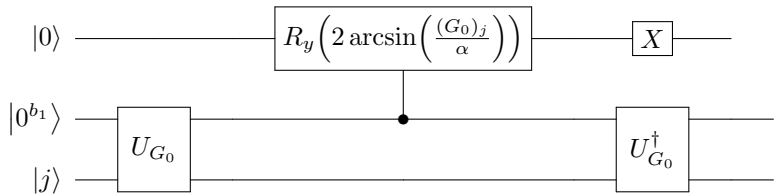
\[
\Qcircuit @C=1.6em @R=1.4em {
%----- ancilla line -----
\lstick{\ket{0}}   & \qw
                   & \qw  
                   & \gate{R_y\!\left(2\arcsin\!\left(\tfrac{(G_0)_j}{\alpha}\right)\right)} 
                   & \qw
                   & \gate{X}
                   & \qw \\
%----- f register -----
\lstick{\ket{0^{b_1}}} & \multigate{1}{U_{G_0}}
                   & \qw
                   & \ctrl{-1}
                   & \qw
                   & \multigate{1}{U_{G_0}^{\dagger}}
                   & \qw
                   & \qw \\
%----- j register -----
\lstick{\ket{j}}   & \ghost{U_{G_0}}
                   & \qw
                   & \qw
                   & \qw
                   & \ghost{U_{G_0}^{\dagger}}
                   & \qw
                   & \qw
}
\]

We have
\begin{align}
\ket{0}\ket{0^{b_1}}\ket{j}
 \xrightarrow{\,U_{G_0}\,}
\ket{0}\ket{[(G_0)_j]}\ket{j}
 \xrightarrow{cR_y\!\left(2\arcsin\!\left(\tfrac{(G_0)_j}{\alpha}\right)\right)}
\left(
\sqrt{1-\frac{|(G_0)_j|^2}{\alpha^2}}\ket{0}
+\frac{(G_0)_j}{\alpha}\ket{1}
\right)\ket{[(G_0)_j]}\ket{j}\nonumber\\
 \xrightarrow{\, XU_{G_0}^\dagger\,}
\!\left( \frac{(G_0)_j}{\alpha}\ket{0}
+
\sqrt{1-\frac{|(G_0)_j|^2}{\alpha^2}}\ket{1}
\right)\ket{0^{b_1}}\ket{j},
\end{align}
such that
\begin{align}
\bra{0}\bra{0^{b_1}}U\ket{0}\ket{0^{b_1}}\ket{j} =\frac{(G_0)_j}{\alpha}\ket{j} = \frac{(F_0)_j}{\alpha}\ket{j},
\end{align}
since $G_0=F_0$ from \cref{eq:problem_1_rescaled}. The rotation angle is well defined since $\alpha \geq \max_j|(G_0)_j|$ and the circuit uses one call to $U_{G_0}$ and one call to $U_{G_0}^{\dagger}$.
\end{proof}

\begin{lemma}[Block encoding of $\widetilde{F}_1$]\label{lemma:complexity_F_1_tilde}
    Given a unitary $U_{G_1}$ that is a $(\beta,0)$-block encoding of $G_1$ and $\nu\in\mathbb{R}_+$ for the rescaling given in \cref{eq:rescaling_definition}, a $(\nu\beta,0)$-block encoding of ${\bf e}_1\otimes\widetilde{F}_1\in\mathbb{C}^{n^2\times n^2}$ for $\widetilde{F}_1$ and ${\bf e}_1$ defined in \cref{eq:F1_tilde} and \cref{eq:padded_Psi_N} respectively, can be constructed using $\mathcal{O}(1)$ calls to $U_{G_1}$.
\end{lemma}

\begin{proof} Let $b_2 = \lceil \log_{2} n \rceil, P =  
\sum_k \ket{k}\bra{k} \otimes \ket{0}\bra{k}$ and $U$ be an $(\nu\beta,0)$-block encoding of the matrix $P(\mathds{1}\otimes F_1)\in\mathbb{C}^{n^2\times n^2}$. Then the following circuit constructs an $(\nu\beta,0)$-block encoding of ${\bf e}_1\otimes\widetilde{F}_1$. 
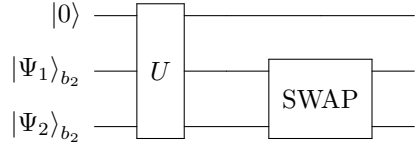
\[
\Qcircuit @C=1.6em @R=1.2em {
%----- ancilla -----
\lstick{\ket{0}}       & \multigate{2}{U}
                   & \qw
                   & \qw
                   & \qw
\\
%----- index 1 -----
\lstick{\ket{\Psi_1}_{b_2}}  & \ghost{U}
                   & \qw
                   & \multigate{1}{\text{SWAP}}
                   & \qw
\\
%----- index 2 -----
\lstick{\ket{\Psi_2}_{b_2}}  & \ghost{U}
                   & \qw
                   & \ghost{\text{SWAP}}
                   & \qw
}
\]
This follows from
\begin{equation}   
P \, (\mathds{1}\otimes F_1) = \sum_k \ket{k}\bra{k} \otimes \ket{0}\bra{k} F_1,
\end{equation}  
and 
\begin{equation}\label{eq:e_tensor_F_1}
\sum_k \ket{0}\bra{k} \otimes \ket{k}\bra{k} F_1 = {\bf e}_1\otimes\widetilde{F}_1
\end{equation}
which can be understood from
\begin{align}
\ket{0}\bra{k}\otimes \ket{k}\bra{k}F_1&=\begin{bmatrix}
0 & \cdots & 0 & \ket{k}\bra{k}F_1 & 0 & \cdots & 0 \\[6pt]
0 & \cdots & 0 & 0 & 0 & \cdots & 0 \\[6pt]
\vdots &  & \vdots & \vdots &  &  & \vdots \\[6pt]
0 & \cdots & 0 & 0 & 0 & \cdots & 0
\end{bmatrix}\nonumber\\
\ket{k}\bra{k}F_1 &= \begin{bmatrix}
0 & 0 & \cdots & 0 \\
\vdots & \vdots & & \vdots \\
0 & 0 & \cdots & 0 \\
(F_1)_{k0} & (F_1)_{k1} & \cdots & (F_1)_{k,n-1} \\
0 & 0 & \cdots & 0 \\
\vdots & \vdots & & \vdots \\
0 & 0 & \cdots & 0
\end{bmatrix}.
\end{align}

Since $F_1=\nu G_1$ from \cref{eq:problem_1_rescaled} a $(\nu\beta,0)$-block encoding of $\mathds{1}\otimes F_1$ can be constructed using \cref{lemma:block_encoding_arithmetic} which requires $\mathcal{O}(1)$ calls to $U_{G_1}$. $P$ has at most one nonzero entry in each row and column. 
It follows from \cref{lemma:block_encoding_sparse} that a $(1,0)$-block encoding of $P$ 
can be constructed using $\mathcal{O}(1)$ calls to the oracles $O_r, O_c, O_A$, 
all of which can be implemented efficiently for $P$. The $(\nu\beta,0)$-block encoding of $P(\mathds{1}\otimes G_1)$, $U$, can then be constructed using \cref{lemma:block_encoding_arithmetic} which requires $\mathcal{O}(1)$ calls to $U_{F_1}, O_r, O_c$ and $O_A$. No additional calls to $U_{G_1}$ are required for the $\text{SWAP}$ gate or to construct the oracles $O_r, O_c, O_A$.

\end{proof}

In the next Lemmas, we use the block encodings of $\widetilde{F}_0$ and $\widetilde{F}_1$ 
to construct the block encodings of $B_j^{(0)}$ and $B_{j+1}^{(1)}$ 
in \cref{eq:carleman-assembly2}.

\begin{lemma}[Block encoding of $B_j^{(0)}$ and $B_{j+1}^{(1)}$] \label{lemma:complexity_B}
Given a real number $\alpha \geq \max_j|(G_0)_j|$, $\nu\in\mathbb{R}_+$ for the rescaling given in \cref{eq:rescaling_definition}, access to an oracle $U_{G_0}$ with the action $U_{G_0}\ket{0^{b_1}}\ket{j}=\ket{[(G_0)_j]}\ket{j}$ and a unitary $U_{G_1}$ that is a $(\beta,0)$-block encoding of $G_1$, a $(j\alpha,0)$-block encoding of $B_j^{(0)}$ and a $(j\nu\beta,0)$-block encoding of ${\bf e}_1 \otimes B_{j+1}^{(1)}$ for $B_j^{(0)}$ and $B_{j+1}^{(1)}$ defined in \cref{eq:carleman-assembly2} can be constructed using $\mathcal{O}(j)$ calls to $U_{G_0}$ and $U_{G_1}$ respectively.  
\end{lemma}

\begin{proof} 
From \cref{lemma:complexity_F_0_tilde} and \cref{lemma:complexity_F_1_tilde} an $(\alpha,0)$-block encoding of $\widetilde{F}_0$ and a $(\nu\beta,0)$-block encoding of ${\bf e}_1\otimes\widetilde{F}_1$ can be constructed using $\mathcal{O}(1)$ calls to $U_{G_0}$ and $U_{G_1}$ respectively. From \cref{lemma:block_encoding_arithmetic}, an $(\alpha,0)$-block encoding of 
$T^{j,k}_0=I^{\otimes k} \otimes i \widetilde{F}_0 \otimes I^{\otimes (j-1-k)}$ can be constructed using $\mathcal{O}(1)$ calls to $U_{G_0}$. Similarly, a $(\nu\beta,0)$-block encoding of $
T^{j,k}_1 = I^{\otimes k} \otimes i({\bf e}_1 \otimes \widetilde{F}_1) \otimes I^{\otimes (j-1-k)}$
can be constructed using $\mathcal{O}(1)$ calls to $U_{G_1}$. 

A $(\nu\beta,0)$-block encoding of $\bar{T}^{j,k}_1 = {\bf e}_1 \otimes\left(I^{\otimes k} \otimes i\widetilde{F}_1 \otimes I^{\otimes (j-1-k)}\right)$ can be constructed from the $(\nu\beta,0)$-block encoding of $T^{j,k}_1$ with a $\text{SWAP}$ gate. This follows the definition of these operators in bra-ket notation;
\begin{align}
\widetilde{F}_1 &= \sum_{i,q=0}^{n} \ket{q}\bra{q}\otimes\bra{i}(F_1)_{i,q}\nonumber\\
    T^{j,k}_1 &=\sum_{i,q,l=0}^{n,n,n^{j-1}} \ket{l}\bra{l}\otimes\ket{0}\bra{q}\otimes \ket{q}\bra{i} (F_1)_{i,q}\nonumber  \\
    \bar{T}^{j,k}_1  &= \sum_{i,q,l=0}^{n,n,n^{j-1}}  \ket{0}\bra{l}\otimes\ket{l}\bra{q}\otimes \ket{q}\bra{i} (F_1)_{i,q}.
\end{align}

\noindent Then, using \cref{lemma:block_encoding_arithmetic} the $j$ term sums $B_j^{(0)}= \sum_k T^{j,k}_0$ and ${\bf e}_1\otimes B_{j+1}^{(1)}=\sum_k\bar{T}^{j,k}_1$ can be be exactly block encoded with scaling factors $j\alpha$ and $j\nu\beta$ respectively and using $\mathcal{O}(j)$ calls to $U_{G_0}$ and $U_{G_1}$ respectively.

\end{proof}

\begin{lemma}[Taylor truncation bound]\label{lemma:bound_on_T}
Let $\|{\bf L}_N^{\rm T}\|h\leq 1$ and  $me^2/(k+1)!\leq 1$. Then, for $W_{\ell,k}$ defined in \cref{def:W_l} and all $0\leq\ell\leq k$
\begin{equation}
    \left\|W_{\ell,k}\right\| \leq e,
\end{equation}
and for the $k$th-order truncated Taylor series of $V = e^{{\bf L}_N^{\rm T}h}$, with  $j\leq m$
\begin{equation}
\norm{V_k^j} \leq C({\bf L}_N^{\rm T})\left(1+\frac{(e-1)je^2}{(k+1)!}\right),
\end{equation}
where $C({\bf L}_N^{\rm T})$ is given in \cref{def:C(A)}.
\end{lemma}

\begin{proof}
The bound on $\|W_{\ell,k}\|$ follows from $h\|{\bf L}_N^{\rm T}\|<1$,
\begin{equation}
    \left\|W_{\ell,k}\right\| \leq \sum_{i=0}^{k-\ell} \frac{\ell! \, }{(\ell + i)!} \leq  \sum_{i=0}^{k-\ell} \frac{1}{\binom{\ell+i}{i} i!}\leq \sum_{i=0}^{k-\ell} \frac{1}{i!} \leq e.
\end{equation}
The bound on $\|W_{0,k}^j\|= \norm{V_k^j}$ follows from \cref{lemma:V_ke^-Lh_bound} and \cref{def:C(A)},
\begin{align}
    \left\|V_k^j\right\| &\leq  \left\|e^{{\bf L}_N^{\rm T}hj}\right\|+\left\|V_k^j -e^{{\bf L}_N^{\rm T}hj}\right\| \nonumber\\ 
&\leq  \left\|e^{{\bf L}_N^{\rm T}hj}\right\| \left(1 +   
\norm{V_k^j e^{-{\bf L}_N^{\rm T}hj} -I} \right)\nonumber\\
&\leq C({\bf L}_N^{\rm T})\left(1+\frac{(e-1)je^2}{(k+1)!}\right).
\end{align}
\end{proof}

\begin{lemma}[Norm of $(I-\mathcal{M}_1)^{-1}$]\label{lemma:I-M_1_norm} Let $h\|{\bf L}_N^{\rm T}\|_2\leq 1$. Then, for $\mathcal{M}_1$ defined in \cref{eq:mathcal_M}, \begin{align}
    \|(I-\mathcal{M}_1)^{-1}\| \leq k
\end{align}
\end{lemma}
\begin{proof} $\|\mathcal{M}_1\|= \max_{0\leq i\leq k-1}\left(\frac{h\|{\bf L}_N^{\rm T}\|}{i+1}\right) =h\|{\bf L}_N^{\rm T}\| \le1$. Then using subadditivity and that $\mathcal{M}_1$ is nilpotent with $\mathcal{M}_1^{k+1}=0$, we have
\begin{align}
\|(I-\mathcal{M}_1)^{-1}\| =\left\|\sum_{j=0}^k \mathcal{M}_1^j\right\| \leq \sum_{j=0}^k \|\mathcal{M}_1\|^j \leq k. \end{align}
\end{proof}

\begin{lemma}[Norm of $\mathcal{A}^{-1}$]\label{lemma:A_norm} Let $h\|{\bf L}_N^{\rm T}\|\leq 1$ and $me^2/(k+1)!\leq 1$. Then, for $\mathcal{A}$ defined in \cref{eq:mathcal_M}, \begin{align}
    \|\mathcal{A}^{-1}\| \leq  2emC({\bf L}_N^{\rm T})\left(1+\frac{(e-1)me^2}{(k+1)!}\right).
\end{align}
\end{lemma}

\begin{proof} 
Following \cite{krovi2023improved}, from \cref{eq:M_2_M_1_action} we have for $\ket{l,v}$ an arbitrary normalised basis state,
\begin{equation}
    \left(\mathcal{M}_{2}(I-\mathcal{M}_{1})^{-1}\right)^q \ket{l,v}= \ket{0}\otimes W_{l,k}^q\ket{v}.
\end{equation}
Then with $r=\min\{q, \max(0, m-j)\}$ and $q\in \{0,,1,\dots,2m-1\}$,
\begin{equation}
\mathcal M^{\,q}\,|l,j,v\rangle
=
\begin{cases}
|0,\,j+q\rangle \,\otimes\, \bigl(V_k^{\,r-1} W_{l,k}\bigr)\,|v\rangle, & r\ge 1,\\[4pt]
|l,\,j+q,\,v\rangle, & r=0,
\end{cases}
\end{equation}
so that 
\begin{equation}
    \|\mathcal{M}^q\| = \sup_{l,j,v} \| M^{\,q}\ket{l,j,v}\| \leq \begin{cases}
    1,\quad q = 0 \\[4pt]
        \| V_k^{q-1}\| \|W_{l,k}\|, \quad 0< q\leq m\\[4pt] 
         \| V_k^{m-1}\| \|W_{l,k}\| \quad m< q\leq 2m-1.
    \end{cases}   
\end{equation}
Then from \cref{lemma:bound_on_T} we have
\begin{equation}
     \|\mathcal{M}^q\| \leq C({\bf L}_N^{\rm T})\left(1+\frac{(e-1)m e^2}{(k+1)!}\right)e,
\end{equation}
and since $\mathcal{M}$ is nilpotent with $\mathcal{M}^{2m}=0$,
\begin{equation}
     \|\mathcal{A}^{-1}\|=\|(I-\mathcal{M})^{-1}\| \leq \sum_{q=0}^{2m-1}\|\mathcal{M}^q\| \leq  2emC({\bf L}_N^{\rm T})\left(1+\frac{(e-1)me^2}{(k+1)!}\right).
\end{equation}

\end{proof}

\section{Gershgorin circle theorem for partitioned matrices}
We now state a lemma that is useful in bounding the eigenvalues of partitioned matrices. 
\begin{lemma}[Gershgorin Circle Theorem for Partitioned Matrices {\cite[Theorem~2]{feingold1962block}}]\label{lemma:greshgorin}
Let $A\in\mathbb{C}^{n\times n}$ be partitioned into blocks
\[
A=\begin{bmatrix}
A_{1,1} & A_{1,2} & \cdots & A_{1,N}\\
A_{2,1} & A_{2,2} & \cdots & A_{2,N}\\
\vdots  & \vdots  & \ddots & \vdots \\
A_{N,1} & A_{N,2} & \cdots & A_{N,N}
\end{bmatrix},
\]
where $A_{j,k}\in\mathbb{C}^{n_j\times n_k}$ and $\sum_{j=1}^N n_j=n$. Then every
eigenvalue $\lambda$ of $A$ satisfies
\[
\Bigl\|\bigl(A_{j,j}-\lambda I_{n_j}\bigr)^{-1}\Bigr\|^{-1}
\;\le\;
\sum_{\substack{k=1\\k\neq j}}^{N}\,\|A_{j,k}\|
\quad\text{for at least one } j\in\{1,\dots,N\},
\]
where $\|\cdot\|$ is any (subordinate) matrix norm and $I_{n_j}$ is the $n_j\times n_j$
identity.
\end{lemma}

\section{Proof of \cref{lemma:rescaling_C(L)<1}}\label{app:rescaling_C(L)<1}

{\noindent \bf  \cref{lemma:rescaling_C(L)<1}} (Stability of the linearized ODE) \textit{Let $\{p,q\} \in [1,\infty)$ be such that $1/p+1/q=1$ and for \cref{prob:problem_statement_formal} suppose that $G_0, G_1$ and ${\bf u}_0$  satisfy  
    \begin{equation}
        \tilde{\mu}_0 := \min_j \Im{(G_0)_j}\ge 0 \quad \text{and} \quad
         R_p := \frac{\|G_1\|_{{\rm row},q} \|e^{i{\bf u}_0}\|_p}{\tilde{\mu}_0}< \min \left\{1,\frac{\|e^{i{\bf u}_0}\|_p}{\|e^{i{\bf u}_0}\|_2}\right\},
    \end{equation}  
Then the rescaling of the variable ${\bf u}$ given in \cref{eq:rescaling_definition} with $\nu =\|e^{i{\bf u}_0}\|_p/R_p$  produces a rescaled differential equation with
\begin{equation}
    C({\bf L}_N^{\rm T})\leq 1 \quad \text{and}\quad \gamma:=\|e^{i{\bf x}_0}\|<1,
\end{equation}
where ${\bf L}_N^{\rm T}$ is defined in \cref{eq:LN} on the rescaled problem. }

\begin{proof}
We follow a similar argument to Ref.~\cite{costa2025further}. From \cref{def:C(A)}, \cref{eq:sec_def} and \cref{lemma:expA}, the eigenvalues of $\frac{{\bf L}_N^{\rm T}+\left({\bf L}_N^{\rm T}\right)^{\dagger}}{2}$ being non-positive implies $C({\bf L}_N^{\rm T}) \leq 1$. We now bound these eigenvalues using \cref{lemma:greshgorin}. We have
\begin{equation}
\frac{1}{2}\!\left(\mathbf{L}_N^{\rm T}+\mathbf{L}_N^{T\dagger}\right)
=
\frac{1}{2}
\begin{pmatrix}
B_1^{(0)}+B_1^{(0)\dagger} &  B_2^{(1)} &  &  &  \\
 B_2^{(1)\dagger} & B_2^{(0)}+B_2^{(0)\dagger} &  B_3^{(1)} &  &  \\
 &  & \ddots  &  &  \\
 &  &  &  B_{N-1}^{(1)\dagger} & B_{N-1}^{(0)}+B_{N-1}^{(0)\dagger} &  B_{N}^{(1)} \\
 &  &  &  &  B_{N}^{(1)\dagger} & B_{N}^{(0)}+B_{N}^{(0)\dagger}
\end{pmatrix}.
\end{equation}
Set $S_j$ be the sum of the $p$-norms of the off diagonal blocks in row $j$. Then,
\begin{equation}\label{eq:sum_off_diag_blocks_bound}
S_j = \frac{1}{2}
\begin{cases}
   \norm{ B_2^{(1)}}_p 
      \leq \norm{\widetilde{F}_1}_p, & j = 1 \\[6pt]
   \norm{ B_j^{(1)\dagger}}_p 
   + \norm{ B_{j+1}^{(1)}}_p 
      \leq (2j-1)\norm{\widetilde{F}_1}_p, & 1 < j < N \\[6pt]
   \norm{ B_N^{(1)\dagger}}_p 
      \leq (N-1)\norm{\widetilde{F}_1}_p 
      \leq (2N-1)\norm{\widetilde{F}_1}_p, & j = N,
\end{cases}
\end{equation}
where we have used $\|B_{j+1}^{(1)}\|_p \leq j\|\widetilde{F}_1\|_p$. 
We therefore have for all $j\in[1,N]$,
\begin{equation}
       S_j \leq   \frac{1}{2} (2j-1) \norm{\widetilde{F}_1}_p.
\end{equation}
From the definition of the logarithmic norm given in \cref{eq:sec_def} we have that $\mu(B_j^{(0)})$ is the maximum eigenvalue of $\frac{B_j^{(0)}+B_j^{(0)\dagger}}{2}$ which is the $j$th diagonal term of $\frac{1}{2}\!\left(\mathbf{L}_N^{\rm T}+\mathbf{L}_N^{T\dagger}\right)$. Then since $F_0=G_0,$
\begin{equation}
    \mu(B_j^{(0)}) = j\mu(i \widetilde{F}_0) = -j\min_j \Im{(F_0)_j}  =-j\tilde{\mu}_0.
\end{equation}
By \cref{lemma:greshgorin} an upper bound on the eigenvalues of $\frac{1}{2}\!\left(\mathbf{L}_N^{\rm T}+\mathbf{L}_N^{T\dagger}\right)$ is therefore
\begin{equation}
   \max_{j\in[1,N]}\left\{ -j\tilde{\mu}_0+\frac{1}{2} (2j-1)\norm{\widetilde{F}_1}_p \right\}.
\end{equation}
A sufficient condition for these eigenvalues to be non-positive is
\begin{equation}
    \tilde{\mu}_0\geq\ \nu\norm{G_1}_{\rm row,q},
\end{equation}
since $\tilde{\mu}_0\geq0$ and using \cref{eq:problem_1_rescaled} and \cref{lemma:F_1_p_F_1_row_q}. This condition can be rewritten as
\begin{equation}
 \nu \leq  \frac{ \|e^{i{\bf u}_0}\|_p}{R_p},
\end{equation}
which is satisfied by the assumption on $\nu$ given in Lemma, proving that $C({\bf L}_N^{\rm T})\le1$.\\

When $p\in [1,2]$ we have $\|e^{i{\bf u}_0}\|_2<\frac{\norm{e^{i{\bf u}_0}}_p}{R_p} =\nu $ since $R_p<1$ and $\|e^{i{\bf u}_0}\|_2\leq\|e^{i{\bf u}_0}\|_p$. When $p\in (2,\infty)$ we have $R_p<\min \left\{1,\frac{\|e^{i{\bf u}_0}\|_p}{\|e^{i{\bf u}_0}\|_2}\right\}=\frac{\|e^{i{\bf u}_0}\|_p}{\|e^{i{\bf u}_0}\|_2}$ so that $\|e^{i{\bf u}_0}\|_2<\|e^{i{\bf u}_0}\|_p/R_p=\nu$.  Then from \cref{eq:rescaling_norm_input} for all $p\in[1,\infty)$ the rescaled differential equation satisfies $ \gamma=\|e^{i{\bf x}_0}\|_2 <1$.
\end{proof}

\section{Proof of \cref{lemma:sol_growth_bound}}\label{app:norm_stabODE}

{\noindent \bf \cref{lemma:sol_growth_bound}}
(Solution growth bound)
\textit{Given $\nu \in \mathbb{R}_+$ for the rescaling defined in \cref{eq:rescaling_definition}
and $p \in [1,\infty)$,
\begin{equation}
   C_p({\bf L}_N^{\rm T})=\sup_{t\in[0,T]}\norm{\exp({\bf L}_N^{\rm T}t)}_p \leq \exp\left(T\left(N\nu\|G_1\|_{\rm row, q} + \max\{-\tilde{\mu}_0, -N\tilde{\mu}_0\}\right)\right) := 
   \Gamma_p(N,T,\nu)
\end{equation}
for ${\bf L}_N^{\rm T}$ defined in \cref{eq:LN} and $\tilde{\mu}_0=\min \Im{(G_0)_j}$.}

\begin{proof} 
    Following the proof of \cref{lemma:LN_le_F} we have 
    \begin{equation}
        {\bf L}_N^{\rm T}  = H_0 + H_1,   
    \end{equation}
where 
\begin{equation}
H_0 = \sum_{j=1}^{N} | j \rangle \langle j | \otimes B_j^{(0)} , \quad 
H_1 = \sum_{j=1}^{N-1} | j \rangle \langle j + 1 | \otimes B_j^{(1)}.
\end{equation}
We know from \cref{lemma:Trian_logN}, that
\begin{equation}
 \mu_p( {\bf L}_N^{\rm T}) \leq \mu_p(H_0) + \mu_p(H_1) \leq \mu_p(H_0) + \|H_1\|_p,   
\end{equation}
where we have used \cref{lemma:muA} in the last inequality above. Now for the norm of $H_1$ we have $\|H_1\|_p=\max_{j}\left(\|B_j^{(1)}\|_p\right)$. We also know from the definition of $B_j^{(1)}$, \cref{eq:carleman-assembly2} that
\begin{align}
    \|B_{j}^{(1)}\|_p = \Bigg\|\sum_{\nu=1}^j i \widetilde{F}_1\otimes I^{\otimes(\nu-1)}\Bigg\|_p\leq \sum_{\nu=1}^j\Bigg\| i \widetilde{F}_1\otimes I^{\otimes(\nu-1)}\Bigg\|_p = \sum_{\nu=1}^j\|\widetilde{F}_1\|_p = j\|\widetilde{F}_1\|_p.
\end{align}
Then
\begin{equation}
    \|H_1\|_p \leq N\|\widetilde{F}_1\|_p.
\end{equation}
We can now follow a similar set of steps for the logarithmic norm of $H_0$. Since it gives the maximum eigenvalue of the hermitian part of $H_0$, we have
\begin{equation}
    \mu_p{(H_0)}  = \max_j\mu_p(B_j^{(0)}),
\end{equation}
where now from the definition given in \cref{eq:carleman-assembly2} for $B_j^{(0)}$ we conclude that
\begin{equation}
    \mu_p{(H_0)}  = \max_jj\mu_p(i\widetilde{F}_0) = \max_j-j\tilde{\mu}_0 = \max\{-\tilde{\mu}_0, -N\tilde{\mu}_0\},
\end{equation}
since $\tilde{\mu}_0=\min \Im{(G_0)_j}=\min \Im{(F_0)_j}$ from \cref{eq:problem_1_rescaled}.
We then, have
\begin{equation}
\mu_p(\mathbf{L}_N^{\rm T}) \leq N\|\widetilde{F}_1\|_p + \max\{-\tilde{\mu}_0, -N\tilde{\mu}_0\} = N\nu\|G_1\|_{\rm row, q} + \max\{-\tilde{\mu}_0, -N\tilde{\mu}_0\},
\end{equation}
where we have used \cref{lemma:F_1_p_F_1_row_q} and the definition of the rescaled coefficient matrix from \cref{eq:problem_1_rescaled}. The proof is completed using \cref{lemma:expA} and the fact that for sufficiently large $N$, $N\nu\|G_1\|_{\rm row, q} + \max\{-\tilde{\mu}_0, -N\tilde{\mu}_0\}\geq 0$.

\end{proof}

\section{Proof of \cref{colrollary:dissipative-result}  }\label{app:dissipative-result}

\noindent \textbf{\cref{colrollary:dissipative-result}}. (Total complexity with dissipative conditions)
\textit{Suppose the Fourier ODE \cref{prob:problem_statement_formal} also satisfies
\begin{equation}
    \alpha \in \mathcal{O}(\norm{G_0}_{\infty}),\quad 
    \beta \in \mathcal{O}(\norm{G_1}),\quad
    K\in \mathcal{O}(1), \quad
    (1-R_p)^{-1} \in \mathcal{O}(1). 
\end{equation}
Then the total query complexity for \cref{theorem:dissipative_result} is 
\begin{equation}
\tilde{\mathcal{O}}\left(\frac{T^{3/2}\norm{G_0}_{\infty}^{3/2}
     \norm{{\bf d}}}{\epsilon}\left(1+\frac{\norm{G_1}}{\|G_1\|_{\rm row, q}}\right)
     \left(1+ \left(\frac{\tilde{\mu}_0}{\|G_1\|_{\rm row,q}}\right)^K\right)
    \right)
\end{equation}
where $\tilde{\mathcal{O}}$ suppresses factors in
\begin{equation}
    {\rm polylog}
     \left(n,T,\norm{G_0}_{\infty},\norm{{\bf d}},1/\epsilon,
     \left(1+ \left(\frac{\tilde{\mu}_0}{\|G_1\|_{\rm row,q}}\right)^K\right)
     \right).
\end{equation}. }

\begin{proof}

 The size of the parameters in \cref{theorem:dissipative_result} can be simplified as; 
\begin{align}
N &= \mathcal{O}\left(\log\left(\frac{s\|{\bf d}\|_q}{\epsilon}\right)\right)\nonumber\\
k &= \mathcal{O} \left( \log \left(\frac{NTz\norm{G_0}_{\infty}}{\epsilon }\right)\right)\nonumber\\
z &= \mathcal{O}\left( \|{\bf d}\|s\right)\nonumber\\
s &= \mathcal{O}\left(1+\left(\frac{\tilde{\mu}_0}{\|G_1\|_{\rm row,q}}\right)^K\right)
\end{align}
using $\tilde{\mu}_0 \le \norm{G_0}$. The complexity of queries to $U_{G_0}$ and $U_{G_1}$ is then given by
\begin{multline}
\label{eq:overallquery1}
    \mathcal{O}\left( \frac{1}{\epsilon}  N^{3/2}
k^{7/2}(T\alpha)^{3/2} z\left(1+\frac{\tilde{\mu}_0\beta}{\alpha\|G_1\|_{\rm row, q}}\right)\cdot \log k \log^2 \left(\frac{(NT\alpha)^{3/2}\sqrt{k}z}{\epsilon} \right)\right)\\ 
    \in \tilde{\mathcal{O}}\left( \frac{T^{3/2}\norm{G_0}_{\infty}^{3/2} z 
    }{\epsilon}\left(1+\frac{\norm{G_1}}{\|G_1\|_{\rm row, q}}\right) \right) \\
    \in 
    \tilde{\mathcal{O}}\left(\frac{T^{3/2}\norm{G_0}_{\infty}^{3/2}
     \norm{{\bf d}}}{\epsilon}\left(1+\frac{\norm{G_1}}{\|G_1\|_{\rm row, q}}\right)
     \left(1+ \left(\frac{\tilde{\mu}_0}{\|G_1\|_{\rm row,q}}\right)^K\right)
    \right),
\end{multline}
where we have used $\norm{{\bf d}}_q \le \norm{{\bf d}}_1 \le \sqrt{Kn^K}\norm{{\bf d}}_1$ and therefore ${\rm polylog}(\norm{{\bf d}}_q) \in {\rm polylog}(n,\norm{{\bf d}})$ to remove $\|\textbf{d}\|_q$ from the suppressed $\text{polylog}$ factor.

The complexities of queries to $U_{u_0}$ and $U_{d}$ can similarly be simplified similarly as
\begin{align}
\mathcal{O}\left(\frac{N^{3/2}\sqrt{T \alpha}z}{\epsilon}\right) \in  \tilde{\mathcal{O}}\left(\frac{\sqrt{T}\norm{G_0}_{\infty}^{1/2}\norm{{\bf d}}}{\epsilon}\left(1+ \left(\frac{\tilde{\mu}_0}{\|G_1\|_{\rm row,q}}\right)^K\right)\right), \\
    \mathcal{O}\left(\frac{\sqrt{NT\alpha}z}{\epsilon}\right) \in 
    \tilde{\mathcal{O}}\left(\frac{\sqrt{T}\norm{G_0}_{\infty}^{1/2}\norm{{\bf d}}}{\epsilon}\left(1+\left(\frac{\tilde{\mu}_0}{\|G_1\|_{\rm row,q}}\right)^K\right)\right),
\end{align}
respectively. The total query complexity is the sum of queries to
$U_{G_0}$, $U_{G_1}$, $U_{u_0}$ and $U_{d}$.

\end{proof}

\section{Proof of \cref{colrollary:non-dissipative-result}}\label{app:non-dissipative-result}

\noindent \textbf{\cref{colrollary:non-dissipative-result}}. (Total complexity without dissipative conditions) \textit{If $p\geq 2$ then the total query complexity for \cref{thm:non-dissipative-result} is
\begin{equation}
\tilde{O}\left(\frac{r^{K\ell}\|\mathbf d\|_{q}^{1+\ell}
    \|e^{i{\bf u}_0}\|_p^{K(\ell +1)}
    }{\epsilon^{1+\ell}}\right).
\end{equation}}

\begin{proof}
Using $\nu = e^{\ell} r\|e^{i{\bf u}_0}\|_p$, 
the expression for $\widetilde{T}_{\rm max}$ can be written as
    \begin{equation}
        \widetilde{T}_{\rm max} = \min\left\{\frac{\ell}{\max\left\{\alpha,
    e^{\ell} r\|e^{i{\bf u}_0}\|_p\norm{G_1}_{{\rm row},q}\right\}(1+1/r)}, 
    \frac{\ln{r/e}}{\alpha+
    e^{\ell} r\|e^{i{\bf u}_0}\|_p\norm{G_1}_{{\rm row},q}}\right\}.
    \end{equation}
    The final time $T$ is then bounded by 
    \begin{equation}
        T \le \widetilde{T}_{\rm max} \le \frac{1}{\alpha}\ln \left(\frac{\nu}{r\|e^{i{\bf u}_0}\|_p}\right) = \frac{\ell}{\alpha} \quad
        \text{and} \quad 
        T \le \widetilde{T}_{\rm max} \le
        \frac{1}{\nu\norm{G_1}_{{\rm row},q}}\ln \left(\frac{\nu}{r\|e^{i{\bf u}_0}\|_p}\right)\le 
        \frac{\ell}{\nu\norm{G_1}_{{\rm row},2}},
    \end{equation}
    where we have used $\norm{G_1}_{{\rm row},q} \ge \norm{G_1}_{{\rm row},2}$ which is true for $q\le 2$ which is satisfied for $p\geq2$. We also have
    \begin{align}
   \Gamma(N,T,\nu) &= \exp\left(T\left(N\nu\|G_1\|_{\rm row, 2} + 
   \max\{-\tilde{\mu}_0, -N\tilde{\mu}_0\}\right)\right) \nonumber\\
   &\le \exp{N(T\nu\|G_1\|_{\rm row, 2})}\exp{N(T\|G_0\|_{\infty})} \nonumber\\
   &\le \exp{2N\ell}.
\end{align}
which implies
\begin{equation}
    \Gamma(N,T,\nu) 
    \in \mathcal{O}\left(\left(1+\dfrac{4K\,s\,\|\mathbf d\|_{q}}{r\epsilon}\right)^\ell\right)
\end{equation}
using \cref{eq:Nscaling}.
We then have
\begin{align}
    (\widetilde{T}_{\rm max})^{3/2}(\alpha +\nu\|G_1\|_{\rm row ,q})^{3/2}z
    &\le \left(\frac{\ln{r/e}}{\alpha+\nu\norm{G_1}_{{\rm row},q}}\right)^{3/2}
    (\alpha+ \nu\|G_1\|_{\rm row ,q})^{3/2}s\|{\bf d}\|\Gamma \nonumber\\
    &\le (\ln{r/e})^{3/2}s\|{\bf d}\|\Gamma \nonumber\\
    &\in \tilde{O}\left(\|{\bf d}\|_qs\left(1+\dfrac{4K\,s\,\|\mathbf d\|_{q}}{r\epsilon}\right)^\ell\right),
\end{align}
where in the last step we used $\|{\bf d}\|_q \ge \|{\bf d}\|$ since $q\le 2$.
For $\|e^{i{\bf u}_0}\|_p > 1$, we have $s = \nu^K$. 
Since $\ell \in \mathcal{O}(1)$, we have $\nu \in \mathcal{O}(r\|e^{i{\bf u}_0}\|_p)$, 
therefore
we can further simplify the above equation as 
\begin{align}
    (\widetilde{T}_{\rm max})^{3/2}(\alpha +\nu\|G_1\|_{\rm row ,q})^{3/2}z
    &\in \tilde{O}\left(\|{\bf d}\|_qe^{K\ell} r^K\|e^{i{\bf u}_0}\|_p^K
    \left(\dfrac{4K\|\mathbf d\|_{q}e^{K\ell} r^K\|
    e^{i{\bf u}_0}\|_p^K}{r\epsilon}\right)^\ell\right) \nonumber\\
    &\in \tilde{O}\left(\frac{e^{K\ell(1+\ell)} r^{K\ell}K^\ell 
    \|\mathbf d\|_{q}^{1+\ell}
    \|e^{i{\bf u}_0}\|_p^{K(1+\ell)}
    }{\epsilon^\ell}\right).
\end{align}
Then using $K,\ell \in \mathcal{O}(1)$ we obtain
\begin{align}
\label{eq:finalfactor1}
    (\widetilde{T}_{\rm max})^{3/2}(\alpha +\nu\|G_1\|_{\rm row ,q})^{3/2}z
    \in \tilde{O}\left(\frac{r^{K\ell}\|\mathbf d\|_{q}^{1+\ell}
    \|e^{i{\bf u}_0}\|_p^{K(\ell +1)}
    }{\epsilon^\ell}\right).
\end{align}
A similar calculation for $z\sqrt{\widetilde{T}_{\rm max}(\alpha +\nu\|G_1\|_{\rm row ,q})}$ yields
\begin{align}
\label{eq:finalfactor2}
    z\sqrt{\widetilde{T}_{\rm max}(\alpha +\nu\|G_1\|_{\rm row ,q})}
    \in \tilde{O}\left(s\|\mathbf d\|\Gamma\right)
    \in \tilde{O}\left(\frac{r^{K\ell}\|\mathbf d\|_{q}^{1+\ell}
    \|e^{i{\bf u}_0}\|_p^{K(\ell +1)}
    }{\epsilon^\ell}\right).
\end{align}
Finally using \cref{eq:finalfactor1} and \cref{eq:finalfactor2} the sum of the query complexities 
in \cref{thm:non-dissipative-result} 
simplifies to
\begin{equation}
\tilde{O}\left(\frac{r^{K\ell}\|\mathbf d\|_{q}^{1+\ell}
    \|e^{i{\bf u}_0}\|_p^{K(\ell +1)}
    }{\epsilon^{1+\ell}}\right).
\end{equation}
\end{proof}

\section{Variable names and conventions}
We generally denote scalars using lower case letters, e.g. $c$, vectors using lower-case bold letters, e.g. ${\bf u}$, and matrices/operators
using upper-case letters, such as A. 

\begin{table}[h]
\label{table:symbols}
\centering
\begin{center}
\textbf{TABLE OF SYMBOLS USED}
\end{center}
\renewcommand{\arraystretch}{1.2}
\setlength{\tabcolsep}{12pt}
\begin{tabular}{ll}
\hline
{\bf Symbol} & {\bf Meaning} \\
\hline
\multicolumn{2}{r}{{\bf Problem description symbols}} \\
\hline
$n$ & Dimension of the problem \\
$t$ & Coordinate in time \\
$\mathbf{u}$ & Solution vector for the nonlinear ODE \\
$\textbf{u}_0$ &  Initial condition to the the nonlinear ODE \\
$G_1, G_0$ & Coefficient matrices for the nonlinear ODE \\
$g$ & Readout function \\
$K$ & Fourier order of the readout function \\
$d_\textbf{j}$ & Fourier coefficients of the readout function\\
$\textbf{d}$ & Vector of all Fourier coefficients of the readout function \\
$T$ & Final time of evolution \\
 $R_p,\tilde{\mu}_0$ & Fixed parameters of the nonlinear ODE problem \\
\hline
\multicolumn{2}{r}{{\bf Approach symbols}} \\
\hline

$\nu$ & Rescaling parameter\\
$\mathbf{x}$ & Solution vector for the rescaled nonlinear ODE \\
$\mathbf{x}_0$ & Initial condition for the rescaled nonlinear ODE \\
$F_1, F_0$ & Rescaled coefficient matrices for the nonlinear ODE \\
$\gamma$ & 2-norm of $e^{i\textbf{x}_0}$\\
$f$ & Rescaled readout function \\
$c_\textbf{j}$ & Fourier coefficients of the rescaled readout function\\
$\textbf{c}$ & Vector of all Fourier coefficients of the rescaled readout function \\
$T$ & Final time of evolution \\
   $\Psi_j$ & $j^{\text{th}}$ component of the solution to the linearized ODE \\
   $N$ & Truncation order of the linearized ODE\\
    $\Psi^{(N)}$ & Solution vector to the order $N$ truncated linearized ODE \\
    $m,h$ & Time step discretization variables \\
  $\Phi_j$ & $j^{\text{th}}$ component of the solution to the linearized algebraic system \\
 $k$ & Taylor series truncation order for  linearized algebraic system \\
 %  $f^{(N)}$  &  Koopman approximation of $f$ \\
 % $\hat{f}^{(N)}$  &  Koopman–Taylor approximation of $f$  computed by the quantum algorithm\\

\hline
\multicolumn{2}{r}{{\bf Error, complexity analysis and final result symbols}} \\
\hline
 $\eta$ & Truncation error vector\\

 $\alpha$ & Block encoding scaling factor for $\widetilde{F}_0$ \\
  $\beta$ & Block encoding scaling factor for $\widetilde{F}_1$  \\
 $ \sigma $ & Error parameter for the block encoding of $\mathcal{A}^{-1}$ \\
 $\tau$ & Error parameter for the block encoding of $(I-\mathcal{M}_1)^{-1}$ \\
 $\mathcal{E}$ & Error in the block encoding of $\mathcal{A}^{-1}$ \\
 $\delta$ & Error in computing the expectation value of the unitary operator \\
 $\epsilon$ & Error in the final solution  \\
 $T_{\max} $ & The largest time up to which the truncation error without dissipative conditions is valid\\
   $\widetilde{T}_{\max}$ & The largest time up to which the final result without dissipative conditions is valid \\
   
\hline
\end{tabular}
\end{table}

\end{document}